\documentclass[12pt,english]{article}
\usepackage[T1]{fontenc}
\usepackage[utf8]{inputenc}
\usepackage[active]{srcltx}
\usepackage{color}
\definecolor{lyxnotefontcolor}{rgb}{0.0195312, 0.203125, 1}
\usepackage{babel}
\usepackage{url}
\usepackage{amsmath}
\usepackage{amsthm}
\usepackage{amssymb}
\usepackage{geometry}
\usepackage{setspace}
\usepackage[authoryear]{natbib}
\usepackage[]
 {hyperref}

\makeatletter
\usepackage{threeparttable}
\usepackage{booktabs}
\usepackage{setspace}
\usepackage{url}
\makeatother

\theoremstyle{plain}
\newtheorem{assumption}{\protect\assumptionname}
\newtheorem{thm}{\protect\theoremname}[section]
\newtheorem{lem}{\protect\lemmaname}[section]
\theoremstyle{remark}
\newtheorem{rem}{\protect\remarkname}[section]
\theoremstyle{plain}
\newtheorem{prop}{\protect\propositionname}[section]
\providecommand{\assumptionname}{Assumption}
\providecommand{\lemmaname}{Lemma}
\providecommand{\propositionname}{Proposition}
\providecommand{\remarkname}{Remark}
\providecommand{\theoremname}{Theorem}

\begin{document}
\title{Estimation and Inference for Peer Effects under Conditional Random Assignment\thanks{Companion software is available as an R package at \protect\url{https://github.com/zengying17/peergmm-r} and as a Stata command at \protect\url{https://github.com/zengying17/peergmm-stata}.}}
\author{Ying Zeng\thanks{School of Economics, Xiamen University, China. Email: \protect\href{mailto:zengying17@gmail.com}{zengying17@gmail.com}.}}

\maketitle
\global\long\def\diag{\operatorname{diag}}%
\global\long\def\E{\operatorname{E}}%
\global\long\def\tr{\operatorname{tr}}%
\global\long\def\Var{\operatorname{Var}}%
\global\long\def\plim{\operatorname{plim}}%
\global\long\def\Diag{\operatorname{Diag}}%
\global\long\def\rank{\operatorname{rank}}%
\global\long\def\Cov{\operatorname{Cov}}%
\global\long\def\argmin{\operatorname{argmin}}%

\begin{abstract}
Empirical studies of peer effects often exploit conditional random assignment to peer groups within urns. We develop a GMM framework for estimation and inference in this setting. The framework separately identifies endogenous and contextual peer effects and nests tests of random peer-group assignment as a special case. It permits unknown heteroskedasticity and corrects finite-urn bias in variance estimation. Its asymptotic theory allows the number of peer groups to grow through more urns, more groups within urns, or both. We establish the asymptotic validity of the procedures and evaluate their finite-sample performance through Monte Carlo simulations. We apply the method to study peer effects on personality among university students. For traits with positive reduced-form peer effects, the estimates indicate that positive contextual effects are partly offset by negative endogenous effects.

Keywords: Peer effects; Conditional random assignment; Randomization tests; Reflection problem; Spatial econometrics; Generalized method of moments.
\end{abstract}
\pagebreak{}

\section{Introduction}

To address the endogeneity of peer-group formation, empirical studies of peer effects often exploit conditional random assignment, whereby individuals are randomly assigned to peer groups within an urn (or selection pool). Examples include college roommates assigned within blocks defined by gender and roommate-matching characteristics \citep{sacerdote_peer_2001}, golfers assigned to competing groups within tournament-category cells \citep{guryan_peer_2009}, students assigned to classrooms within school-grade cells \citep{burke_classroom_2013,graham_identifying_2008}, and students assigned to study groups or sections within classes \citep{booij_ability_2017a,feld_understanding_2016,golsteyn_impact_2021}.\footnote{Recent peer-effects studies exploiting conditional random-assignment designs include \citet{barwick_digital_2026,bietenbeck_motivated_2025,chen_how_2024,hampole_peer_2026,huang_poverty_2025,radbruch_interview_2025,rivera_peers_2025,shan_peers_2025}.}

Despite the widespread use of these designs, two econometric gaps remain. First, existing tests of random peer-group assignment typically examine whether predetermined characteristics exhibit dependence among peers. Such tests can reveal departures from random assignment, but say little about the magnitude of peer dependence or the associated uncertainty. Second, methods for addressing the reflection problem under conditional random assignment remain limited. Beyond the i.i.d. error model of \citet{caeyers_exclusion_2024}, there are few methods for separately identifying the endogenous effects of peers' contemporaneous outcomes and the exogenous (contextual) effects of peers' characteristics. We develop a GMM framework that addresses both gaps.

Our first contribution is to treat the peer-dependence parameter $\lambda_{0}$ as an estimand rather than restricting attention to the null $H_{0}:\lambda_{0}=0$. Under random assignment, predetermined characteristics are uncorrelated among peers conditional on urn fixed effects, so random assignment corresponds to $\lambda_{0}=0$. Existing procedures differ in their implementation, assumptions on the error structure, requirements on urn-size variation, and asymptotic framework \citep{sacerdote_peer_2001,guryan_peer_2009,wang_peer_2010,stevenson_tests_2015,jochmans_testing_2023,caeyers_exclusion_2024}.\footnote{Section \ref{sec:Review of Tests} of the Online Appendix provides a detailed theoretical comparison. It also establishes asymptotic results for several existing procedures and shows that, under the null, the numerator of the homoskedastic test in \citet{jochmans_testing_2023} coincides up to normalization with the quadratic moment underlying our scalar GMM estimator.} Our GMM estimator measures the sign and magnitude of peer dependence and quantifies its uncertainty. It permits unknown heteroskedasticity and does not require variation in either peer-group size or urn size. Our asymptotic theory requires the total number of peer groups to diverge, with growth arising through more urns, more peer groups within urns, or both. This framework includes complete random assignment within a single urn.

Our second contribution is to develop estimation and inference for endogenous and exogenous (contextual) peer effects in the linear-in-means model of under conditional random assignment. Separating these effects is difficult because an individual's outcome and peers' contemporaneous outcomes are jointly determined, giving rise to the reflection problem \citep{manski_identification_1993}. Methods for separately estimating these effects in this setting are limited. \citet{caeyers_exclusion_2024} provide an approach under i.i.d. errors, while empirical studies in this setting often estimate reduced-form effects of peer characteristics instead \citep{burke_classroom_2013,feld_understanding_2016,guryan_peer_2009,shan_peers_2025,rivera_peers_2025}. Such reduced-form effects generally combine endogenous and contextual peer effects and therefore provide limited guidance for policy analysis \citep{fruehwirth_can_2013} and for assessing social multiplier effects \citep{glaeser_social_2003}. In our GMM framework, a quadratic moment implied by independence of the idiosyncratic errors identifies the endogenous peer effect, while linear moments implied by covariate exogeneity identify the contextual peer effects conditional on the endogenous effect. Related approaches to the reflection problem include the conditional QMLE of \citet{lee_identification_2007} with group fixed effects, and the GMM estimator of \citet{kuersteiner_efficient_2023} with random group effects. Relative to these approaches, our method does not require variation in group size and permits unknown individual-level heteroskedasticity.

Methodologically, our GMM framework builds on the linear and quadratic moment conditions introduced by \citet{kelejian_generalized_1998,kelejian_generalized_1999}. More directly, it builds on the GMM estimators for spatial autoregressive (SAR) models under unknown heteroskedasticity developed by \citet{kelejian_specification_2010} and \citet{lin_gmm_2010}.\footnote{\citet{kelejian_specification_2010} also allow for SAR disturbances.} Both papers maintain fixed-dimensional covariate frameworks and therefore do not directly accommodate a growing number of urn fixed effects. We instead eliminate the urn fixed effects by within-urn demeaning and adapt their GMM methods to the peer-group and urn structure. Exploiting this structure, we derive explicit identification conditions, including a covariate-free specification in which the endogenous peer-effect parameter is identified without external instruments. This specification is particularly useful for testing random assignment because the corresponding testing equation often contains no covariates. We also construct a heteroskedasticity-consistent variance estimator that corrects the finite-urn bias induced by within-urn demeaning. Our analysis thus extends spatial-econometric GMM methods to peer-group interactions with a growing number of urn fixed effects.

We complement these theoretical results with Monte Carlo simulations and an empirical application. The simulations examine the finite-sample performance under two growth patterns, an increasing number of urns and an increasing number of peer groups within urns. Bias and size distortions are small in the baseline designs at moderate sample sizes and decline as the total number of peer groups increases. We then apply the framework to the data of \citet{shan_peers_2025}, who study peer effects on personality in university study groups. Five of the six tests fail to reject random peer-group assignment within urns. For the traits with positive reduced-form peer effects, the structural estimates point to positive contextual effects that are partly offset by negative endogenous effects.

The remainder of the paper is organized as follows. Section \ref{sec:Motivating-Example} presents the model, estimator, and asymptotic theory. Section \ref{sec:Monte-Carlo-Simulations} reports Monte Carlo evidence. Section \ref{sec:Empirical-Application} presents an empirical application, and Section \ref{sec:Conclusion} concludes. The Appendix contains the main proofs, while the Online Appendix contains calculations, theoretical results for existing tests, additional simulations, and an additional empirical application.

\section{Model and Estimation}\label{sec:Motivating-Example}

\subsection{Model without Covariates}\label{subsec:Model-without-Covariates}

We first consider a model without covariates. This specification is particularly useful for testing conditional random assignment when the conditioning set consists only of urn fixed effects. In this case, the covariance structure of the innovations identifies $\lambda_{0}$ without instruments constructed from covariates. Identification of endogenous peer effects through the variance-covariance structure has also been studied by \citet{lee_identification_2007}, \citet{graham_identifying_2008}, \citet{kuersteiner_efficient_2023}, and \citet{caeyers_exclusion_2024}, although these approaches do not allow for unknown heteroskedasticity.

Suppose the sample includes $R$ urns, indexed by $r=1,\ldots,R$. Urn $r$ contains $G_{r}$ peer groups, indexed by $g=1,\ldots,G_{r}$. Peer group $g$ in urn $r$ has $m_{gr}$ members, and urn $r$ has $n_{r}=\sum^{G_{r}}_{g=1}m_{gr}$ members. The total sample size is $n=\sum^{R}_{r=1}n_{r}$, and the total number of peer groups is $G=\sum^{R}_{r=1}G_{r}$. Suppose that peers within the same group interact equally and that there is no interaction across peer groups. Then the peer-effects model is
\begin{equation}
Y_{igr}=\alpha_{r}+\lambda_{0}\bar{Y}_{(-i)gr}+\epsilon_{igr},\label{eq:scalar}
\end{equation}
Here $Y_{igr}$ is either an outcome or a characteristic of individual $i$ in peer group $g$ of urn $r$, and $\bar{Y}_{(-i)gr}=\sum_{i'\neq i}Y_{i'gr}/(m_{gr}-1)$ is the leave-one-out mean of $Y$ in the group, i.e., the average of $Y$ among $i$'s peers. The terms $\alpha_{r}$ and $\epsilon_{igr}$ denote the urn fixed effect and the idiosyncratic term. Throughout, the realized urn and peer-group structure and all matrices constructed from it are treated as nonstochastic arrays that may vary with $n$, with dependence on $n$ suppressed in the notation. When $Y$ is a predetermined characteristic, random assignment to peer groups within urns implies that $Y_{igr}$ is uncorrelated with $\bar{Y}_{(-i)gr}$ conditional on the urn fixed effects, and hence $\lambda_{0}=0$. Thus, testing $H_{0}:\lambda_{0}=0$ provides a test of conditional random assignment. When $Y$ is an outcome, the linear-in-means specification can arise as the Nash equilibrium of an interaction game with quadratic utility (see, e.g., \citealp{calvo-armengol_peer_2009,blume_identification_2010,pereda-fernandez_social_2017a}).

The matrix form of model (\ref{eq:scalar}) for group $g$ in urn $r$ is
\begin{eqnarray}
Y_{gr} & = & \alpha_{r}\mathbf{1}_{gr}+\lambda_{0}W_{gr}Y_{gr}+\epsilon_{gr},\label{eq:SAR_group}
\end{eqnarray}
where $Y_{gr}=(Y_{1gr},\ldots,Y_{m_{gr}gr})^{\prime}$, $\mathbf{1}_{gr}$ is an $m_{gr}\times1$ vector of ones, $W_{gr}=(J_{gr}-I_{gr})/(m_{gr}-1)$, $I_{gr}$ is the identity matrix of dimension $m_{gr}$, and $J_{gr}=\mathbf{1}_{gr}\mathbf{1}^{\prime}_{gr}$ is the $m_{gr}\times m_{gr}$ matrix of ones. In spatial econometrics, $W_{gr}$ is a spatial weight matrix and (\ref{eq:SAR_group}) is a Cliff--Ord type spatial model \citep{cliff_spatial_1973,cliff_spatial_1981a}. Stacking model (\ref{eq:SAR_group}) over peer groups gives the urn-level model
\begin{equation}
Y_{r}=\alpha_{r}\mathbf{1}_{r}+\lambda_{0}W_{r}Y_{r}+\epsilon_{r},\label{eq:SAR_urn}
\end{equation}
where $Y_{r}=(Y^{\prime}_{1r},\ldots,Y^{\prime}_{G_{r}r})^{\prime}$, $\epsilon_{r}=(\epsilon^{\prime}_{1r},\ldots,\epsilon^{\prime}_{G_{r}r})^{\prime}$, $W_{r}=\diag^{G_{r}}_{g=1}\{W_{gr}\}$, and $\mathbf{1}_{r}$ is the $n_{r}\times1$ vector of ones.

We next discuss estimation and inference for $\lambda_{0}$ in (\ref{eq:SAR_urn}), which naturally nests the test of $H_{0}:\lambda_{0}=0$. OLS is biased because of the reflection problem arising from the reciprocal relationship between $Y_{igr}$ and $\bar{Y}_{(-i)gr}$. We instead exploit the covariance structure of the innovations to identify and estimate $\lambda_{0}$. Let $\Lambda$ be the parameter space for $\lambda$.
\begin{assumption}
\label{assu:lambda}Suppose $\lambda_{0}$ is in the interior of $\Lambda$, and $\Lambda$ is a compact subset of $(-1,1)$.
\end{assumption}
Urns with only one group and groups with only one member do not contribute to estimation. We therefore restrict attention to the following case.
\begin{assumption}
\label{assu:size}For all $g,r$, $G_{r}\geqslant2$ and $m_{gr}\geqslant2$. The average peer-group size satisfies $n/G\leqslant C_{m}<\infty$ for some positive constant $C_{m}$ that does not depend on $n$.
\end{assumption}
Assumption \ref{assu:size} implies $2\leqslant n/G\leqslant C_{m}$, so the total sample size $n$ tends to infinity if and only if the total number of peer groups $G$ tends to infinity. We impose no further restrictions on the asymptotic behavior of $G_{r}$ or $n_{r}$. Each may remain bounded or tend to infinity. In particular, $G_{r}$ and $n_{r}$ may diverge for some urns while remaining bounded for others, so the framework accommodates settings with both large and small urns. For example, \citet{sacerdote_peer_2001} has twenty-five nonempty urns, but 99\% of the sample is in the sixteen largest urns. The number of urns $R$ may also remain bounded or tend to infinity. This flexibility contrasts with \citet{kelejian_specification_2010} and \citet{lin_gmm_2010}, whose frameworks keep the $R$ bounded, and with \citet{jochmans_testing_2023}, whose asymptotic theory requires $R$ to tend to infinity. The asymptotic results we establish in the Online Appendix for the procedures of \citet{guryan_peer_2009} and \citet{caeyers_exclusion_2024} likewise require $R$ to tend to infinity.

We assume that the innovations are independent while allowing for unknown heteroskedasticity.
\begin{assumption}
\label{assu:epsilon}Suppose the innovations $\epsilon_{igr}$ are totally independent, with $\E(\epsilon_{igr})=0$ and $\E(\epsilon^{2}_{igr})=\sigma^{2}_{igr}$, where $0<c_{\sigma}\leqslant\sigma^{2}_{igr}\leqslant C_{\sigma}<\infty$ for some constants $c_{\sigma}$ and $C_{\sigma}$. In addition, $\sup_{n\geqslant1}\sup_{i,g,r}\E|\epsilon_{igr}|^{4+c_{\epsilon}}<\infty$ for some $c_{\epsilon}>0$.
\end{assumption}
The independence assumption is common in the spatial and peer-effects literature and is maintained, for example, by \citet{kelejian_specification_2010}, \citet{lin_gmm_2010} and \citet{jochmans_testing_2023}. It is important for our covariance-based identification strategy. If dependence among the innovations is unrestricted, covariance generated by peer dependence cannot in general be distinguished from covariance already present in the innovations, so the covariance structure alone cannot identify $\lambda_{0}$. 

A possible violation of the independence assumption is a common group-level shock. This concern is less relevant when testing random assignment using variables determined before group assignment, which cannot be affected by post-assignment group-level shocks. When $Y$ is an outcome, however, such shocks may be present. If valid instruments for $\bar{Y}_{(-i)gr}$ are available, however, $\lambda_{0}$ may still be identified from linear IV moments under appropriate exogeneity conditions, as discussed in the next section. When suitable instruments are unavailable, possible alternatives include extending the group-random-effects model of \citet{kuersteiner_efficient_2023} to accommodate urn fixed effects or using the group-fixed-effects approach of \citet{lee_identification_2007}, both of which rely on additional identifying restrictions.

Rather than estimating the urn fixed effects, we remove them by within-urn demeaning. Define $I^{\ast}_{r}=I_{r}-J_{r}/n_{r}$, where $J_{r}=\mathbf{1}_{r}\mathbf{1}^{\prime}_{r}$ is the $n_{r}\times n_{r}$ matrix of ones. Then $I^{\ast}_{r}\mathbf{1}_{r}=0$, and premultiplying a vector by $I^{\ast}_{r}$ subtracts its urn mean. For example, $I^{\ast}_{r}Y_{r}=Y_{r}-\bar{Y}_{r}\mathbf{1}_{r}$, where $\bar{Y}_{r}$ is the urn mean of $Y_{r}$. Premultiplying both sides of (\ref{eq:SAR_urn}) by $I^{\ast}_{r}$ gives the within-urn equation
\begin{equation}
I^{\ast}_{r}Y_{r}=\lambda_{0}I^{\ast}_{r}W_{r}Y_{r}+I^{\ast}_{r}\epsilon_{r}.\label{eq:SAR_urn_star}
\end{equation}

Let $\Omega_{gr}=\E\left(\epsilon_{gr}\epsilon^{\prime}_{gr}\right)$ and $\Omega_{r}=\E\left(\epsilon_{r}\epsilon^{\prime}_{r}\right)$ denote the variance-covariance matrices of group $g$ and urn $r$, respectively. Under Assumption \ref{assu:epsilon}, $\Omega_{r}=\diag^{G_{r}}_{g=1}\{\Omega_{gr}\}$ is diagonal. For any $n_{r}\times n_{r}$ matrix $A_{r}$ such that $A^{\ast}_{r}=I^{\ast}_{r}A_{r}I^{\ast}_{r}$ has zero diagonals, $\E\left[\left(I^{\ast}_{r}\epsilon_{r}\right)^{\prime}A_{r}\left(I^{\ast}_{r}\epsilon_{r}\right)\right]=\tr\left[\left(I^{\ast}_{r}A_{r}I^{\ast}_{r}\right)\Omega_{r}\right]=0$. This construction is closely related to the quadratic moment conditions in \citet{kelejian_specification_2010} and \citet{lin_gmm_2010}. A key distinction is that, after eliminating the urn fixed effects by within-urn demeaning, the zero-diagonal restriction is imposed on the transformed matrix $I^{\ast}_{r}A_{r}I^{\ast}_{r}$ rather than directly on $A_{r}$. Many matrices satisfy this condition. Following \citet{jochmans_testing_2023}, this paper uses 
\[
A_{r}=W_{r}+\frac{1}{(n_{r}-1)}I_{r}.
\]
As shown in Section \ref{subsec:diag_Atilde} of the Online Appendix, $\diag(I^{\ast}_{r}A_{r}I^{\ast}_{r})=0$ for this choice of $A_{r}$. This special form of $A_{r}$ also allows us to construct a heteroskedasticity-consistent covariance estimator that corrects for the finite-urn bias induced by within-urn demeaning.

For any $\lambda\in\Lambda$, define $\epsilon^{\ast}_{r}(\lambda)=I^{\ast}_{r}(I_{r}-\lambda W_{r})Y_{r}$. At the true value, $\epsilon^{\ast}_{r}(\lambda_{0})=I^{\ast}_{r}\epsilon_{r}$, so the preceding zero-diagonal condition yields the quadratic moment condition $\E h^{q}_{r}(\lambda_{0})=0$, where $h^{q}_{r}(\lambda)=\epsilon^{\ast}_{r}(\lambda)^{\prime}A_{r}\epsilon^{\ast}_{r}(\lambda)$. Let $h^{q}_{n}(\lambda)=\frac{1}{n}\sum^{R}_{r=1}h^{q}_{r}(\lambda)=\frac{1}{n}\sum^{R}_{r=1}\epsilon^{\ast}_{r}(\lambda)^{\prime}A_{r}\epsilon^{\ast}_{r}(\lambda)$, and let $Q^{q}_{n}(\lambda)=h^{q}_{n}(\lambda)^{2}$ be the criterion function. The GMM estimator is $\hat{\lambda}=\argmin{}_{\lambda\in\Lambda}Q^{q}_{n}(\lambda)$.\footnote{\begin{singlespace}
Following \citet{kelejian_specification_2010} and \citet{lin_gmm_2010}, one could increase efficiency by adding multiple quadratic moments, for example, $h^{q}_{r}(\lambda)=(\epsilon^{\ast}_{r}(\lambda)^{\prime}A_{1r}\epsilon^{\ast}_{r}(\lambda),\ldots,\epsilon^{\ast}_{r}(\lambda)^{\prime}A_{Lr}\epsilon^{\ast}_{r}(\lambda))^{\prime}$ with $\diag(I^{\ast}_{r}A_{lr}I^{\ast}_{r})=0$ for $l=1,\ldots,L$. Such an extension would require verifying the identifying contribution of each moment and deriving a more complex variance-covariance estimator.
\end{singlespace}
}
\begin{thm}[Consistency and Asymptotic Normality]
\label{thm:Consistency_lambda}Suppose Assumptions \ref{assu:lambda} to \ref{assu:epsilon} hold. Suppose further that $\E h^{q}_{n}(\lambda)\rightarrow\bar{h}^{q}(\lambda)$ for every $\lambda\in\Lambda$. Then

(i) $\hat{\lambda}\xrightarrow{p}\lambda_{0}$ as $n\rightarrow\infty$.

(ii) Let $A^{\ast}_{r}=I^{\ast}_{r}A_{r}I^{\ast}_{r}$, 
\begin{align*}
V_{\lambda,n} & =\frac{2}{n}\sum^{R}_{r=1}\tr\left(A^{\ast}_{r}\Omega_{r}A^{\ast}_{r}\Omega_{r}\right),\\
D_{\lambda,n} & =-\frac{2}{n}\sum^{R}_{r=1}\tr\left(A^{\ast}_{r}W_{r}(I_{r}-\lambda_{0}W_{r})^{-1}\Omega_{r}\right),
\end{align*}
and suppose $\lim_{n\rightarrow\infty}V_{\lambda,n}=\bar{V}_{\lambda}$, $\lim_{n\rightarrow\infty}D_{\lambda,n}=\bar{D}_{\lambda}$. Then $\bar{V}_{\lambda}>0$, $\bar{D}_{\lambda}<0$, and $\sqrt{n}\left(\hat{\lambda}-\lambda_{0}\right)\xrightarrow{d}N(0,\Sigma_{\lambda})$ as $n\rightarrow\infty$, where $\Sigma_{\lambda}=\bar{D}^{-2}_{\lambda}\bar{V}_{\lambda}$.
\end{thm}
The theorem follows by specializing the corresponding proofs of Theorems \ref{thm:consistency_theta} and \ref{thm:asym_theta} for the general model below to the model without covariates. The theorem provides the basis for testing random assignment through $H_{0}:\lambda_{0}=0$ when $Y$ is predetermined. Section \ref{subsec:VC-estimation} defines estimators $\hat{V}_{\lambda}$ and $\hat{D}_{\lambda}$ of $\bar{V}_{\lambda}$ and $\bar{D}_{\lambda}$, respectively, which are consistent under the conditions in Theorem \ref{thm:convergence_Sigma}. With $\hat{\Sigma}_{\lambda}=\hat{D}^{-2}_{\lambda}\hat{V}_{\lambda}$, $z_{\lambda}=\hat{\lambda}/\sqrt{\hat{\Sigma}_{\lambda}/n}\xrightarrow{d}N(0,1)$ under $H_{0}:\lambda_{0}=0$, yielding a feasible Wald test.

\subsection{General Model}

We now extend the model to include covariates, with the preceding model without covariates as a special case. In scalar form, the model is
\begin{equation}
Y_{igr}=\alpha_{r}+\lambda_{0}\bar{Y}_{(-i)gr}+X^{(1)\prime}_{igr}\beta_{1,0}+\bar{X}^{(1)\prime}_{(-i)gr}\beta_{2,0}+X^{(2)\prime}_{igr}\beta_{3,0}+\epsilon_{igr},\label{eq:general}
\end{equation}
where $X^{(1)}_{igr}$ is a $k_{1}\times1$ vector of exogenous nonstochastic covariates, which may include past outcomes, $\bar{X}^{(1)}_{(-i)gr}=\sum^{m_{gr}}_{i'\neq i}X^{(1)}_{i'gr}/(m_{gr}-1)$ is the corresponding peer average. The $k_{2}\times1$ vector $X^{(2)}_{igr}$ contains exogenous nonstochastic covariates whose peer averages are excluded and may include group-level characteristics. All other variables are defined as in (\ref{eq:scalar}). In the terminology of \citet{manski_identification_1993}, $\lambda_{0}$ represents endogenous peer effects, whereas $\beta_{2,0}$ represents exogenous, or contextual, peer effects.

Let $X_{igr}=(X^{(1)\prime}_{igr},\bar{X}^{(1)\prime}_{(-i)gr},X^{(2)\prime}_{igr})^{\prime}$ denote the $k_{x}\times1$ vector of exogenous covariates and let $\beta_{0}=(\beta^{\prime}_{1,0},\beta^{\prime}_{2,0},\beta^{\prime}_{3,0})^{\prime}$ collect the corresponding coefficients. The model can then be written compactly as
\[
Y_{igr}=\alpha_{r}+\lambda_{0}\bar{Y}_{(-i)gr}+X^{\prime}_{igr}\beta_{0}+\epsilon_{igr}.
\]
Let $X_{r}$ be the $n_{r}\times k_{x}$ matrix of all exogenous covariates. Using the notation in (\ref{eq:SAR_urn}), the matrix form of the general model for urn $r$ is
\begin{equation}
Y_{r}=\alpha_{r}\mathbf{1}_{r}+\lambda_{0}W_{r}Y_{r}+X_{r}\beta_{0}+\epsilon_{r}.\label{eq:before_demean}
\end{equation}
Premultiplying both sides of the equation by $I^{\ast}_{r}$ eliminates the urn fixed effect and gives
\begin{equation}
I^{\ast}_{r}Y_{r}=\lambda_{0}I^{\ast}_{r}W_{r}Y_{r}+I^{\ast}_{r}X_{r}\beta_{0}+I^{\ast}_{r}\epsilon_{r}.\label{eq:general_urn_demeaned}
\end{equation}
Let $\theta=(\lambda,\beta^{\prime})^{\prime}$ be the $\left(k_{x}+1\right)\times1$ parameter vector, and define
\begin{equation}
\epsilon^{\ast}_{r}(\theta)=I^{\ast}_{r}(I_{r}-\lambda W_{r})Y_{r}-I^{\ast}_{r}X_{r}\beta.\label{eq:epsstar_general}
\end{equation}
At the true parameter value, $\epsilon^{\ast}_{r}(\theta_{0})=I^{\ast}_{r}\epsilon_{r}$. The quadratic moment condition from Section \ref{subsec:Model-without-Covariates} therefore remains valid, $\E\left[\epsilon^{\ast}_{r}(\theta_{0})^{\prime}A_{r}\epsilon^{\ast}_{r}(\theta_{0})\right]=0$. In addition, exogeneity of the covariates implies the linear moment condition $\E\left[H^{\prime}_{r}\epsilon^{\ast}_{r}(\theta_{0})\right]=0$, where $H_{r}$ is a nonstochastic $n_{r}\times k_{h}$ IV matrix that contains $I^{\ast}_{r}X_{r}$ and may include additional columns constructed from $I^{\ast}_{r}W_{r}X_{r},I^{\ast}_{r}W^{2}_{r}X_{r},\ldots$, with redundant columns omitted after stacking.\footnote{When $X_{r}=[X^{(1)}_{r},W_{r}X^{(1)}_{r}]$, the distinct candidate blocks in $[I^{\ast}_{r}X_{r},I^{\ast}_{r}W_{r}X_{r}]$ are $I^{\ast}_{r}X^{(1)}_{r}$, $I^{\ast}_{r}W_{r}X^{(1)}_{r}$, and $I^{\ast}_{r}W^{2}_{r}X^{(1)}_{r}$.} We assume that $k_{h}$ and $k_{x}$ are fixed and allow either $k_{x}=k_{h}=0$ or $1\leqslant k_{x}\leqslant k_{h}<\infty$.

The moment function for urn $r$ is $h_{r}(\theta)=\left(\epsilon^{\ast}_{r}(\theta)^{\prime}A_{r}\epsilon^{\ast}_{r}(\theta),\begin{array}{c}
\epsilon^{\ast}_{r}(\theta)^{\prime}H_{r}\end{array}\right)^{\prime}$. The sum of $h_{r}(\theta)$ across urns, normalized by the total sample size, is
\begin{equation}
h_{n}(\theta)=\frac{1}{n}\sum^{R}_{r=1}h_{r}(\theta)=\frac{1}{n}\left(\begin{array}{c}
\sum^{R}_{r=1}\epsilon^{\ast}_{r}(\theta)^{\prime}A_{r}\epsilon^{\ast}_{r}(\theta)\\
\sum^{R}_{r=1}H^{\prime}_{r}\epsilon^{\ast}_{r}(\theta)
\end{array}\right).\label{eq:h_n_theta}
\end{equation}
When the model contains no covariates, $k_{x}=k_{h}=0$, $X_{r}$, $H_{r}$, and $\beta$ are omitted and $\theta=\lambda$. Hence $\epsilon^{\ast}_{r}(\theta)=\epsilon^{\ast}_{r}(\lambda)=I^{\ast}_{r}(I_{r}-\lambda W_{r})Y_{r}$, and the sample moment reduces to the scalar quadratic moment $h^{q}_{n}(\lambda)=\frac{1}{n}\sum^{R}_{r=1}\epsilon^{\ast}_{r}(\lambda)^{\prime}A_{r}\epsilon^{\ast}_{r}(\lambda)$, as in Section \ref{subsec:Model-without-Covariates}. The GMM estimator is defined as $\hat{\theta}=\argmin_{\theta\in\Theta}Q_{n}(\theta)$, where $Q_{n}(\theta)=h_{n}(\theta)^{\prime}\Xi_{n}h_{n}(\theta)$ and $\Xi_{n}$ is a symmetric positive-definite $(k_{h}+1)\times(k_{h}+1)$ weighting matrix. Conditions on $\Xi_{n}$ and specific choices of the weighting matrix are discussed below. When $k_{x}=k_{h}=0$, we set $\Xi_{n}=1$, so the criterion function reduces to $Q^{q}_{n}(\lambda)=h^{q}_{n}(\lambda)^{2}$, as in Section \ref{subsec:Model-without-Covariates}.

\subsection{Asymptotic Properties of the GMM Estimator}

We next establish conditions for consistency and asymptotic normality of $\hat{\theta}$. Let $\Theta$ denote the parameter space of $\theta=(\lambda,\beta^{\prime})^{\prime}$. Assumption \ref{assu:Theta} extends Assumption \ref{assu:lambda} to the model with covariates.
\begin{assumption}
\label{assu:Theta}Let $\Theta=\Lambda\times\mathcal{B}$, where $\Lambda$ is a compact subset of $(-1,1)$ and $\mathcal{B}$ is a compact subset of $\mathbb{R}^{k_{x}}$. The true value $\theta_{0}=(\lambda_{0},\beta_{0}')'$ lies in the interior of $\Theta$. When $k_{x}=0$, $\mathcal{B}=\mathbb{R}^{0}$ and $\Theta=\Lambda$.
\end{assumption}
Define the whole-sample matrices $Y=(Y^{\prime}_{1},\ldots,Y^{\prime}_{R})^{\prime}$, $X=(X^{\prime}_{1},\ldots,X^{\prime}_{R})^{\prime}$, $\epsilon=(\epsilon^{\prime}_{1},\ldots,\epsilon^{\prime}_{R})^{\prime}$, $I^{\ast}=\diag^{R}_{r=1}\{I^{\ast}_{r}\}$, $W=\diag^{R}_{r=1}\{W_{r}\}$, and $H=(H^{\prime}_{1},\ldots,H^{\prime}_{R})^{\prime}$. Stacking (\ref{eq:general_urn_demeaned}) across urns gives
\begin{equation}
I^{\ast}Y=\lambda_{0}I^{\ast}WY+I^{\ast}X\beta_{0}+I^{\ast}\epsilon.\label{eq:general_sample}
\end{equation}
Define $\tilde{X}=\left[W(I-\lambda_{0}W)^{-1}X\beta_{0},X\right]$. It follows that $I^{\ast}\tilde{X}=\E\left(I^{\ast}WY,I^{\ast}X\right)$, which is the expected matrix of right-hand-side variables in (\ref{eq:general_sample}).\footnote{Note that $I-\lambda_{0}W$ is invertible by (\ref{eq:I_lW_inv}), and $I^{\ast}$, $W$ and $(I-\lambda_{0}W)^{-1}$ commute by Remark \ref{rem:calculation}, (\ref{eq:general_sample}) implies $I^{\ast}Y=I^{\ast}(I-\lambda_{0}W)^{-1}\left(X\beta_{0}+\epsilon\right)$, and hence $I^{\ast}WY=I^{\ast}W(I-\lambda_{0}W)^{-1}\left(X\beta_{0}+\epsilon\right)$. Therefore, $\E\left(I^{\ast}WY\right)=I^{\ast}W(I-\lambda_{0}W)^{-1}X\beta_{0}$.} Identification depends on the rank of $I^{\ast}\tilde{X}$.
\begin{assumption}
\label{assu:identification}(i) If $k_{x}>0$, suppose $H$ contains $I^{\ast}X$ and possibly additional instruments. The elements of $I^{\ast}X$ and $H$ are uniformly bounded in absolute value, and the smallest eigenvalues of $X^{\prime}I^{\ast}X/n$ and $H^{\prime}I^{\ast}H/n$ are uniformly bounded below by a positive constant.

(ii) If $k_{x}>0$, $Q_{H\tilde{X}}=\lim_{n\rightarrow\infty}H^{\prime}I^{\ast}\tilde{X}/n$ has the same rank as $Q_{\tilde{X}\tilde{X}}=\lim_{n\rightarrow\infty}\left(\tilde{X}^{\prime}I^{\ast}\tilde{X}\right)/n$.

When $k_{x}=k_{h}=0$, parts (i) and (ii) are not imposed.
\end{assumption}
When $k_{x}>0$, part (i) requires $I^{\ast}X$ in (\ref{eq:general_sample}) to serve as its own instrument. Part (ii) requires $H$ to preserve the rank of $I^{\ast}\tilde{X}$. Since $\tilde{X}$ contains $X$ and part (i) implies $\rank(\lim_{n\rightarrow\infty}X^{\prime}I^{\ast}X/n)=k_{x}$, it follows that $k_{x}\leqslant\rank(Q_{\tilde{X}\tilde{X}})\leqslant k_{x}+1$.

If $\rank(Q_{\tilde{X}\tilde{X}})=k_{x}$, then $I^{\ast}W(I-\lambda_{0}W)^{-1}X\beta_{0}$ is asymptotically linearly dependent on $I^{\ast}X$. One such case is $\beta_{0}=0$, for which $\E(I^{\ast}WY)=0$. Another arises when every peer group has the same size $m$ and $X=[X^{(1)},WX^{(1)}]$, so $W(I-\lambda_{0}W)^{-1}X\beta_{0}$ lies in the span of $X$ and $\rank(Q_{\tilde{X}\tilde{X}})=k_{x}$. Part (ii) gives $\rank(Q_{H\tilde{X}})=k_{x}$ hence $H=I^{\ast}X$ suffices. The quadratic moment identifies $\lambda_{0}$, while the linear moments identify $\beta_{0}$ given $\lambda_{0}$. Variation in peer-group size is therefore unnecessary for identification.

If instead $\rank(Q_{\tilde{X}\tilde{X}})=k_{x}+1$, part (ii) requires $\rank(Q_{H\tilde{X}})=k_{x}+1$, so $H$ must contain at least one additional instrument for $I^{\ast}WY$ beyond $I^{\ast}X$. In this case, $\E\left(I^{\ast}WY\right)=I^{\ast}W(I-\lambda_{0}W)^{-1}X\beta_{0}$ is asymptotically linearly independent of $I^{\ast}X$. Since $I^{\ast}W(I-\lambda_{0}W)^{-1}X\beta_{0}=I^{\ast}WX\beta_{0}+\lambda_{0}I^{\ast}W^{2}X\beta_{0}+\cdots$, columns of $I^{\ast}WX$, $I^{\ast}W^{2}X$, and higher-order spatial lags provide natural candidates for the additional instruments needed to satisfy the rank condition. In this full-rank case, the proof of Theorem \ref{thm:consistency_theta} shows that the linear moment alone identifies $\theta$. Thus, with valid instruments for $I^{\ast}_{r}W_{r}Y_{r}$, the quadratic moment may be dropped and the independence assumption on the idiosyncratic terms in Assumption \ref{assu:epsilon} can be relaxed. We do not pursue this extension here and maintain Assumption \ref{assu:epsilon} throughout the paper.
\begin{assumption}
\label{assu:lim}Suppose $\Xi_{n}\xrightarrow{p}\bar{\Xi}$, where $\bar{\Xi}$ is finite and positive definite, and $\lim_{n\rightarrow\infty}\E h_{n}(\theta)=\bar{h}(\theta)$ for every $\theta\in\Theta$.
\end{assumption}
When $k_{x}\geqslant1$, Assumptions \ref{assu:identification} and \ref{assu:lim} ensure identification through the linear and quadratic moments. When $k_{x}=k_{h}=0$, identification follows from the quadratic moment and Assumption \ref{assu:lim}. As shown in the proof of Theorem \ref{thm:consistency_theta}, $\E h_{n}(\theta)$ is well defined and finite for every $n$ and $\theta\in\Theta$. Thus, with respect to the expected moment, Assumption \ref{assu:lim} only requires the existence of its pointwise limit $\bar{h}(\theta)$. The proof further establish that this convergence is uniform over $\Theta$, $\sup_{\theta\in\Theta}\left\Vert \E h_{n}(\theta)-\bar{h}(\theta)\right\Vert \rightarrow0$.
\begin{thm}[Consistency]
\label{thm:consistency_theta}Suppose Assumptions \ref{assu:size}--\ref{assu:lim} hold. Then $\hat{\theta}\xrightarrow{p}\theta_{0}$ as $n\rightarrow\infty$.
\end{thm}
Appendix \ref{subsec:Proof_Consistency_Theta} contains the proof.

For the asymptotic distribution of $\hat{\theta}$, let $H^{\ast}=I^{\ast}H$, $A^{\ast}=\diag^{R}_{r=1}\{A^{\ast}_{r}\}$, and $\Omega=\diag^{R}_{r=1}\{\Omega_{r}\}$, and define
\begin{align}
V_{n} & =\frac{1}{n}\left[\begin{array}{cc}
2\tr\left(A^{\ast}\Omega A^{\ast}\Omega\right) & 0\\
0 & H^{\ast\prime}\Omega H^{\ast}
\end{array}\right],\label{eq:Vn}\\
D_{n} & =-\frac{1}{n}\left(\begin{array}{c}
\begin{array}{cc}
2\tr\left(A^{\ast}W(I-\lambda_{0}W)^{-1}\Omega\right) & 0\\
H^{\prime}I^{\ast}W(I-\lambda_{0}W)^{-1}X\beta_{0} & H^{\prime}I^{\ast}X
\end{array}\end{array}\right).\label{eq:Dn}
\end{align}
When $k_{x}=0$, the zero-dimensional blocks in (\ref{eq:Vn}) and (\ref{eq:Dn}) are omitted, so $V_{n}=2\tr\left(A^{\ast}\Omega A^{\ast}\Omega\right)/n$ and $D_{n}=-2\tr\left(A^{\ast}W(I-\lambda_{0}W)^{-1}\Omega\right)/n$. As shown in the proof of Theorem \ref{thm:asym_theta}, $V_{n}=\Var\left(\sqrt{n}h_{n}(\theta_{0})\right)$ and $D_{n}=\E\left(\nabla_{\theta}h_{n}(\theta)\left|_{\theta=\theta_{0}}\right.\right)$.
\begin{thm}[Asymptotic Normality]
\label{thm:asym_theta}Suppose Assumptions \ref{assu:size}--\ref{assu:lim} hold and that $\lim_{n\rightarrow\infty}V_{n}=\bar{V}$, $\lim_{n\rightarrow\infty}D_{n}=\bar{D}$. Then $\bar{V}$ is positive definite, $\bar{D}$ has full column rank, and $\sqrt{n}(\hat{\theta}-\theta_{0})\xrightarrow{d}N(0,\Sigma_{\theta})$, where $\Sigma_{\theta}=(\bar{D}^{\prime}\bar{\Xi}\bar{D})^{-1}\bar{D}^{\prime}\bar{\Xi}\bar{V}\bar{\Xi}\bar{D}(\bar{D}^{\prime}\bar{\Xi}\bar{D})^{-1}$.
\end{thm}
The proof is provided in Appendix \ref{subsec:Proof_Normality_Theta}. When $k_{x}=0$, we set $\Xi_{n}=1$. When $k_{x}>0$, the optimal weighting matrix is $\bar{\Xi}=\bar{V}^{-1}$, yielding $\Sigma_{\theta}=(\bar{D}^{\prime}\bar{V}^{-1}\bar{D})^{-1}$. Following \citet{kuersteiner_dynamic_2020}, for the first step of the GMM estimator, we replace $\Omega$ in $V_{n}$ with $I$, define $V^{(1)}_{n}=\frac{1}{n}\diag\left[2\tr(A^{\ast2}),H^{\ast\prime}H^{\ast}\right]$, and use $\Xi^{(1)}_{n}=(V^{(1)}_{n})^{-1}$ as the weighting matrix.\footnote{Provided that $V^{(1)}_{n}\rightarrow\bar{V}^{(1)}$, the proof of Theorem \ref{thm:asym_theta} shows that $\bar{V}^{(1)}$ is positive definite. Hence $\Xi^{(1)}_{n}\xrightarrow{p}\bar{\Xi}^{(1)}=(\bar{V}^{(1)})^{-1}$, where $\bar{\Xi}^{(1)}$ is finite and positive definite, thereby satisfying the weighting-matrix condition in Assumption \ref{assu:lim}.} In the second step, $\Xi^{(2)}_{n}=\hat{V}^{-1}$, where $\hat{V}$ is a consistent estimator of $\bar{V}$ based on the first-step estimates, as discussed below.\footnote{In finite samples, $\hat{V}$ evaluated at the first-step estimate may not be positive definite. In that case, we retain the one-step GMM estimate as the final estimate. Because $\hat{V}\xrightarrow{p}\bar{V}$ and $\bar{V}$ is positive definite, this event has probability approaching zero and does not affect the asymptotic results.}

\subsection{Estimation of the Variance-Covariance Matrix}\label{subsec:VC-estimation}

Feasible estimation of the variance-covariance matrix $\Sigma_{\theta}$ in Theorem \ref{thm:asym_theta} requires estimators of $\bar{V}$ and $\bar{D}$, the limits of $V_{n}$ and $D_{n}$ in (\ref{eq:Vn}) and (\ref{eq:Dn}). The main difficulty is that within-urn demeaning makes the transformed innovations dependent, even though the original innovations are independent under Assumption \ref{assu:epsilon}. We correct for this dependence and obtain a consistent estimator of $\Sigma_{\theta}$ under the same asymptotic sequences considered above, where the total number of peer groups $G$ grows through more urns, more groups within urns, or both. Our approach differs from the urn-clustered estimator of \citet{jochmans_testing_2023}, whose validity requires $R\rightarrow\infty$ and bounded urn sizes. It also differs from the variance estimators of \citet{kelejian_specification_2010} and \citet{lin_gmm_2010}. When applied to our model with urn fixed effects, their fixed-dimensional regressor frameworks require $R$ to remain fixed and urn sizes to grow, so the dependence induced by within-urn demeaning vanishes asymptotically.

Let $\varsigma_{r}=(\sigma^{2}_{1r},\ldots,\sigma^{2}_{n_{r}r})^{\prime}$ so that $\Omega_{r}=\diag(\varsigma_{r})$, and let $\epsilon^{\ast}_{r}=I^{\ast}_{r}\epsilon_{r}=\epsilon_{r}-\bar{\epsilon}_{r}\mathbf{1}_{r}$. When $n_{r}$ is finite, $\E\left(\epsilon^{\ast}_{r}\odot\epsilon^{\ast}_{r}\right)\neq\varsigma_{r}$, where $\odot$ denotes element-wise multiplication. We therefore first define a bias-corrected estimator of $\varsigma_{r}$ by $\tilde{\varsigma}_{r}=P_{r}\left(\epsilon^{\ast}_{r}\odot\epsilon^{\ast}_{r}\right)$, where
\begin{equation}
P_{r}=\frac{n_{r}}{n_{r}-2}\left[I_{r}-\frac{1}{n_{r}(n_{r}-1)}J_{r}\right].\label{eq:P_r}
\end{equation}
The corresponding estimator of $\Omega_{r}$ is $\tilde{\Omega}_{r}=\diag(\tilde{\varsigma}_{r})$. Under Assumption \ref{assu:epsilon}, $\E(\tilde{\varsigma}_{r})=\varsigma_{r}$ and $\E(\tilde{\Omega}_{r})=\Omega_{r}$. To see this, note that the $j$th element of $\epsilon^{\ast}_{r}$ is $\epsilon^{\ast}_{jr}=\epsilon_{jr}-\bar{\epsilon}_{r}$, where $j=1,\ldots,n_{r}$ indexes individuals within urn $r$. By Assumption \ref{assu:epsilon}, 
\begin{align*}
\E(\epsilon^{\ast2}_{jr}) & =\E(\epsilon^{2}_{jr}-2\epsilon_{jr}\bar{\epsilon}_{r}+\bar{\epsilon}^{2}_{r})=(1-\frac{2}{n_{r}})\sigma^{2}_{jr}+\frac{1}{n^{2}_{r}}\sum^{n_{r}}_{j'=1}\sigma^{2}_{j'r},\\
\E(\epsilon^{\ast}_{r}\odot\epsilon^{\ast}_{r}) & =(1-\frac{2}{n_{r}})\varsigma_{r}+\frac{1}{n^{2}_{r}}\mathbf{1}_{r}\mathbf{1}^{\prime}_{r}\varsigma_{r}=\left(\frac{n_{r}-2}{n_{r}}I_{r}+\frac{1}{n^{2}_{r}}J_{r}\right)\varsigma_{r}.
\end{align*}
For $P_{r}$ in (\ref{eq:P_r}), it is readily verified that $P_{r}\left(\frac{n_{r}-2}{n_{r}}I_{r}+\frac{1}{n^{2}_{r}}J_{r}\right)=I_{r}$. Hence $\E(\tilde{\varsigma}_{r})=\E\left[P_{r}\left(\epsilon^{\ast}_{r}\odot\epsilon^{\ast}_{r}\right)\right]=\varsigma_{r}$ and $\E(\tilde{\Omega}_{r})=\Omega_{r}$. Thus, $P_{r}$ corrects the finite-urn bias induced by replacing $\epsilon_{jr}$ with its within-urn demeaned counterpart $\epsilon^{\ast}_{jr}$. This bias vanishes elementwise as $n_{r}\rightarrow\infty$, since $\E(\epsilon^{\ast}_{r}\odot\epsilon^{\ast}_{r})-\varsigma_{r}\rightarrow0$. Consequently, under the large-urn asymptotics of \citet{kelejian_specification_2010} and \citet{lin_gmm_2010}, $\epsilon^{\ast}_{r}\odot\epsilon^{\ast}_{r}$ can be used directly in place of $\varsigma_{r}$ in the variance estimators. By contrast, \citet{jochmans_testing_2023} avoids estimating $\varsigma_{r}$ by using urn-clustered standard errors, whose validity relies on $R\rightarrow\infty$.

Correcting the bias in $\epsilon^{\ast}_{r}\odot\epsilon^{\ast}_{r}$ is not sufficient for the $(1,1)$ element of $V_{n}$. In particular, $\E\left[\tr(A^{\ast}_{r}\tilde{\Omega}_{r}A^{\ast}_{r}\tilde{\Omega}_{r})\right]\neq\tr(A^{\ast}_{r}\Omega_{r}A^{\ast}_{r}\Omega_{r})$ because the elements of $\tilde{\varsigma}_{r}$ are correlated.\footnote{When $R$ is fixed and $n_{r}$ grows, the correlations among $\epsilon^{\ast2}_{jr}$ also vanish, so the variance-covariance estimators of \citet{kelejian_specification_2010} and \citet{lin_gmm_2010} remain valid.} We therefore apply an additional correction to the quadratic component of $V_{n}$. First, observe that for any $n_{r}\times n_{r}$ matrices $M_{1r}$ and $M_{2r}$,
\begin{equation}
\tr(M_{1r}\Omega_{r}M_{2r}\Omega_{r})=\sum^{n_{r}}_{j=1}\sum^{n_{r}}_{j'=1}M_{1r,jj'}\sigma^{2}_{j'r}M_{2r,j'j}\sigma^{2}_{jr}=\varsigma^{\prime}_{r}(M_{1r}\odot M^{\prime}_{2r})\varsigma_{r}.\label{eq:gamma_quadratic}
\end{equation}
Since $A^{\ast}_{r}=A^{\ast\prime}_{r}$, 
\begin{equation}
\tr(A^{\ast}\Omega A^{\ast}\Omega)=\sum^{R}_{r=1}\tr(A^{\ast}_{r}\Omega_{r}A^{\ast}_{r}\Omega_{r})=\sum^{R}_{r=1}\varsigma^{\prime}_{r}\left(A^{\ast}_{r}\odot A^{\ast}_{r}\right)\varsigma_{r}.\label{eq:trAOAO}
\end{equation}
To estimate $\lim_{n\rightarrow\infty}\tr(A^{\ast}\Omega A^{\ast}\Omega)/n$, we use $\sum^{R}_{r=1}\tilde{\varsigma}^{\prime}_{r}A^{\dagger}_{r}\tilde{\varsigma}_{r}/n$, where
\begin{align}
A^{\dagger}_{r} & =\frac{(n_{r}-2)^{2}}{\left[4+(n_{r}-2)^{2}\right]}\left[A^{\ast}_{r}\odot A^{\ast}_{r}+\frac{4}{(n_{r}-2)^{2}(n_{r}-1)+4}T_{r}\right]+t_{r}(J_{r}-I_{r}),\label{eq:A_dagger}\\
T_{r} & =\left(\mathcal{D}_{r}-\frac{\tr(\mathcal{D}_{r})}{4\left(n_{r}-1\right)}I_{r}\right)(J_{r}-I_{r})+(J_{r}-I_{r})\left(\mathcal{D}_{r}-\frac{\tr(\mathcal{D}_{r})}{4\left(n_{r}-1\right)}I_{r}\right),\nonumber \\
t_{r} & =\frac{4(n_{r}-2)(3n_{r}-4)\tr(\mathcal{D}_{r})}{n_{r}(n_{r}-1)(n_{r}-3)\left[4+(n_{r}-2)^{2}\right]\left[(n_{r}-2)^{2}(n_{r}-1)+4\right]},\nonumber 
\end{align}
and $\mathcal{D}_{r}=\diag^{G_{r}}_{g=1}\{\frac{(n_{r}-m_{gr})}{(m_{gr}-1)(n_{r}-1)}I_{gr}\}$. As shown in the proof of Theorem \ref{thm:convergence_Sigma}, this correction satisfies $\E\left(\tilde{\varsigma}^{\prime}_{r}A^{\dagger}_{r}\tilde{\varsigma}_{r}\right)=\varsigma^{\prime}_{r}\left(A^{\ast}_{r}\odot A^{\ast}_{r}\right)\varsigma_{r}$ given the special form of $A_{r}$ we choose.

To obtain feasible estimators, replace $\epsilon^{\ast}_{r}$ with the estimation residual $\hat{\epsilon}^{\ast}_{r}=\epsilon^{\ast}_{r}(\hat{\theta})$, where $\epsilon^{\ast}_{r}(\theta)$ is defined in (\ref{eq:epsstar_general}). Define $\hat{\varsigma}_{r}=P_{r}\left(\hat{\epsilon}^{\ast}_{r}\odot\hat{\epsilon}^{\ast}_{r}\right)$, with $P_{r}$ given in (\ref{eq:P_r}), and let $\hat{\Omega}_{r}=\diag(\hat{\varsigma}_{r})$ and $\hat{\Omega}=\diag^{R}_{r=1}\{\hat{\Omega}_{r}\}$. Finally, let $A^{\dagger}=\diag^{R}_{r=1}\{A^{\dagger}_{r}\}$, where $A^{\dagger}_{r}$ is defined in (\ref{eq:A_dagger}). Our estimators of $\bar{V}$ and $\bar{D}$ are
\begin{align}
\hat{V} & =\frac{1}{n}\left(\begin{array}{cc}
2\hat{\varsigma}^{\prime}A^{\dagger}\hat{\varsigma} & 0\\
0 & H^{\ast\prime}\hat{\Omega}H^{\ast}
\end{array}\right),\label{eq:Vhat}\\
\hat{D} & =-\frac{1}{n}\left(\begin{array}{c}
\begin{array}{cc}
2\tr\left(A^{\ast}W(I-\hat{\lambda}W)^{-1}\hat{\Omega}\right) & 0\\
H^{\prime}I^{\ast}W(I-\hat{\lambda}W)^{-1}X\hat{\beta} & H^{\prime}I^{\ast}X
\end{array}\end{array}\right).\label{eq:Dhat}
\end{align}
When $k_{x}=k_{h}=0$, $\hat{V}=\hat{V}_{\lambda}$ and $\hat{D}=\hat{D}_{\lambda}$ are their scalar $(1,1)$ entries. For the two-step GMM estimator described above, $\hat{V}$ evaluated at the first-step estimates provides the second-step weight $\Xi^{(2)}_{n}=\hat{V}^{-1}$. The variance-covariance matrix is then estimated using $\hat{V}$ and $\hat{D}$ evaluated at the final estimates. For a given weighting matrix $\Xi_{n}$, the resulting estimator of $\Sigma_{\theta}$ is $\hat{\Sigma}_{\theta}=(\hat{D}^{\prime}\Xi_{n}\hat{D})^{-1}\hat{D}^{\prime}\Xi_{n}\hat{V}\Xi_{n}\hat{D}(\hat{D}^{\prime}\Xi_{n}\hat{D})^{-1}$. Under optimal second-step weighting, $\hat{\Sigma}_{\theta}=(\hat{D}^{\prime}\hat{V}^{-1}\hat{D})^{-1}$.
\begin{thm}
\label{thm:convergence_Sigma}Suppose all conditions in Theorem \ref{thm:asym_theta} hold. In addition, suppose $\sup_{n\geqslant1}\sup_{i,g,r}\E\left|\epsilon_{igr}\right|^{8}<\infty$. Then $\hat{D}\xrightarrow{p}\bar{D}$, $\hat{V}\xrightarrow{p}\bar{V}$, and $\hat{\Sigma}_{\theta}\xrightarrow{p}\Sigma_{\theta}$ as $n\rightarrow\infty$.
\end{thm}
Appendix \ref{subsec:Proof_Consistency_Sigma} proves the theorem.

\section{Monte Carlo Simulations}\label{sec:Monte-Carlo-Simulations}

We conduct Monte Carlo simulations to examine the finite-sample properties of our estimators under two growth sequences covered by the asymptotic theory. In the first design, the number of urns is fixed at $R=2$, while the number of peer groups in each urn increases over $G_{r}\in\{10,20,50,100,200,500\}$. In the second design, each urn contains $G_{r}=2$ peer groups, while the number of urns increases over $R\in\{10,20,50,100,200,500\}$. In both designs, half of the peer groups in each urn have size $m_{gr}=2$ and half have size $m_{gr}=4$. Each urn therefore contains $n_{r}=3G_{r}$ observations, and the total sample size is $n=3G$. The growing-$R$ design holds urn size fixed at $n_{r}=6$, whereas the growing-$G_{r}$ design lets urn size diverge.

In the baseline design, data are generated according to
\[
Y_{igr}=\alpha_{r}+\lambda_{0}\bar{Y}_{(-i)gr}+X^{(1)\prime}_{igr}\beta_{1,0}+\bar{X}^{(1)}_{(-i)gr}\beta_{2,0}+\epsilon_{igr},
\]
where $\alpha_{r}$ and $X^{(1)}_{igr}$ are drawn independently from $N(0,1)$. The innovations are drawn independently from $N(0,m_{gr}+1)$, and are thus heteroskedastic in peer-group size. We set $\beta_{1,0}=\beta_{2,0}=1$ and consider $\lambda_{0}\in\{-0.4,0,0.4\}$. In each replication, the urn-level outcome vector is generated as $Y_{r}=(I_{r}-\lambda_{0}W_{r})^{-1}(\alpha_{r}\mathbf{1}_{r}+X^{(1)}_{r}\beta_{1,0}+W_{r}X^{(1)}_{r}\beta_{2,0}+\epsilon_{r})$. We use the two-step GMM procedure introduced after Theorem \ref{thm:asym_theta}. The IV matrix is $H_{r}=[I^{\ast}_{r}W^{2}_{r}X^{(1)}_{r},I^{\ast}_{r}W^{3}_{r}X^{(1)}_{r}]$. Each design uses 5,000 replications.

% Generated by simulations/export_mc_tables.R.
% Do not edit this file by hand.

\begin{table}[!htbp]
\centering
\begin{threeparttable}
\caption{Monte Carlo results with $R=2$ and growing $G_r$}
\label{tab:mc-g}
\setlength{\tabcolsep}{3pt}
\begin{tabular}{llccc@{\hspace{0.5em}}c@{\hspace{0.5em}}ccc@{\hspace{0.5em}}c@{\hspace{0.5em}}ccc}
\hline\hline
 & & \multicolumn{3}{c}{$\lambda_0=-0.4$} & & \multicolumn{3}{c}{$\lambda_0=0$} & & \multicolumn{3}{c}{$\lambda_0=0.4$} \\
\cline{3-5}\cline{7-9}\cline{11-13}
$G$ & Statistic & $\lambda$ & $\beta_1$ & $\beta_2$ &  & $\lambda$ & $\beta_1$ & $\beta_2$ &  & $\lambda$ & $\beta_1$ & $\beta_2$ \\
\hline
20 & Bias & -0.030 & 0.004 & 0.017 &  & -0.024 & 0.010 & 0.013 &  & -0.019 & 0.014 & 0.027 \\
 & MC SD & 0.139 & 0.287 & 0.380 &  & 0.135 & 0.298 & 0.400 &  & 0.103 & 0.304 & 0.419 \\
 & Avg. SE & 0.121 & 0.259 & 0.339 &  & 0.123 & 0.268 & 0.357 &  & 0.093 & 0.277 & 0.375 \\
 & Rej. freq. & 0.111 & 0.086 & 0.093 &  & 0.089 & 0.086 & 0.088 &  & 0.084 & 0.082 & 0.084 \\
\noalign{\smallskip}
40 & Bias & -0.011 & 0.007 & 0.008 &  & -0.012 & 0.008 & 0.010 &  & -0.009 & 0.011 & 0.006 \\
 & MC SD & 0.093 & 0.198 & 0.259 &  & 0.092 & 0.205 & 0.271 &  & 0.069 & 0.213 & 0.288 \\
 & Avg. SE & 0.087 & 0.187 & 0.245 &  & 0.087 & 0.192 & 0.254 &  & 0.065 & 0.198 & 0.265 \\
 & Rej. freq. & 0.073 & 0.068 & 0.075 &  & 0.070 & 0.067 & 0.068 &  & 0.066 & 0.067 & 0.073 \\
\noalign{\smallskip}
100 & Bias & -0.005 & 0.003 & 0.002 &  & -0.004 & 0.001 & 0.003 &  & -0.003 & 0.006 & 0.003 \\
 & MC SD & 0.056 & 0.123 & 0.157 &  & 0.058 & 0.125 & 0.163 &  & 0.042 & 0.131 & 0.172 \\
 & Avg. SE & 0.056 & 0.119 & 0.156 &  & 0.055 & 0.122 & 0.162 &  & 0.041 & 0.126 & 0.169 \\
 & Rej. freq. & 0.054 & 0.057 & 0.055 &  & 0.061 & 0.053 & 0.048 &  & 0.057 & 0.064 & 0.056 \\
\noalign{\smallskip}
200 & Bias & -0.002 & 0.001 & 0.001 &  & -0.003 & 0.001 & 0.001 &  & -0.002 & 0.001 & 0.003 \\
 & MC SD & 0.040 & 0.085 & 0.109 &  & 0.039 & 0.089 & 0.117 &  & 0.029 & 0.089 & 0.122 \\
 & Avg. SE & 0.039 & 0.085 & 0.110 &  & 0.039 & 0.087 & 0.115 &  & 0.029 & 0.090 & 0.119 \\
 & Rej. freq. & 0.054 & 0.053 & 0.052 &  & 0.053 & 0.058 & 0.056 &  & 0.046 & 0.049 & 0.057 \\
\noalign{\smallskip}
400 & Bias & -0.001 & 0.000 & -0.000 &  & -0.001 & 0.003 & 0.001 &  & -0.001 & 0.000 & 0.001 \\
 & MC SD & 0.028 & 0.060 & 0.078 &  & 0.028 & 0.062 & 0.082 &  & 0.021 & 0.063 & 0.087 \\
 & Avg. SE & 0.028 & 0.060 & 0.078 &  & 0.028 & 0.062 & 0.081 &  & 0.021 & 0.063 & 0.085 \\
 & Rej. freq. & 0.050 & 0.049 & 0.052 &  & 0.052 & 0.054 & 0.050 &  & 0.053 & 0.049 & 0.058 \\
\noalign{\smallskip}
1,000 & Bias & -0.000 & 0.000 & -0.000 &  & -0.000 & 0.000 & -0.000 &  & -0.000 & 0.000 & -0.000 \\
 & MC SD & 0.018 & 0.038 & 0.050 &  & 0.018 & 0.039 & 0.052 &  & 0.013 & 0.040 & 0.054 \\
 & Avg. SE & 0.018 & 0.038 & 0.050 &  & 0.018 & 0.039 & 0.051 &  & 0.013 & 0.040 & 0.053 \\
 & Rej. freq. & 0.053 & 0.051 & 0.053 &  & 0.053 & 0.053 & 0.055 &  & 0.051 & 0.053 & 0.053 \\
\noalign{\hrule height 1pt}
\end{tabular}
\begin{tablenotes}[flushleft]
\footnotesize
\item[] \textit{Notes:} The first column reports the total number of peer groups, $G=\sum_r G_r$. For each value of $G$, the four rows report bias, Monte Carlo standard deviation, average estimated standard error, and the rejection frequency of the nominal 5 percent two-sided Wald test of the corresponding true parameter value. Within each urn, peer groups are split evenly between sizes 2 and 4. The design draws $\alpha_r$ and $X^{(1)}_{igr}$ independently from $N(0,1)$, sets $\beta_{1,0}=\beta_{2,0}=1$, and draws independent $\epsilon_{igr}\sim N(0,m_{gr}+1)$. Each design cell has 5,000 replications.
\end{tablenotes}
\end{threeparttable}
\end{table}
% Generated by simulations/export_mc_tables.R.
% Do not edit this file by hand.

\begin{table}[!htbp]
\centering
\begin{threeparttable}
\caption{Monte Carlo results with $G_r=2$ and growing $R$}
\label{tab:mc-r}
\setlength{\tabcolsep}{3pt}
\begin{tabular}{llccc@{\hspace{0.5em}}c@{\hspace{0.5em}}ccc@{\hspace{0.5em}}c@{\hspace{0.5em}}ccc}
\hline\hline
 & & \multicolumn{3}{c}{$\lambda_0=-0.4$} & & \multicolumn{3}{c}{$\lambda_0=0$} & & \multicolumn{3}{c}{$\lambda_0=0.4$} \\
\cline{3-5}\cline{7-9}\cline{11-13}
$G$ & Statistic & $\lambda$ & $\beta_1$ & $\beta_2$ &  & $\lambda$ & $\beta_1$ & $\beta_2$ &  & $\lambda$ & $\beta_1$ & $\beta_2$ \\
\hline
20 & Bias & -0.045 & 0.007 & 0.001 &  & -0.048 & 0.019 & 0.025 &  & -0.033 & 0.018 & 0.032 \\
 & MC SD & 0.161 & 0.340 & 0.463 &  & 0.168 & 0.343 & 0.484 &  & 0.131 & 0.359 & 0.509 \\
 & Avg. SE & 0.129 & 0.282 & 0.376 &  & 0.136 & 0.292 & 0.393 &  & 0.107 & 0.304 & 0.419 \\
 & Rej. freq. & 0.147 & 0.115 & 0.129 &  & 0.138 & 0.105 & 0.134 &  & 0.121 & 0.102 & 0.122 \\
\noalign{\smallskip}
40 & Bias & -0.022 & 0.005 & 0.003 &  & -0.020 & 0.013 & 0.010 &  & -0.016 & 0.013 & 0.010 \\
 & MC SD & 0.105 & 0.226 & 0.304 &  & 0.111 & 0.233 & 0.323 &  & 0.085 & 0.241 & 0.335 \\
 & Avg. SE & 0.096 & 0.207 & 0.278 &  & 0.101 & 0.213 & 0.293 &  & 0.078 & 0.222 & 0.307 \\
 & Rej. freq. & 0.094 & 0.085 & 0.088 &  & 0.094 & 0.082 & 0.091 &  & 0.079 & 0.079 & 0.085 \\
\noalign{\smallskip}
100 & Bias & -0.008 & 0.004 & 0.003 &  & -0.008 & 0.005 & 0.001 &  & -0.007 & 0.005 & 0.005 \\
 & MC SD & 0.065 & 0.139 & 0.190 &  & 0.068 & 0.141 & 0.202 &  & 0.053 & 0.145 & 0.206 \\
 & Avg. SE & 0.062 & 0.134 & 0.183 &  & 0.066 & 0.137 & 0.190 &  & 0.051 & 0.143 & 0.199 \\
 & Rej. freq. & 0.069 & 0.061 & 0.068 &  & 0.063 & 0.062 & 0.068 &  & 0.065 & 0.059 & 0.059 \\
\noalign{\smallskip}
200 & Bias & -0.004 & 0.002 & 0.001 &  & -0.004 & 0.002 & -0.000 &  & -0.003 & 0.003 & 0.004 \\
 & MC SD & 0.045 & 0.099 & 0.134 &  & 0.048 & 0.099 & 0.140 &  & 0.037 & 0.103 & 0.145 \\
 & Avg. SE & 0.045 & 0.095 & 0.131 &  & 0.047 & 0.098 & 0.136 &  & 0.036 & 0.102 & 0.142 \\
 & Rej. freq. & 0.060 & 0.062 & 0.061 &  & 0.057 & 0.056 & 0.061 &  & 0.059 & 0.054 & 0.056 \\
\noalign{\smallskip}
400 & Bias & -0.003 & 0.002 & 0.003 &  & -0.002 & 0.001 & -0.000 &  & -0.002 & 0.002 & 0.001 \\
 & MC SD & 0.033 & 0.068 & 0.095 &  & 0.033 & 0.071 & 0.100 &  & 0.026 & 0.071 & 0.100 \\
 & Avg. SE & 0.032 & 0.068 & 0.093 &  & 0.033 & 0.070 & 0.097 &  & 0.026 & 0.072 & 0.101 \\
 & Rej. freq. & 0.063 & 0.048 & 0.058 &  & 0.054 & 0.055 & 0.054 &  & 0.054 & 0.050 & 0.046 \\
\noalign{\smallskip}
1,000 & Bias & -0.001 & -0.000 & 0.000 &  & -0.001 & 0.000 & -0.002 &  & -0.001 & 0.000 & 0.000 \\
 & MC SD & 0.020 & 0.044 & 0.060 &  & 0.021 & 0.045 & 0.062 &  & 0.016 & 0.046 & 0.064 \\
 & Avg. SE & 0.020 & 0.043 & 0.059 &  & 0.021 & 0.044 & 0.061 &  & 0.016 & 0.046 & 0.064 \\
 & Rej. freq. & 0.051 & 0.052 & 0.055 &  & 0.057 & 0.056 & 0.054 &  & 0.049 & 0.052 & 0.049 \\
\noalign{\hrule height 1pt}
\end{tabular}
\begin{tablenotes}[flushleft]
\footnotesize
\item[] \textit{Notes:} The first column reports the total number of peer groups, $G=\sum_r G_r$. For each value of $G$, the four rows report bias, Monte Carlo standard deviation, average estimated standard error, and the rejection frequency of the nominal 5 percent two-sided Wald test of the corresponding true parameter value. Within each urn, peer groups are split evenly between sizes 2 and 4. The design draws $\alpha_r$ and $X^{(1)}_{igr}$ independently from $N(0,1)$, sets $\beta_{1,0}=\beta_{2,0}=1$, and draws independent $\epsilon_{igr}\sim N(0,m_{gr}+1)$. Each design cell has 5,000 replications.
\end{tablenotes}
\end{threeparttable}
\end{table}

Table \ref{tab:mc-g} reports results for the designs with $R=2$ and growing $G_{r}$, while Table \ref{tab:mc-r} reports results for the designs with $G_{r}=2$ and growing $R$. The first column of each table reports the total number of groups, $G=\sum_{r}G_{r}$. For each value of $G$, four successive rows report bias, Monte Carlo standard deviation, average estimated standard error, and the rejection frequency of the nominal 5 percent two-sided Wald test.

The estimators perform well under both growth sequences, with bias and Monte Carlo standard deviations generally declining as $G$ increases. At $G=100$, the absolute biases of the estimators of $\lambda$, $\beta_{1}$, and $\beta_{2}$ are below 0.01 in both designs. The bias of $\hat{\lambda}$ is somewhat smaller when $G_{r}$ grows than when $R$ grows at smaller values of $G$, while the differences for $\hat{\beta}_{1}$ and $\hat{\beta}_{2}$ are less systematic.

Inference is less accurate at smaller values of $G$. The average estimated standard errors understate the corresponding Monte Carlo standard deviations, particularly when $G_{r}=2$ and $R$ grows, leading to some over-rejection by the Wald tests. For the test of $\lambda=\lambda_{0}$, rejection frequencies across the two designs and the three values of $\lambda_{0}$ range from 8.4 to 14.7 percent at $G=20$. The estimated standard errors become close to the Monte Carlo standard deviations as $G$ increases, and the corresponding rejection frequencies range from 4.6 to 6.0 percent at $G=200$.

The Online Appendix reports additional simulation results for designs without covariates, with non-normal innovations, with constant peer-group sizes and homoskedastic innovations, and with larger peer groups. Across these designs, bias decreases and rejection frequencies approach the nominal level as $G$ increases. At smaller values of $G$, however, inference is more distorted under non-normal innovations, particularly when $G_{r}=2$ and $R$ increases.

\section{Empirical Application}\label{sec:Empirical-Application}

We apply our GMM estimator to the field experiment of \citet{shan_peers_2025}, which studies peer effects in personality development among university students.

The experiment covers six cohorts of students enrolled in an introductory economics course from the 2018/19 through 2023/24 academic years. Within each cohort, students were randomly assigned to four-person study groups within three broad study-program strata. For the 2020/21 cohort, assignment was additionally stratified by the last digit of each student's ID. We treat each cohort-by-stratum cell as an urn. The baseline data contain 1,776 students in 444 study groups and 45 urns. The 30 urns from 2020/21 contain between 2 and 34 students, whereas the 15 urns from the other five cohorts contain between 12 and 283 students. We exclude 23 groups whose members span more than one urn and eight urns with only one remaining group, since such urns do not contribute to estimation. The baseline estimation sample then contains 1,652 students in 413 groups and 35 urns. It combines small and large urns, as permitted by our asymptotic framework, while all peer groups have the same size, illustrating that identification does not require variation in group size. Before study-group assignment, the experiment measured six personality traits---competitiveness, openness, conscientiousness, extraversion, agreeableness, and neuroticism---together with gender, age, major, course-retaking status, and high-school characteristics. These baseline variables are observed for all 1,776 students. At the end of the semester (endline), 1,229 students, or 69 percent of the baseline sample, reported all six personality traits. \citet{shan_peers_2025} provide further details on the experimental design and data.

% Generated by application/peer_cra_SZ2025.do.
\begin{table}[!htbp]
\centering
\begin{threeparttable}
\setlength{\tabcolsep}{3pt}
\def\sym#1{\ifmmode^{#1}\else\(^{#1}\)\fi}
\caption{Randomization test}
\label{tab:sz-randomization}
\begin{tabular}{@{}l*{6}{c}@{}}
\hline\hline
            &\multicolumn{1}{c}{(1)}&\multicolumn{1}{c}{(2)}&\multicolumn{1}{c}{(3)}&\multicolumn{1}{c}{(4)}&\multicolumn{1}{c}{(5)}&\multicolumn{1}{c}{(6)}\\
            &\multicolumn{1}{c}{\shortstack{Competi-\\tiveness}}&\multicolumn{1}{c}{\shortstack{Open-\\ness}}&\multicolumn{1}{c}{\shortstack{Conscien-\\tiousness}}&\multicolumn{1}{c}{\shortstack{Extra-\\version}}&\multicolumn{1}{c}{\shortstack{Agreeable-\\ness}}&\multicolumn{1}{c}{\shortstack{Neuroti-\\cism}}\\
\hline
$\hat{\lambda}$&       0.009         &      -0.024         &      -0.063\sym{**} &      -0.006         &      -0.027         &       0.033         \\
            &     (0.032)         &     (0.032)         &     (0.031)         &     (0.032)         &     (0.030)         &     (0.031)         \\
\hline
Observations&       1,652         &       1,652         &       1,652         &       1,652         &       1,652         &       1,652         \\
Urns        &          35         &          35         &          35         &          35         &          35         &          35         \\
Peer groups &         413         &         413         &         413         &         413         &         413         &         413         \\
\noalign{\hrule height 1pt}
\end{tabular}
\begin{tablenotes}[flushleft]
\footnotesize
\item[] \textit{Notes:} Each column reports the scalar GMM estimate $\hat{\lambda}$ for the indicated baseline personality trait. The sample excludes peer groups spanning multiple urns and urns containing fewer than two remaining peer groups. Standard errors are in parentheses below the estimates. \sym{*} \(p<0.10\), \sym{**} \(p<0.05\), \sym{***} \(p<0.01\).
\end{tablenotes}
\end{threeparttable}
\end{table}

Table 2 of \citet{shan_peers_2025} reports the randomization test proposed by \citet{guryan_peer_2009} for the six baseline personality traits. We revisit the same randomization implication using the scalar GMM estimator developed in Section \ref{subsec:Model-without-Covariates}. Table \ref{tab:sz-randomization} reports the estimates. For five of the six traits, $\hat{\lambda}$ ranges from $-0.027$ to $0.033$ and is not statistically different from zero at the 5 percent level. For conscientiousness, $\hat{\lambda}=-0.063$, and $H_{0}:\lambda_{0}=0$ is rejected at the 5 percent level. Because the table tests six related outcomes, however, one rejection at this level can arise by chance and does not by itself constitute strong evidence against random assignment.\footnote{\citet{jochmans_testing_2023} provides a multivariate test of random assignment. Our GMM framework could likewise be extended to multiple variables by stacking the variable-specific quadratic moments and estimating their full covariance matrix. We leave the formal development of this extension to future research.}

Table 3 of \citet{shan_peers_2025} estimates the effect of baseline peer personality on endline personality using a reduced-form specification that omits contemporaneous peer personality. As discussed above, the resulting coefficient on baseline peer personality generally combines endogenous and contextual peer effects. Specifically, when all peer groups have the same size $m$, as in this application, the structural model in (\ref{eq:general}) implies that the reduced-form coefficient on $\bar{X}^{(1)}_{(-i)gr}$ is\footnote{To see this, the reduced form corresponding to (\ref{eq:general}), with $X^{(2)}$ omitted, is 
\begin{align*}
Y_{gr} & =(I_{gr}-\lambda_{0}W_{gr})^{-1}\left(\alpha_{r}\mathbf{1}_{gr}+X^{(1)}_{gr}\beta_{1,0}+W_{gr}X^{(1)}_{gr}\beta_{2,0}+\epsilon_{gr}\right)\\
 & =\frac{\alpha_{r}}{1-\lambda_{0}}\mathbf{1}_{gr}+X^{(1)}_{gr}\widetilde{\beta}_{1,0}+W_{gr}X^{(1)}_{gr}\widetilde{\beta}_{2,0}+(I_{gr}-\lambda_{0}W_{gr})^{-1}\epsilon_{gr},
\end{align*}
where $\widetilde{\beta}_{1,0}=\frac{\left[(m-1)-(m-2)\lambda_{0}\right]\beta_{1,0}+\lambda_{0}\beta_{2,0}}{(1-\lambda_{0})(m-1+\lambda_{0})}$, and $\widetilde{\beta}_{2,0}=\frac{(m-1)(\beta_{2,0}+\lambda_{0}\beta_{1,0})}{(1-\lambda_{0})(m-1+\lambda_{0})}$ is the coefficients on the peer mean $W_{gr}X^{(1)}_{gr}$.}
\[
\widetilde{\beta}_{2,0}=\frac{(m-1)(\beta_{2,0}+\lambda_{0}\beta_{1,0})}{(1-\lambda_{0})(m-1+\lambda_{0})}.
\]

To separate the two peer effects, we estimate the structural model for each endline personality trait. Each specification includes the following variables and their peer averages: (i) the six baseline personality traits; (ii) high-school characteristics, including math and language grades, study hours, and an indicator for German as the language of instruction; and (iii) gender, course-retaking status, age fixed effects, and major fixed effects. This rich set of predetermined own and peer controls absorbs systematic variation associated with observed characteristics, making the mutual independence of the remaining idiosyncratic innovations in Assumption \ref{assu:epsilon} more plausible. For comparability with \citet{shan_peers_2025}, we standardize the baseline and endline personality traits.\footnote{Standardizing either the outcome or a covariate also standardizes its peer average. In both cases, $\lambda_{0}$ is unchanged and only the relevant $\beta$ coefficients are rescaled.} Constructing peer-average outcomes requires complete endline data for all members of a group. We therefore retain only groups for which endline personality is observed for every member. After further removing urns that contain only one remaining group, the estimation sample comprises 344 students in 86 groups and 13 urns. % Generated by application/peer_cra_SZ2025.do.
\begin{table}[!htbp]
\centering
\begin{threeparttable}
\setlength{\tabcolsep}{2.5pt}
\def\sym#1{\ifmmode^{#1}\else\(^{#1}\)\fi}
\caption{Peer effects on endline personality}
\label{tab:sz-personality}
\begin{tabular}{@{}l*{6}{c}@{}}
\hline\hline
            &\multicolumn{1}{c}{(1)}&\multicolumn{1}{c}{(2)}&\multicolumn{1}{c}{(3)}&\multicolumn{1}{c}{(4)}&\multicolumn{1}{c}{(5)}&\multicolumn{1}{c}{(6)}\\
            &\multicolumn{1}{c}{\shortstack{Competi-\\tiveness}}&\multicolumn{1}{c}{\shortstack{Open-\\ness}}&\multicolumn{1}{c}{\shortstack{Conscien-\\tiousness}}&\multicolumn{1}{c}{\shortstack{Extra-\\version}}&\multicolumn{1}{c}{\shortstack{Agreeable-\\ness}}&\multicolumn{1}{c}{\shortstack{Neuroti-\\cism}}\\
\hline
$\hat{\lambda}$&      -0.160\sym{**} &      -0.201\sym{**} &      -0.137\sym{*}  &      -0.037         &      -0.076         &      -0.085         \\
            &     (0.072)         &     (0.084)         &     (0.072)         &     (0.075)         &     (0.078)         &     (0.072)         \\
Peer competitiveness&       \textbf{0.210\sym{**}} &       0.128         &       0.044         &      -0.076         &      -0.094         &      -0.054         \\
            &     (0.097)         &     (0.085)         &     (0.070)         &     (0.063)         &     (0.086)         &     (0.078)         \\
Peer openness&      -0.149         &       \textbf{0.236\sym{**}} &       0.104         &       0.073         &       0.167\sym{*}  &      -0.057         \\
            &     (0.092)         &     (0.104)         &     (0.083)         &     (0.083)         &     (0.089)         &     (0.091)         \\
Peer conscientiousness&      -0.112         &      -0.221\sym{**} &       \textbf{0.149}         &       0.158\sym{*}  &      -0.020         &      -0.093         \\
            &     (0.099)         &     (0.093)         &     (0.093)         &     (0.088)         &     (0.105)         &     (0.097)         \\
Peer extraversion&       0.163\sym{**} &      -0.005         &      -0.027         &       \textbf{0.018}         &      -0.040         &      -0.005         \\
            &     (0.083)         &     (0.085)         &     (0.090)         &     (0.098)         &     (0.101)         &     (0.083)         \\
Peer agreeableness&      -0.068         &       0.065         &       0.005         &       0.020         &       \textbf{0.047}         &      -0.037         \\
            &     (0.079)         &     (0.078)         &     (0.080)         &     (0.071)         &     (0.096)         &     (0.080)         \\
Peer neuroticism&       0.129         &       0.172\sym{**} &       0.187\sym{**} &       0.042         &       0.026         &      \textbf{-0.109}         \\
            &     (0.081)         &     (0.085)         &     (0.083)         &     (0.074)         &     (0.096)         &     (0.099)         \\
\hline
Implied reduced-form&       0.098         &       0.081         &       0.055         &      -0.009         &      -0.003         &-0.156\sym{*}         \\
coefficient on peer mean&     (0.076)         &     (0.071)         &     (0.071)         &     (0.078)         &     (0.077)         &     (0.086)         \\
Observations&         344         &         344         &         344         &         344         &         344         &         344         \\
Urns        &          13         &          13         &          13         &          13         &          13         &          13         \\
Peer groups &          86         &          86         &          86         &          86         &          86         &          86         \\
\noalign{\hrule height 1pt}
\end{tabular}
\begin{tablenotes}[flushleft]
\footnotesize
\item[] \textit{Notes:} Each column reports a separate specification for the indicated standardized endline personality trait. The $\hat{\lambda}$ row reports the estimated endogenous peer effect, while the peer-personality rows report estimated exogenous peer effects. Boldface marks the exogenous peer-effect estimate for the baseline trait corresponding to the column's endline outcome. The implied reduced-form rows report the coefficient on the peer mean of that baseline trait and its delta-method standard error. The sample excludes peer groups spanning multiple urns, retains only peer groups in which every member has all six endline personality traits observed, and excludes urns containing fewer than two remaining peer groups. Baseline and endline personality traits are standardized. All specifications include the six baseline personality traits, high-school characteristics, gender, course-retaking status, age fixed effects, and major fixed effects, together with their peer averages. Standard errors are in parentheses below the estimates. \sym{*} \(p<0.10\), \sym{**} \(p<0.05\), \sym{***} \(p<0.01\).
\end{tablenotes}
\end{threeparttable}
\end{table}

Table \ref{tab:sz-personality} reports the estimated endogenous peer effect and the six contextual effects of baseline peer personality. In each column, the outcome is a standardized endline personality trait. The bold entry is the coefficient on the peer average of the baseline trait corresponding to the column's endline outcome, and the bottom rows report the corresponding implied reduced-form coefficient $\widetilde{\beta}_{2,0}$ and its delta-method standard error.

\citet{shan_peers_2025} find positive and statistically significant reduced-form peer effects for competitiveness, openness, and conscientiousness, but not for the other three traits. Our structural estimates show how these reduced-form effects decompose into endogenous and contextual components. For competitiveness and openness, the endogenous-effect estimates are $-0.160$ and $-0.201$, while the corresponding same-trait contextual-effect estimates are $0.210$ and $0.236$. All four estimates are statistically significant at the 5 percent level. For conscientiousness, the endogenous-effect estimate is $-0.137$ and significant at the 10 percent level, whereas the contextual-effect estimate is $0.149$ and not statistically significant. The implied reduced-form coefficients for competitiveness, openness, and conscientiousness are $0.098$, $0.081$, and $0.055$, respectively. These point estimates are close to the corresponding fully controlled estimates of $0.077$, $0.067$, and $0.058$ reported by \citet{shan_peers_2025}, although they are less precisely estimated and not statistically significant in our smaller complete-group sample.

Our structural estimates reveal a pattern that is not apparent from the reduced-form evidence alone. For competitiveness, openness, and conscientiousness, positive same-trait contextual effects are partly offset by negative endogenous effects, yielding smaller positive implied reduced-form coefficients. Under an interaction-game interpretation of the linear-in-means model, this pattern is consistent with peer characteristics acting as complementary environmental inputs and contemporaneous peer outcomes acting as strategic substitutes within the group.

\section{Conclusion}\label{sec:Conclusion}

Existing work under conditional random assignment has largely focused on testing the assignment mechanism, while empirical applications often rely on reduced-form specifications that do not distinguish endogenous from contextual peer effects. We develop a GMM framework that separately identifies these two effects, with tests of random peer-group assignment arising as a special case of inference.

Our approach builds on spatial-econometric GMM methods. It combines a quadratic moment implied by the covariance structure of independent, heteroskedastic innovations with linear instrumental-variable moments. The asymptotic framework allows the total number of peer groups to diverge through an increasing number of urns, an increasing number of groups within urns, or both. We also construct a heteroskedasticity-consistent covariance estimator that corrects the finite-urn bias induced by within-urn demeaning. The Monte Carlo simulations show good finite-sample performance at moderate sample sizes, while an application to personality among university students illustrates how the proposed framework can uncover patterns that are obscured by reduced-form estimates.

The analysis focuses on linear-in-means interactions within peer groups. Extending the framework to more general forms of social-network interaction is a natural direction for future research.

\bibliographystyle{elsarticle-harv}
\bibliography{random_peer}

\newpage{}

\appendix

\part*{Appendix}

\section{Preliminaries}

The calculations and proofs use the following matrix identities, boundedness properties, and limit results.

\subsection{Special Forms of Matrices}

Many matrices used in the paper have the block-diagonal form $M=\diag^{R}_{r=1}\{M_{r}\}$, where 
\[
M_{r}=\diag^{G_{r}}_{g=1}\{p_{gr}I^{\ast}_{gr}+q_{gr}J^{\ast}_{gr}\},
\]
$p_{gr}$ and $q_{gr}$ are scalars, $J^{\ast}_{gr}=J_{gr}/m_{gr}=\mathbf{1}_{gr}\mathbf{1}^{\prime}_{gr}/m_{gr}$ is the projection matrix onto the space spanned by $\mathbf{1}_{gr}$, and $I^{\ast}_{gr}=I_{gr}-J^{\ast}_{gr}$ is the corresponding residual-maker matrix. In particular,
\begin{align}
W_{r} & =\diag^{G_{r}}_{g=1}\{-\frac{1}{m_{gr}-1}I^{\ast}_{gr}+J^{\ast}_{gr}\},\label{eq:W}\\
I_{r}-\lambda W_{r} & =\diag^{G_{r}}_{g=1}\{\frac{m_{gr}-1+\lambda}{m_{gr}-1}I^{\ast}_{gr}+(1-\lambda)J^{\ast}_{gr}\},\label{eq:I_lW}\\
A_{r} & =\diag^{G_{r}}_{g=1}\{-\frac{n_{r}-m_{gr}}{(n_{r}-1)(m_{gr}-1)}I^{\ast}_{gr}+\frac{n_{r}}{n_{r}-1}J^{\ast}_{gr}\}.\label{eq:Ar_pIqJ}
\end{align}

The matrices $I^{\ast}_{gr}$ and $J^{\ast}_{gr}$ are symmetric and idempotent and satisfy $I^{\ast}_{gr}+J^{\ast}_{gr}=I_{gr}$ and $I^{\ast}_{gr}J^{\ast}_{gr}=0$. Therefore, for matrices $M^{(l)}_{r}=\diag^{G_{r}}_{g=1}\{p^{(l)}_{gr}I^{\ast}_{gr}+q^{(l)}_{gr}J^{\ast}_{gr}\}$, $l=1,\ldots,L$, we have
\begin{equation}
\prod^{L}_{l=1}M^{(l)}_{r}=\diag^{G_{r}}_{g=1}\left\{ \left(\prod^{L}_{l=1}p^{(l)}_{gr}\right)I^{\ast}_{gr}+\left(\prod^{L}_{l=1}q^{(l)}_{gr}\right)J^{\ast}_{gr}\right\} .\label{eq:prod_s}
\end{equation}
Thus, matrices of this form commute. If $p_{gr}\neq0$ and $q_{gr}\neq0$, then $M^{-1}_{r}=\diag^{G_{r}}_{g=1}\{\frac{1}{p_{gr}}I^{\ast}_{gr}+\frac{1}{q_{gr}}J^{\ast}_{gr}\}$. Accordingly, for $\lambda\in(-1,1)$ and $m_{gr}\geqslant2$, the matrix $I_{r}-\lambda W_{r}$ in (\ref{eq:I_lW}) is invertible, with
\begin{equation}
(I_{r}-\lambda W_{r})^{-1}=\diag^{G_{r}}_{g=1}\{\frac{m_{gr}-1}{m_{gr}-1+\lambda}I^{\ast}_{gr}+\frac{1}{1-\lambda}J^{\ast}_{gr}\}.\label{eq:I_lW_inv}
\end{equation}
See Appendix B.1 of \citet{kuersteiner_efficient_2023} for a detailed discussion of matrices of this form.

Analogously to $J^{\ast}_{gr}$ and $I^{\ast}_{gr}$, define the urn-level matrices $J_{r}=\mathbf{1}_{r}\mathbf{1}^{\prime}_{r}$, $J^{\ast}_{r}=J_{r}/n_{r}$, and $I^{\ast}_{r}=I_{r}-J^{\ast}_{r}$. The next lemma records two identities used repeatedly below.
\begin{lem}
\label{lem:calculation}(i) Suppose $\check{M}_{r}=\diag^{G_{r}}_{g=1}\{p_{gr}I^{\ast}_{gr}+q_{gr}J^{\ast}_{gr}\}+\diag^{G_{r}}_{g=1}\{\varpi_{gr}I_{gr}\}I^{\ast}_{r}$ for scalars $p_{gr},q_{gr},\varpi_{gr}\in\mathbb{R}$. Then $\check{M}_{r}J_{r}=\diag^{G_{r}}_{g=1}\{q_{gr}I_{gr}\}J_{r}$.

(ii) Suppose $M_{r}=\diag^{G_{r}}_{g=1}\{p_{gr}I^{\ast}_{gr}+q_{gr}J^{\ast}_{gr}\}$, and $q_{gr}=q_{r}$ for all $g$. Then 
\begin{align*}
I^{\ast}_{r}M_{r}I^{\ast}_{r}=I^{\ast}_{r}M_{r} & =M_{r}I^{\ast}_{r}=\diag^{G_{r}}_{g=1}\left\{ \left(p_{gr}-q_{r}\right)I^{\ast}_{gr}\right\} +q_{r}I^{\ast}_{r}.
\end{align*}
\end{lem}
\begin{proof}
(i) For part (i), note that $J^{\ast}_{gr}\mathbf{1}_{gr}=\mathbf{1}_{gr}$ for every $g$ and that $\mathbf{1}_{r}=(\mathbf{1}^{\prime}_{1},\ldots,\mathbf{1}^{\prime}_{G_{r}})^{\prime}$. Therefore, $\diag^{G_{r}}_{g=1}\{J^{\ast}_{gr}\}\mathbf{1}_{r}=((J^{\ast}_{1}\mathbf{1}_{1})^{\prime},\ldots,(J^{\ast}_{G_{r}}\mathbf{1}_{G_{r}})^{\prime})^{\prime}=\mathbf{1}_{r}$, which implies $\diag^{G_{r}}_{g=1}\{J^{\ast}_{gr}\}J_{r}=\left[\diag^{G_{r}}_{g=1}\{J^{\ast}_{gr}\}\mathbf{1}_{r}\right]\mathbf{1}^{\prime}_{r}=J_{r}$. In contrast, $\diag^{G_{r}}_{g=1}\{I^{\ast}_{gr}\}J_{r}=\left[I_{r}-\diag^{G_{r}}_{g=1}\{J^{\ast}_{gr}\}\right]J_{r}=0$. Moreover, $I^{\ast}_{r}J_{r}=0$. Since $\check{M}_{r}=\diag^{G_{r}}_{g=1}\{p_{gr}I_{gr}\}\diag^{G_{r}}_{g=1}\{I^{\ast}_{gr}\}+\diag^{G_{r}}_{g=1}\{q_{gr}I_{gr}\}\diag^{G_{r}}_{g=1}\{J^{\ast}_{gr}\}+\diag^{G_{r}}_{g=1}\{\varpi_{gr}I_{gr}\}I^{\ast}_{r}$, the result follows.

(ii) For part (ii), applying part (i) with $\varpi_{gr}=0$ and $q_{gr}=q_{r}$ gives $M_{r}J^{\ast}_{r}=q_{r}J^{\ast}_{r}=q_{r}(I_{r}-I^{\ast}_{r})$. Hence,
\begin{align*}
M_{r}I^{\ast}_{r} & =M_{r}\left(I_{r}-J^{\ast}_{r}\right)=M_{r}-q_{r}I_{r}+q_{r}I^{\ast}_{r}\\
 & =\diag^{G_{r}}_{g=1}\{p_{gr}I^{\ast}_{gr}+q_{r}J^{\ast}_{gr}\}-\diag^{G_{r}}_{g=1}\{q_{r}I^{\ast}_{gr}+q_{r}J^{\ast}_{gr}\}+q_{r}I^{\ast}_{r}\\
 & =\diag^{G_{r}}_{g=1}\{(p_{gr}-q_{r})I^{\ast}_{gr}\}+q_{r}I^{\ast}_{r}.
\end{align*}
Because $M_{r}$ and $I^{\ast}_{r}$ are symmetric and the expression above is symmetric, we have $I^{\ast}_{r}M_{r}=(M_{r}I^{\ast}_{r})^{\prime}=M_{r}I^{\ast}_{r}$. Finally, since $I^{\ast}_{r}$ is idempotent, $I^{\ast}_{r}M_{r}I^{\ast}_{r}=M_{r}I^{\ast}_{r}I^{\ast}_{r}=M_{r}I^{\ast}_{r}$.
\end{proof}

\begin{rem}
\label{rem:calculation} Equations (\ref{eq:W}), (\ref{eq:I_lW}), (\ref{eq:Ar_pIqJ}), and (\ref{eq:I_lW_inv}) show that $W_{r}$, $I_{r}-\lambda W_{r}$, $(I_{r}-\lambda W_{r})^{-1}$, and $A_{r}$ all have $q_{gr}=q_{r}$. It follows from Lemma \ref{lem:calculation} and (\ref{eq:prod_s}) that $W_{r}$, $I_{r}-\lambda W_{r}$, $(I_{r}-\lambda W_{r})^{-1}$, $A_{r}$, and $I^{\ast}_{r}$ commute. For example, $I^{\ast}_{r}W_{r}=W_{r}I^{\ast}_{r}$ and $I^{\ast}_{r}W_{r}(I_{r}-\lambda_{0}W_{r})^{-1}=W_{r}(I_{r}-\lambda_{0}W_{r})^{-1}I^{\ast}_{r}$.
\end{rem}

\subsection{Uniform Boundedness}

Following \citet{kelejian_asymptotic_2001}, an $n\times n$ matrix $M$ is said to have row and column sums uniformly bounded in absolute value (\textbf{UBRC}) if there exists a constant $0<C_{M}<\infty$ that does not depend on $n$ such that
\begin{equation}
\begin{array}{c}
\max_{1\leqslant j\leqslant n}\sum^{n}_{i=1}|M_{ij}|\leq C_{M},\quad\max_{1\leqslant i\leqslant n}\sum^{n}_{j=1}|M{}_{ij}|\leq C_{M}\end{array}\quad\text{for all }n\in\mathbb{N}.\label{eq:UBRC}
\end{equation}
Let $F$ be an $n\times k_{f}$ matrix, where $k_{f}$ may be fixed or may depend on $n$. For example, $F$ is an $n\times n$ matrix when $k_{f}=n$ and reduces to a vector when $k_{f}=1$. We call $F$ uniformly bounded (\textbf{UB}) if there exists a constant $0<C_{F}<\infty$, which does not depend on $n$, such that $\max_{1\leqslant i\leqslant n,1\leqslant j\leqslant k_{f}}|F_{ij}|\leqslant C_{F}<\infty$.

\begin{rem}
\label{rem:rub1}(i) The definition extends naturally to a collection of $n_{r}\times n_{r}$ matrices $M_{r}$ that may depend on $n$. We call $\{M_{r}\}$ \textbf{UBRC} if
\[
\begin{array}{c}
\max_{r,1\leqslant j\leqslant n_{r}}\sum^{n_{r}}_{i=1}|M_{r,ij}|\leq C_{M},\quad\max_{r,1\leqslant i\leqslant n_{r}}\sum^{n_{r}}_{j=1}|M{}_{r,ij}|\leq C_{M}\end{array}\quad\text{for all }n\in\mathbb{N}.
\]
If $M=\diag^{R}_{r=1}\{M_{r}\}$, then $M$ is UBRC if and only if the collection of matrices $\{M_{r}\}$ is UBRC in this sense.

(ii) For $M=\diag^{R}_{r=1}\{M_{r}\}$, where $M_{r}=\diag^{G_{r}}_{g=1}\{p_{gr}I^{\ast}_{gr}+q_{gr}J^{\ast}_{gr}\}$, if $p_{gr}$ and $q_{gr}$ are uniformly bounded, then $M_{r}$ and hence $M$ are UBRC. Consequently, $I^{\ast}=\diag^{R}_{r=1}\left\{ I^{\ast}_{r}\right\} $, $A=\diag^{R}_{r=1}\{A_{r}\}$, and $W=\diag\{W_{r}\}$ are UBRC. Under Assumption \ref{assu:epsilon}, $\Omega$ is also UBRC. Moreover, when $\Lambda$ is a compact subset of $(-1,1)$, both $I-\lambda W$ in (\ref{eq:I_lW}) and $\left(I-\lambda W\right)^{-1}$ in (\ref{eq:I_lW_inv}) are UBRC uniformly over $\lambda\in\Lambda$. Under Assumption \ref{assu:identification}(i), $I^{\ast}X$ is UB. The fixed-dimensional parameter vectors $\beta_{0}$ and $\theta_{0}$ are also bounded.
\end{rem}
\begin{lem}
\label{lem:spectral} Suppose that the $n\times n$ matrix $M$ is UBRC as defined in (\ref{eq:UBRC}). Then, for any $n\times1$ vectors $\alpha_{1}$ and $\alpha_{2}$, $\left|\alpha^{\prime}_{1}M\alpha_{2}\right|\leqslant C_{M}\|\alpha_{1}\|\|\alpha_{2}\|$.
\end{lem}
\begin{proof}
By the definition of the spectral norm, $\left|\alpha^{\prime}_{1}M\alpha_{2}\right|\leqslant\|M\|_{2}\|\alpha_{1}\|\|\alpha_{2}\|$, where $\|M\|_{2}=\sqrt{\lambda_{max}\left(M^{\prime}M\right)}$. A standard matrix-norm inequality gives $\|M\|_{2}\leqslant\sqrt{\|M\|_{1}\|M\|_{\infty}}$, where\footnote{See, e.g., Section 5.6 of \citet{horn_matrix_2012a} for a discussion of matrix norms.}
\[
\|M\|_{1}=\max_{1\leqslant j\leqslant n}\sum^{n}_{i=1}|M_{ij}|\leq C_{M},\;\|M\|_{\infty}=\max_{1\leqslant i\leqslant n}\sum^{n}_{j=1}|M{}_{ij}|\leq C_{M}.
\]
Therefore, $\|M\|_{2}\leqslant C_{M}$, and the result follows.
\end{proof}

The following lemma follows from an argument analogous to that in Footnote 20 of \citet{kelejian_generalized_1999}. The proof is omitted.
\begin{lem}
\label{lem:rub}(i) If $M_{1}$ and $M_{2}$ are both UBRC, then $M_{1}M_{2}$ is UBRC.

(ii) If $M$ is UBRC and $F$ is UB, then $MF$ is UB.
\end{lem}
\begin{rem}
\label{rem:rub2} By Lemma \ref{lem:rub}(i), $W(I-\lambda_{0}W)^{-1}$ is UBRC because both $W$ and $(I-\lambda_{0}W)^{-1}$ are UBRC. By part (ii) of the lemma and Remark \ref{rem:calculation}, $I^{\ast}W(I-\lambda_{0}W)^{-1}X=W(I-\lambda_{0}W)^{-1}\left(I^{\ast}X\right)$ is UB. Hence, $I^{\ast}\tilde{X}=\left[I^{\ast}W(I-\lambda_{0}W)^{-1}X\beta_{0},I^{\ast}X\right]$ is UB. It follows that $A^{\ast}\tilde{X}=I^{\ast}AI^{\ast}\tilde{X}$ is UB, and the entries of the $(k_{x}+1)\times(k_{x}+1)$ matrix $\tilde{X}^{\prime}A^{\ast}\tilde{X}/n=\left(I^{\ast}\tilde{X}\right)^{\prime}A\left(I^{\ast}\tilde{X}\right)/n$ are uniformly bounded.
\end{rem}

\subsection{Limit Theorems for Linear Quadratic Forms in $\epsilon$}

The following lemma is adapted from Lemma A.1 of \citet{kelejian_specification_2010}.
\begin{lem}
\label{lem:varS}Suppose Assumption \ref{assu:epsilon} holds. Let $s_{l}=\epsilon^{\prime}M_{l}\epsilon+f^{\prime}_{l}\epsilon$, $l=1,2$, where $M_{l}$ are nonstochastic $n\times n$ matrices and $f_{l}$ are nonstochastic $n\times1$ vectors. Then $\E s_{l}=\tr(M_{l}\Omega)$, $l=1,2,$ and 
\begin{align*}
\Cov(s_{1},s_{2}) & =\tr\left(\Omega M_{1}\Omega(M_{2}+M^{\prime}_{2})\right)+\sum^{n}_{i=1}M_{1,ii}M_{2,ii}\left[\E(\epsilon^{4}_{i})-3\sigma^{4}_{i}\right]\\
 & +\sum^{n}_{i=1}\left(M_{1,ii}f_{2,i}+M_{2,ii}f_{1,i}\right)\E(\epsilon^{3}_{i})+f^{\prime}_{1}\Omega f_{2}.
\end{align*}
\end{lem}
\begin{rem}
\label{rem:vars}(i) If $\diag(M_{1})=0$, then $\E\left(\epsilon^{\prime}M_{1}\epsilon\right)=0$ and
\[
\Cov(\epsilon^{\prime}M_{1}\epsilon,\epsilon^{\prime}M_{2}\epsilon)=\E\left(\epsilon^{\prime}M_{1}\epsilon\epsilon^{\prime}M_{2}\epsilon\right)=\tr\left[\Omega M_{1}\Omega(M_{2}+M^{\prime}_{2})\right].
\]
(ii) If $\diag(M)=0$, then $\Cov(\epsilon^{\prime}M\epsilon,f^{\prime}\epsilon)=0$. For $F=[f_{1},\ldots,f_{k_{f}}]$, $\Cov(\epsilon^{\prime}M\epsilon,F^{\prime}\epsilon)=0$.

(iii) For a nonstochastic $n_{r}\times n_{r}$ matrix $M_{r}$, the lemma implies that $\Var(\epsilon^{\prime}_{r}M_{r}\epsilon_{r})=\tr\left(\Omega_{r}M_{r}\Omega_{r}(M_{r}+M^{\prime}_{r})\right)+\sum^{n_{r}}_{j=1}M^{2}_{r,jj}\left[\E(\epsilon^{4}_{jr})-3\sigma^{4}_{jr}\right]$. Under Assumption \ref{assu:epsilon}, this variance is $O(n_{r})$ uniformly in $r$ and $n$ if $M_{r}$ is UBRC.
\end{rem}
The following lemma is adapted from Lemma A.3 of \citet{lin_gmm_2010}.
\begin{lem}[Convergence of linear and quadratic forms of $\epsilon$]
\label{lem:convergence_ip}(i) Let $s=\epsilon^{\prime}M\epsilon$, where $M$ is a nonstochastic $n\times n$ matrix that is UBRC. Under Assumption \ref{assu:epsilon}, $\left|\frac{1}{n}s-\E\left(\frac{1}{n}s\right)\right|\xrightarrow{p}0$, and $\E(\frac{1}{n}s)$ is bounded uniformly in $n$.

(ii) Let $s=F^{\prime}\epsilon$, where $F$ is a nonstochastic $n\times k_{f}$ matrix that is UB and $k_{f}$ is fixed. Under Assumption \ref{assu:epsilon}, $\frac{1}{n}s\xrightarrow{p}0$.
\end{lem}
\begin{proof}
For part (i), Lemma \ref{lem:varS} gives $\E\left(\frac{1}{n}\epsilon^{\prime}M\epsilon\right)=\frac{1}{n}\tr(M\Omega)$ and
\[
\Var(\frac{1}{n}\epsilon^{\prime}M\epsilon)=\frac{1}{n^{2}}\tr\left(\Omega M\Omega(M+M^{\prime})\right)+\frac{1}{n^{2}}\sum^{n}_{i=1}M^{2}_{ii}\left[\E(\epsilon^{4}_{i})-3\sigma^{4}_{i}\right].
\]
Because $M$ and $\Omega$ are UBRC, $\Omega M\Omega(M+M^{\prime})$ is UBRC by Lemma \ref{lem:rub}, and hence its trace is $O(n)$. Moreover, $M^{2}_{ii}$ is uniformly bounded and, by Assumption \ref{assu:epsilon}, $\E(\epsilon^{4}_{i})-3\sigma^{4}_{i}$ is uniformly bounded. Therefore, there exists a constant $C_{s}$ that does not depend on $n$ such that $\Var(\frac{1}{n}\epsilon^{\prime}M\epsilon)\leqslant\frac{1}{n}C_{s}\rightarrow0$. Chebyshev's inequality then yields $\left|\frac{1}{n}\epsilon^{\prime}M\epsilon-\E\left(\frac{1}{n}\epsilon^{\prime}M\epsilon\right)\right|\xrightarrow{p}0$. Since $M\Omega$ is UBRC, its diagonal elements are uniformly bounded, which also implies that $\E\left(\frac{1}{n}\epsilon^{\prime}M\epsilon\right)=\frac{1}{n}\tr(M\Omega)$ is uniformly bounded in $n$.

(ii) For part (ii), Assumption \ref{assu:epsilon} implies $\E\left(F^{\prime}\epsilon\right)=0$. In addition, $\E\|\frac{1}{n}F^{\prime}\epsilon\|^{2}=\frac{1}{n^{2}}\tr(F^{\prime}\Omega F)=O(1/n)$ because $F$ is UB, $k_{f}$ is fixed, and $\Omega$ is diagonal with uniformly bounded elements. Therefore, $\frac{1}{n}F^{\prime}\epsilon\xrightarrow{p}0$.
\end{proof}

\section{Proofs}

\subsection{Proof of Theorem \ref{thm:consistency_theta}: Consistency}\label{subsec:Proof_Consistency_Theta}
\begin{proof}
We first establish $\sup_{\theta\in\Theta}|Q_{n}(\theta)-\bar{Q}(\theta)|\xrightarrow{p}0$, where $\bar{Q}(\theta)=\bar{h}(\theta)^{\prime}\bar{\Xi}\bar{h}(\theta)$, with $\bar{h}(\theta)=\lim_{n\rightarrow\infty}\E(h_{n}(\theta))$ and $\bar{\Xi}=\plim\Xi_{n}$ as defined in Assumption \ref{assu:lim}. The triangle and Cauchy--Schwarz inequalities give
\begin{align}
 & |Q_{n}(\theta)-\bar{Q}(\theta)|=|h_{n}(\theta)^{\prime}\Xi_{n}h_{n}(\theta)-\bar{h}(\theta)^{\prime}\bar{\Xi}\bar{h}(\theta)|\nonumber \\
= & \left|\bar{h}(\theta)^{\prime}\left(\Xi_{n}-\bar{\Xi}\right)\bar{h}(\theta)+\left(h_{n}(\theta)-\bar{h}(\theta)\right)^{\prime}\Xi_{n}\left(h_{n}(\theta)-\bar{h}(\theta)\right)+2\bar{h}(\theta)^{\prime}\Xi_{n}\left(h_{n}(\theta)-\bar{h}(\theta)\right)\right|\nonumber \\
\leqslant & \left\Vert \bar{h}(\theta)\right\Vert ^{2}\left\Vert \Xi_{n}-\bar{\Xi}\right\Vert +\left\Vert h_{n}(\theta)-\bar{h}(\theta)\right\Vert ^{2}\left\Vert \Xi_{n}\right\Vert +2\left\Vert \bar{h}(\theta)\right\Vert \left\Vert \Xi_{n}\right\Vert \left\Vert h_{n}(\theta)-\bar{h}(\theta)\right\Vert ,\label{eq:consisproof_convergence}
\end{align}
where $\|\cdot\|$ denotes the Euclidean norm for vectors and the spectral norm for matrices.

To establish $\sup_{\theta\in\Theta}\|h_{n}(\theta)-\bar{h}(\theta)\|\xrightarrow{p}0$ for $h_{n}(\theta)$ in (\ref{eq:h_n_theta}), we first derive a convenient representation of $\epsilon^{\ast}(\theta)$, $h_{n}(\theta)$ and $\E h_{n}(\theta)$. Recall from Remark \ref{rem:calculation} that $I^{\ast}_{r}$, $W_{r}$, $(I_{r}-\lambda_{0}W_{r})^{-1}$, and $A_{r}$ are symmetric and commute. Equation (\ref{eq:before_demean}) implies $Y_{r}=(I_{r}-\lambda_{0}W_{r})^{-1}(\alpha_{r}\mathbf{1}_{r}+X_{r}\beta_{0}+\epsilon_{r})$. Because $I^{\ast}_{r}\mathbf{1}_{r}=0$ and the relevant matrices commute, $I^{\ast}_{r}(I_{r}-\lambda W_{r})(I_{r}-\lambda_{0}W_{r})^{-1}\mathbf{1}_{r}=0$. Substituting the expression for $Y_{r}$ into (\ref{eq:epsstar_general}) therefore gives
\begin{align*}
\epsilon^{\ast}_{r}(\theta) & =I^{\ast}_{r}(I_{r}-\lambda W_{r})(I_{r}-\lambda_{0}W_{r})^{-1}(X_{r}\beta_{0}+\epsilon_{r})-I^{\ast}_{r}X_{r}\beta.
\end{align*}
Let $B_{r}=W_{r}(I-\lambda_{0}W_{r})^{-1}$. Since $(I_{r}-\lambda W_{r})(I_{r}-\lambda_{0}W_{r})^{-1}=\left[I_{r}-\lambda_{0}W_{r}-(\lambda-\lambda_{0})W_{r}\right](I_{r}-\lambda_{0}W_{r})^{-1}=I_{r}-(\lambda-\lambda_{0})B_{r}$, we have
\begin{align*}
\epsilon^{\ast}_{r}(\theta) & =I^{\ast}_{r}\left[I_{r}-(\lambda-\lambda_{0})B_{r}\right](X_{r}\beta_{0}+\epsilon_{r})-I^{\ast}_{r}X_{r}\beta\\
 & =I^{\ast}_{r}\left[I_{r}-(\lambda-\lambda_{0})B_{r}\right]\epsilon_{r}-(\lambda-\lambda_{0})I^{\ast}_{r}B_{r}X_{r}\beta_{0}-I^{\ast}_{r}X_{r}(\beta-\beta_{0})\\
 & =I^{\ast}_{r}\left[I_{r}-(\lambda-\lambda_{0})B_{r}\right]\epsilon_{r}-I^{\ast}_{r}[B_{r}X_{r}\beta_{0},X_{r}]\left(\begin{array}{c}
\lambda-\lambda_{0}\\
\beta-\beta_{0}
\end{array}\right).
\end{align*}
Recall that $\tilde{X}_{r}=[B_{r}X_{r}\beta_{0},X_{r}]$, $\tilde{X}=[\tilde{X}^{\prime}_{1},\ldots,\tilde{X}^{\prime}_{R}]^{\prime}$, and let $B=\diag^{R}_{r=1}\{B_{r}\}$. Then
\begin{equation}
\epsilon^{\ast}(\theta)=(\epsilon^{\ast}_{1}(\theta)^{\prime},\ldots,\epsilon^{\ast}_{R}(\theta)^{\prime})^{\prime}=I^{\ast}[I-(\lambda-\lambda_{0})B]\epsilon-I^{\ast}\tilde{X}(\theta-\theta_{0}).\label{eq:epstheta}
\end{equation}
Let $A^{\ast}=\diag^{R}_{r=1}\{A^{\ast}_{r}\}=I^{\ast}AI^{\ast}$. Using (\ref{eq:epstheta}), the components of $h_{n}(\theta)$ can be written as
\begin{align}
\frac{1}{n}\epsilon^{\ast}(\theta)^{\prime}A\epsilon^{\ast}(\theta) & =\frac{1}{n}\epsilon^{\prime}\tilde{A}(\lambda)\epsilon+(\theta-\theta_{0})^{\prime}\left(\frac{1}{n}\tilde{X}^{\prime}A^{\ast}\tilde{X}\right)(\theta-\theta_{0})\nonumber \\
 & -\frac{2}{n}\epsilon^{\prime}[I-(\lambda-\lambda_{0})B]^{\prime}A^{\ast}\tilde{X}(\theta-\theta_{0}),\label{eq:quadratic_theta}\\
\frac{1}{n}H^{\prime}\epsilon^{\ast}(\theta) & =\frac{1}{n}H^{\prime}I^{\ast}[I-(\lambda-\lambda_{0})B]\epsilon-\frac{1}{n}H^{\prime}I^{\ast}\tilde{X}(\theta-\theta_{0}),\label{eq:linear_theta}
\end{align}
where
\begin{equation}
\tilde{A}(\lambda)=[I-(\lambda-\lambda_{0})B]^{\prime}A^{\ast}[I-(\lambda-\lambda_{0})B]=A^{\ast}-2(\lambda-\lambda_{0})BA^{\ast}+(\lambda-\lambda_{0})^{2}B^{2}A^{\ast}.\label{eq:A_tilda}
\end{equation}
Taking expectations gives
\begin{align}
\E h_{n}(\theta) & =\left(\begin{array}{c}
\frac{1}{n}\tr\left(\tilde{A}(\lambda)\Omega\right)+(\theta-\theta_{0})^{\prime}\left(\frac{1}{n}\tilde{X}^{\prime}A^{\ast}\tilde{X}\right)(\theta-\theta_{0})\\
-\frac{1}{n}H^{\prime}I^{\ast}\tilde{X}(\theta-\theta_{0})
\end{array}\right).\label{eq:Ehbar}
\end{align}

We next show that $\sup_{\theta\in\Theta}\|h_{n}(\theta)-\E h_{n}(\theta)\|\xrightarrow{p}0$. For the first stochastic component in (\ref{eq:quadratic_theta}), we have
\begin{align*}
\frac{1}{n}\epsilon^{\prime}\tilde{A}(\lambda)\epsilon & =\left[\frac{1}{n}\epsilon^{\prime}A^{\ast}\epsilon\right]-2\left(\lambda-\lambda_{0}\right)\left[\frac{1}{n}\epsilon^{\prime}BA^{\ast}\epsilon\right]+(\lambda-\lambda_{0})^{2}\left[\frac{1}{n}\epsilon^{\prime}B^{2}A^{\ast}\epsilon\right].
\end{align*}
By Lemma \ref{lem:rub}, $A^{\ast}$, $BA^{\ast}$, and $B^{2}A^{\ast}$ are UBRC. Hence, Lemma \ref{lem:convergence_ip}(i) implies that each bracketed quadratic form converges in probability to its expectation, with the expectations uniformly bounded in $n$. Since $\lambda$ belongs to a compact set, the coefficients multiplying these quadratic forms are uniformly bounded, and hence the convergence is uniform in $\theta$. The second term in (\ref{eq:quadratic_theta}) is nonstochastic and therefore disappears upon centering. For the third term, Remark \ref{rem:rub2} and Lemma \ref{lem:rub} imply that $A^{\ast}\tilde{X}$ and $B^{\prime}A^{\ast}\tilde{X}$ are both UB. Lemma \ref{lem:convergence_ip}(ii), together with the compactness of $\Theta$, therefore implies that this term converges uniformly in probability to zero. Applying the same argument to (\ref{eq:linear_theta}), using that $H$ is UB and $k_{h}$ is fixed, gives uniform convergence of its stochastic component to zero when $k_{x}>0$. Consequently $\sup_{\theta\in\Theta}\left\Vert h_{n}(\theta)-\E h_{n}(\theta)\right\Vert \xrightarrow{p}0.$

It remains to establish uniform convergence of $\E h_{n}(\theta)$. By Assumption \ref{assu:lim}, $\E h_{n}(\theta)\rightarrow\bar{h}(\theta)$ pointwise on $\Theta$. Equation (\ref{eq:Ehbar}) shows that every component of $\E h_{n}(\theta)$ is a polynomial in $\theta$ of degree at most two, with coefficients depending on $n$ but not on $\theta$. Since $\theta_{0}$ is an interior point of $\Theta$ by Assumption \ref{assu:Theta}, pointwise convergence on $\Theta$ implies convergence of the coefficients of these finite-dimensional polynomials. Compactness of $\Theta$ therefore gives $\sup_{\theta\in\Theta}\left\Vert \E h_{n}(\theta)-\bar{h}(\theta)\right\Vert \rightarrow0.$ Combining this result with the preceding stochastic convergence yields $\sup_{\theta\in\Theta}\left\Vert h_{n}(\theta)-\bar{h}(\theta)\right\Vert \xrightarrow{p}0$. As a uniform limit of continuous functions, $\bar{h}(\theta)$ is continuous and bounded on $\Theta$. Moreover, Assumption \ref{assu:lim} implies $\Xi_{n}\xrightarrow{p}\bar{\Xi}$, so that $\left\Vert \Xi_{n}\right\Vert =O_{p}(1)$. Substituting these results into (\ref{eq:consisproof_convergence}) gives $\sup_{\theta\in\Theta}|Q_{n}(\theta)-\bar{Q}(\theta)|\xrightarrow{p}0$.

We next establish identification by showing that $\bar{h}(\theta)=0$ if and only if $\theta=\theta_{0}$. First consider the case $k_{x}=0$, so that $\theta=\lambda$ and $\E h_{n}(\theta)$ becomes $\E h_{n}(\lambda)=\frac{1}{n}\tr\left(\tilde{A}(\lambda)\Omega\right).$ For $A_{r}=\left[W_{r}+I_{r}/\left(n_{r}-1\right)\right]$, the diagonal elements of $\tilde{A}_{r}(\lambda)$ in (\ref{eq:A_tilda}) associated with peer group $g$ are $-(\lambda-\lambda_{0})\frac{n_{r}-m_{gr}}{m_{gr}(n_{r}-1)}\varphi_{gr}(\lambda)$, where\footnote{See Section \ref{subsec:diag_Atilde} of the Online Appendix for details of the calculations.} 
\[
\varphi_{gr}(\lambda)=\left(\frac{1}{1-\lambda_{0}}+\frac{1}{m_{gr}-1+\lambda_{0}}\right)\left(\frac{1-\lambda}{1-\lambda_{0}}+\frac{m_{gr}-1+\lambda}{m_{gr}-1+\lambda_{0}}\right).
\]
Therefore,
\begin{equation}
\frac{1}{n}\tr\left(\tilde{A}(\lambda)\Omega\right)=-(\lambda-\lambda_{0})\left[\frac{1}{n}\sum^{R}_{r=1}\sum^{G_{r}}_{g=1}\frac{n_{r}-m_{gr}}{m_{gr}(n_{r}-1)}\varphi_{gr}(\lambda)\tr(\Omega_{gr})\right].\label{eq:Ebarh_lambda}
\end{equation}
By Assumptions \ref{assu:epsilon} and \ref{assu:Theta}, there exist constants $c_{\varphi}>0$ and $c_{\sigma}>0$ such that $\varphi_{gr}(\lambda)\geqslant c_{\varphi}$ uniformly over $g,r,\lambda$, and $\tr(\Omega_{gr})\geqslant m_{gr}c_{\sigma}$. Hence
\begin{align*}
\frac{1}{n}\sum^{R}_{r=1}\sum^{G_{r}}_{g=1}\frac{n_{r}-m_{gr}}{m_{gr}(n_{r}-1)}\varphi_{gr}(\lambda)\tr(\Omega_{gr}) & \geqslant c_{\sigma}c_{\varphi}\frac{1}{n}\sum^{R}_{r=1}\sum^{G_{r}}_{g=1}\frac{n_{r}-m_{gr}}{n_{r}-1}=c_{\sigma}c_{\varphi}\frac{1}{n}\sum^{R}_{r=1}\frac{n_{r}(G_{r}-1)}{(n_{r}-1)}.
\end{align*}
Since $n_{r}/(n_{r}-1)\geqslant1$ and by Assumption \ref{assu:size}, $G_{r}\geqslant2$ and hence $G_{r}-1\geqslant G_{r}/2$, $\frac{1}{n}\sum^{R}_{r=1}\frac{n_{r}(G_{r}-1)}{(n_{r}-1)}\geqslant\frac{1}{n}\sum^{R}_{r=1}\frac{G_{r}}{2}=\frac{G}{2n}\geqslant\frac{1}{2C_{m}}>0$. It follows that $\frac{1}{n}\left|\tr\left(\tilde{A}(\lambda)\Omega\right)\right|\geqslant\left|\lambda-\lambda_{0}\right|c_{\sigma}c_{\varphi}/(2C_{m})$. Since the limit exists by Assumption \ref{assu:lim}, the preceding bound implies that $\bar{h}(\lambda)=\lim_{n\rightarrow\infty}\frac{1}{n}\tr\left(\tilde{A}(\lambda)\Omega\right)=0$ if and only if $\lambda=\lambda_{0}$.

Now consider the case $k_{x}>0$. Suppose that $\bar{h}(\theta)=\lim_{n\rightarrow\infty}\E h_{n}(\theta)=0$. The second component of (\ref{eq:Ehbar}) implies $\lim_{n\rightarrow\infty}\frac{1}{n}H^{\prime}I^{\ast}\tilde{X}(\theta-\theta_{0})=0$. Because $H$ contains $I^{\ast}X$, it follows that $\lim_{n\rightarrow\infty}\frac{1}{n}X^{\prime}I^{\ast}\tilde{X}(\theta-\theta_{0})=0$. Let $\tilde{X}=[\phi,X]$, where $\phi=W(I-\lambda_{0}W)^{-1}X\beta_{0}$. Then
\[
\lim_{n\rightarrow\infty}\frac{1}{n}X^{\prime}I^{\ast}\tilde{X}(\theta-\theta_{0})=(\lambda-\lambda_{0})Q_{X\phi}+Q_{XX}(\beta-\beta_{0})=0,
\]
where $Q_{X\phi}=\lim_{n\rightarrow\infty}\frac{1}{n}X^{\prime}I^{\ast}\phi$ and $Q_{XX}=\lim_{n\rightarrow\infty}\frac{1}{n}X^{\prime}I^{\ast}X$. The existence of these limits follows from Assumption \ref{assu:identification}(ii), and Assumption \ref{assu:identification}(i) implies that $Q_{XX}$ is positive definite. Therefore, 
\begin{equation}
\beta-\beta_{0}=-(\lambda-\lambda_{0})Q^{-1}_{XX}Q_{X\phi}.\label{eq:linear_relation}
\end{equation}
Thus $\beta_{0}$ is identified once $\lambda_{0}$ is identified. Let $Q_{\phi\phi}=\lim_{n\rightarrow\infty}\phi^{\prime}I^{\ast}\phi/n$. Then
\[
Q_{\tilde{X}\tilde{X}}=\lim_{n\rightarrow\infty}\frac{1}{n}\tilde{X}^{\prime}I^{\ast}\tilde{X}=\left(\begin{array}{cc}
Q_{\phi\phi} & Q^{\prime}_{X\phi}\\
Q_{X\phi} & Q_{XX}
\end{array}\right).
\]
For any $\theta$ satisfying (\ref{eq:linear_relation}),
\begin{align}
(\theta-\theta_{0})^{\prime}Q_{\tilde{X}\tilde{X}}(\theta-\theta_{0}) & =(\lambda-\lambda_{0})^{2}Q_{\phi\phi}+\left(\beta-\beta_{0}\right)^{\prime}Q_{XX}\left(\beta-\beta_{0}\right)+2(\lambda-\lambda_{0})Q^{\prime}_{X\phi}\left(\beta-\beta_{0}\right)\nonumber \\
 & =(\lambda-\lambda_{0})^{2}Q_{\phi\phi}+\left(\lambda-\lambda_{0}\right)^{2}Q^{\prime}_{X\phi}Q^{-1}_{XX}Q_{X\phi}-2(\lambda-\lambda_{0})^{2}Q^{\prime}_{X\phi}Q^{-1}_{XX}Q_{X\phi}\nonumber \\
 & =(\lambda-\lambda_{0})^{2}\left[Q_{\phi\phi}-Q^{\prime}_{X\phi}Q^{-1}_{XX}Q_{X\phi}\right].\label{eq:Qxx_id}
\end{align}
Because $Q_{\tilde{X}\tilde{X}}$ is positive semidefinite and $Q_{XX}$ is positive definite, $Q_{\phi\phi}-Q^{\prime}_{X\phi}Q^{-1}_{XX}Q_{X\phi}\geqslant0$. If $Q_{\phi\phi}-Q^{\prime}_{X\phi}Q^{-1}_{XX}Q_{X\phi}>0$, then $\rank(Q_{\tilde{X}\tilde{X}})=k_{x}+1$. Assumption \ref{assu:identification}(ii) therefore implies $\rank\left(Q_{H\tilde{X}}\right)=k_{x}+1$, so the linear moment in the second component of (\ref{eq:Ehbar}) alone implies $\theta=\theta_{0}$. If instead $Q_{\phi\phi}-Q^{\prime}_{X\phi}Q^{-1}_{XX}Q_{X\phi}=0$, then (\ref{eq:Qxx_id}) implies that $\lim_{n\rightarrow\infty}\frac{1}{n}\left\Vert I^{\ast}\tilde{X}(\theta-\theta_{0})\right\Vert ^{2}=(\theta-\theta_{0})^{\prime}Q_{\tilde{X}\tilde{X}}(\theta-\theta_{0})=0.$ Since $A$ is UBRC, Lemma \ref{lem:spectral} and $A^{\ast}=I^{\ast}AI^{\ast}$ give 
\[
\left|(\theta-\theta_{0})^{\prime}\left(\frac{1}{n}\tilde{X}^{\prime}A^{\ast}\tilde{X}\right)(\theta-\theta_{0})\right|\leqslant C_{A}\frac{1}{n}\left\Vert I^{\ast}\tilde{X}(\theta-\theta_{0})\right\Vert ^{2}
\]
for some constant $0<C_{A}<\infty$. Hence $\lim_{n\rightarrow\infty}(\theta-\theta_{0})^{\prime}\left(\frac{1}{n}\tilde{X}^{\prime}A^{\ast}\tilde{X}\right)(\theta-\theta_{0})=0$. Therefore, the first component of $\bar{h}(\theta)$ reduces to $\lim_{n\rightarrow\infty}\frac{1}{n}\tr\left(\tilde{A}(\lambda)\Omega\right)$, which, by the argument for the case $k_{x}=0$, equals zero if and only if $\lambda=\lambda_{0}$. Equation (\ref{eq:linear_relation}) then gives $\beta=\beta_{0}$.

Combining the two cases, $\bar{h}(\theta)=0$ if and only if $\theta=\theta_{0}$. Since $\bar{h}(\theta)$ is continuous and $\bar{\Xi}$ is positive definite, $\bar{Q}(\theta)$ is continuous and uniquely minimized at $\theta_{0}$. The uniform convergence of $Q_{n}$ to $\bar{Q}$ and compactness of $\Theta$ therefore imply, by Theorem 4.1.1 of \citet{amemiya_advanced_1985}, that $\hat{\theta}\xrightarrow{p}\theta_{0}$.
\end{proof}

\subsection{Proof of Theorem \ref{thm:asym_theta}: Asymptotic Normality}\label{subsec:Proof_Normality_Theta}
\begin{proof}
By Theorem \ref{thm:consistency_theta}, $\hat{\theta}\xrightarrow{p}\theta_{0}$. Since $\theta_{0}$ lies in the interior of $\Theta$, $\hat{\theta}$ is in the interior of $\Theta$ with probability approaching one (w.p.a.1). Hence, the first-order condition satisfies $\nabla_{\theta}h_{n}(\hat{\theta})^{\prime}\Xi_{n}h_{n}(\hat{\theta})=0$ w.p.a.1. By the mean value theorem, $h_{n}(\hat{\theta})=h_{n}(\theta_{0})+\nabla_{\theta}h_{n}(\tilde{\theta})(\hat{\theta}-\theta_{0})$, where, with a slight abuse of notation, each row of $\nabla_{\theta}h_{n}(\tilde{\theta})$ may be evaluated at a possibly different point between $\theta_{0}$ and $\hat{\theta}$. Therefore,
\[
\nabla_{\theta}h_{n}(\hat{\theta})^{\prime}\Xi_{n}\left[h_{n}(\theta_{0})+\nabla_{\theta}h_{n}(\tilde{\theta})(\hat{\theta}-\theta_{0})\right]=0,
\]
so that, whenever the matrix on the left below is nonsingular, 
\begin{align}
\sqrt{n}\left(\hat{\theta}-\theta_{0}\right) & =-\left[\nabla_{\theta}h_{n}(\hat{\theta})^{\prime}\Xi_{n}\nabla_{\theta}h_{n}(\tilde{\theta})\right]^{-1}\nabla_{\theta}h_{n}(\hat{\theta})^{\prime}\Xi_{n}\left(\sqrt{n}h_{n}(\theta_{0})\right).\label{eq:sqrtthetahat}
\end{align}
We establish below that the matrix being inverted is nonsingular w.p.a.1.

First, we show that $\nabla_{\theta}h_{n}(\hat{\theta})\xrightarrow{p}\bar{D}$ and $\nabla_{\theta}h_{n}(\tilde{\theta})\xrightarrow{p}\bar{D}$. We begin by proving that $\sup_{\theta\in\Theta}\left\Vert \nabla_{\theta}h_{n}(\theta)-\E\nabla_{\theta}h_{n}(\theta)\right\Vert \xrightarrow{p}0.$ The representations in (\ref{eq:quadratic_theta}) and (\ref{eq:linear_theta}) imply that every entry of $\nabla_{\theta}h_{n}(\theta)-\E\left(\nabla_{\theta}h_{n}(\theta)\right)$ is a finite sum of centered quadratic and linear forms in $\epsilon$, possibly multiplied by components of $\theta-\theta_{0}$. For example, using (\ref{eq:A_tilda}), 
\begin{align}
\frac{\partial}{\partial\lambda}\left(\frac{1}{n}\epsilon^{\prime}\tilde{A}(\lambda)\epsilon\right) & =\frac{1}{n}\epsilon^{\prime}\frac{\partial\tilde{A}(\lambda)}{\partial\lambda}\epsilon=\frac{1}{n}\epsilon^{\prime}\left[-2BA^{\ast}+2(\lambda-\lambda_{0})B^{2}A^{\ast}\right]\epsilon\nonumber \\
 & =-2\left(\frac{1}{n}\epsilon^{\prime}BA^{\ast}\epsilon\right)+2(\lambda-\lambda_{0})\left(\frac{1}{n}\epsilon^{\prime}B^{2}A^{\ast}\epsilon\right)\label{eq:derivative_lambda}
\end{align}
By Lemma \ref{lem:rub}, $BA^{\ast}$ and $B^{2}A^{\ast}$ are UBRC, so Lemma \ref{lem:convergence_ip}(i) implies that both quadratic forms in (\ref{eq:derivative_lambda}) converge in probability to their expectations. Differentiating the remaining stochastic term in (\ref{eq:quadratic_theta}) gives linear forms in $\epsilon$ whose coefficient vectors are columns of $A^{\ast}\tilde{X}$ and $B^{\prime}A^{\ast}\tilde{X}$. These coefficient matrices are UB by Remark \ref{rem:rub2} and Lemma \ref{lem:rub}. The nonstochastic term in (\ref{eq:quadratic_theta}) disappears after centering. Finally, differentiating the stochastic term in (\ref{eq:linear_theta}) gives the linear form $-H^{\prime}I^{\ast}B\epsilon/n$ for the derivative with respect to $\lambda$ and zero for the derivatives with respect to $\beta$. The coefficient matrix $B^{\prime}I^{\ast}H$ is UB by Lemma \ref{lem:rub}, because $B^{\prime}I^{\ast}=(I^{\ast}B)^{\prime}$ is UBRC and $H$ is UB. Lemma \ref{lem:convergence_ip}(ii) therefore applies to these linear forms. Since $\Theta$ is compact, the components of $\theta-\theta_{0}$ are uniformly bounded over $\Theta$. It follows that $\sup_{\theta\in\Theta}\left\Vert \nabla_{\theta}h_{n}(\theta)-\E\nabla_{\theta}h_{n}(\theta)\right\Vert \xrightarrow{p}0.$

We next verify that the matrix $D_{n}$ defined in (\ref{eq:Dn}) satisfies $D_{n}=\E\nabla_{\theta}h_{n}(\theta)\left|_{\theta=\theta_{0}}\right.=\nabla_{\theta}\E h_{n}(\theta)\left|_{\theta=\theta_{0}}\right.$. The second equality follows because each component of $h_{n}(\theta)$ is a finite sum of polynomial functions of $\theta$ multiplied by integrable random variables. Therefore, expectation can be taken term by term and commutes with differentiation with respect to $\theta$. For $\E h_{n}(\theta)$ in (\ref{eq:Ehbar}), (\ref{eq:A_tilda}) gives 
\begin{align}
\left.\frac{\partial\tr\left(\tilde{A}(\lambda)\Omega\right)}{\partial\lambda}\right|_{\lambda_{0}} & =\tr\left[\left.\left(-2BA^{\ast}+2(\lambda-\lambda_{0})B^{2}A^{\ast}\right)\right|_{\lambda_{0}}\Omega\right]\nonumber \\
 & =-2\tr\left(BA^{\ast}\Omega\right)=-2\tr\left(A^{\ast}B\Omega\right),\label{eq:derivative1}
\end{align}
where the last equality holds because $A^{\ast}B=BA^{\ast}$ by Remark \ref{rem:calculation}. Moreover,
\[
\left.\nabla_{\theta}\left[(\theta-\theta_{0})^{\prime}\left(\frac{1}{n}\tilde{X}^{\prime}A^{\ast}\tilde{X}\right)(\theta-\theta_{0})\right]\right|_{\theta_{0}}=\left.2\left(\frac{1}{n}\tilde{X}^{\prime}A^{\ast}\tilde{X}\right)(\theta-\theta_{0})\right|_{\theta_{0}}=0,
\]
and
\[
\nabla_{\theta}\left[H^{\prime}I^{\ast}\tilde{X}(\theta-\theta_{0})\right]=H^{\prime}I^{\ast}\tilde{X}=[H^{\prime}I^{\ast}BX\beta_{0},H^{\prime}I^{\ast}X].
\]
Thus,
\begin{equation}
\left(\nabla_{\theta}\E h_{n}(\theta)\right)\left|_{\theta=\theta_{0}}\right.=-\frac{1}{n}\left(\begin{array}{cc}
2\tr(A^{\ast}B\Omega) & 0\\
H^{\prime}I^{\ast}BX\beta_{0} & H^{\prime}I^{\ast}X
\end{array}\right)=D_{n}.\label{eq:D_n_proof}
\end{equation}

The proof of Theorem \ref{thm:consistency_theta} shows that $\E h_{n}(\theta)$ is a polynomial in $\theta$ of degree at most two whose coefficients converge. Therefore $\E\nabla_{\theta}h_{n}(\theta)$ is a polynomial in $\theta$ of degree at most one whose coefficients are uniformly bounded for all sufficiently large $n$. Hence, it is uniformly locally Lipschitz in a neighborhood of $\theta_{0}$. Since $\hat{\theta}\xrightarrow{p}\theta_{0}$ and every mean-value point $\tilde{\theta}$ lies between $\hat{\theta}$ and $\theta_{0}$, we have $\E\nabla_{\theta}h_{n}(\hat{\theta})-D_{n}=o_{p}(1)$, $\E\nabla_{\theta}h_{n}(\tilde{\theta})-D_{n}=o_{p}(1)$. Together with the uniform stochastic convergence established above and $D_{n}\rightarrow\bar{D}$, this yields $\nabla_{\theta}h_{n}(\hat{\theta})\xrightarrow{p}\bar{D}$ and $\nabla_{\theta}h_{n}(\tilde{\theta})\xrightarrow{p}\bar{D}$. Because $\Xi_{n}\xrightarrow{p}\bar{\Xi}$, the continuous mapping theorem implies $\nabla_{\theta}h_{n}(\hat{\theta})^{\prime}\Xi_{n}\nabla_{\theta}h_{n}(\tilde{\theta})\xrightarrow{p}\bar{D}^{\prime}\bar{\Xi}\bar{D}$.

We next show that $\bar{D}$ has full column rank. Recall that $D_{n}$ is given in (\ref{eq:D_n_proof}). Differentiating (\ref{eq:Ebarh_lambda}) at $\lambda=\lambda_{0}$ and using (\ref{eq:derivative1}) gives
\[
2\tr\left(A^{\ast}B\Omega\right)=-\frac{\partial\tr(\tilde{A}(\lambda)\Omega)}{\partial\lambda}|_{\lambda_{0}}=\sum^{R}_{r=1}\sum^{G_{r}}_{g=1}\frac{n_{r}-m_{gr}}{m_{gr}(n_{r}-1)}\varphi_{gr}(\lambda_{0})\tr(\Omega_{gr}).
\]
The lower bound established in the proof of Theorem \ref{thm:consistency_theta} therefore implies $\frac{2}{n}\tr(A^{\ast}B\Omega)\geqslant\frac{c_{\sigma}c_{\varphi}}{2C_{m}}>0.$ Since $D_{n}\rightarrow\bar{D}$, the first element of $\bar{D}$ is strictly negative, while all other elements in its first row are zero. If $k_{x}=0$, this immediately implies that $\bar{D}$ has full column rank. Now suppose $k_{x}>0$. Because $H$ contains $I^{\ast}X$, the matrix $H^{\prime}I^{\ast}X/n$ contains $X^{\prime}I^{\ast}X/n$ as a $k_{x}\times k_{x}$ block. Assumption \ref{assu:identification}(ii) ensures that $H^{\prime}I^{\ast}X/n$ has a limit, while Assumption \ref{assu:identification}(i) implies that the smallest eigenvalue of $X^{\prime}I^{\ast}X/n$ is uniformly bounded away from zero. Hence $\lim_{n\rightarrow\infty}n^{-1}H^{\prime}I^{\ast}X$ has full column rank $k_{x}$. Together with the nonzero first element of the first row of $\bar{D}$, this implies that $\bar{D}$ has full column rank $k_{x}+1$. Because $\bar{\Xi}$ is positive definite by Assumption \ref{assu:lim} and $\bar{D}$ has full column rank, $\bar{D}^{\prime}\bar{\Xi}\bar{D}$ is positive definite and hence nonsingular. Since $\nabla_{\theta}h_{n}(\hat{\theta})^{\prime}\Xi_{n}\nabla_{\theta}h_{n}(\tilde{\theta})\xrightarrow{p}\bar{D}^{\prime}\bar{\Xi}\bar{D},$ the matrix inverted in (\ref{eq:sqrtthetahat}) is nonsingular w.p.a.1.

It remains to establish that $\sqrt{n}h_{n}(\theta_{0})\xrightarrow{d}N(0,\bar{V})$. At $\theta_{0}$, 
\[
h_{n}(\theta_{0})=\frac{1}{n}\left(\begin{array}{c}
\epsilon^{\prime}I^{\ast}AI^{\ast}\epsilon\\
H^{\prime}I^{\ast}\epsilon
\end{array}\right)=\frac{1}{n}\left(\begin{array}{c}
\epsilon^{\prime}A^{\ast}\epsilon\\
H^{\ast\prime}\epsilon
\end{array}\right),
\]
where $A^{\ast}=I^{\ast}AI^{\ast}$ is symmetric with zero diagonal and $H^{\ast}=I^{\ast}H$. By Remark \ref{rem:vars}, 
\[
\Var\left(\sqrt{n}h_{n}(\theta_{0})\right)=\frac{1}{n}\left[\begin{array}{cc}
2\tr\left(A^{\ast}\Omega A^{\ast}\Omega\right) & 0\\
0 & H^{\ast\prime}\Omega H^{\ast}
\end{array}\right]=V_{n}.
\]
We first verify that the limiting variance matrix $\bar{V}$ is positive definite. Since $A^{\ast}_{r}$ is symmetric and $\sigma^{2}_{jr}\geqslant c_{\sigma}>0$ under Assumption \ref{assu:epsilon}, (\ref{eq:trAOAO}) gives
\begin{align*}
\tr(A^{\ast}_{r}\Omega_{r}A^{\ast}_{r}\Omega_{r}) & =\sum^{n_{r}}_{j=1}\sum^{n_{r}}_{j'=1}(A^{\ast}_{r,jj'})^{2}\sigma^{2}_{jr}\sigma^{2}_{j'r}\geqslant c^{2}_{\sigma}\sum^{n_{r}}_{j=1}\sum^{n_{r}}_{j'=1}(A^{\ast}_{r,jj'})^{2}=c^{2}_{\sigma}\tr(A^{\ast2}_{r}).
\end{align*}
Equation (\ref{eq:lAAl}) in the Online Appendix gives
\begin{align*}
\tr(A^{\ast2}_{r}) & =\frac{1}{(n_{r}-1)}\sum^{G_{r}}_{g=1}\frac{(n_{r}-m_{gr})m_{gr}}{m_{gr}-1}\geqslant\frac{\sum^{G_{r}}_{g=1}(n_{r}-m_{gr})}{(n_{r}-1)}=\frac{n_{r}(G_{r}-1)}{(n_{r}-1)}\geqslant G_{r}-1\geqslant\frac{G_{r}}{2}.
\end{align*}
Consequently $\frac{1}{n}\tr\left(A^{\ast}\Omega A^{\ast}\Omega\right)\geqslant c^{2}_{\sigma}G/(2n)\geqslant c^{2}_{\sigma}/(2C_{m})>0$ by Assumption \ref{assu:size}. Thus, the limiting variance of the quadratic moment is strictly positive. When $k_{x}>0$, Assumption \ref{assu:identification}(i) implies that, for some $c_{H}>0$, $\lambda_{\min}\left(\frac{1}{n}H^{\ast\prime}H^{\ast}\right)=\lambda_{\min}\left(\frac{1}{n}H^{\prime}I^{\ast}H\right)\geqslant c_{H}.$ Hence, for every nonzero $k_{h}\times1$ vector $\alpha$, $\alpha^{\prime}\left(\frac{1}{n}H^{\ast\prime}\Omega H^{\ast}\right)\alpha=\frac{1}{n}(H^{\ast}\alpha)^{\prime}\Omega(H^{\ast}\alpha)\geqslant c_{\sigma}\frac{1}{n}(H^{\ast}\alpha)^{\prime}(H^{\ast}\alpha)\geqslant c_{\sigma}c_{H}\alpha^{\prime}\alpha$. Together with the block-diagonal form of $V_{n}$, these bounds imply that $\bar{V}$ is positive definite. When $k_{x}=k_{h}=0$, only the strictly positive scalar quadratic-moment variance remains.

Finally, $A^{\ast}$ is UBRC with zero diagonal, $H^{\ast}=I^{\ast}H$ is UB by Lemma \ref{lem:rub}, and Assumption \ref{assu:epsilon} provides the required moment conditions. Therefore, the central limit theorem for vectors of linear and quadratic forms in Theorem A.1 of \citet{kelejian_specification_2010} gives $\sqrt{n}h_{n}(\theta_{0})\xrightarrow{d}N(0,\bar{V})$. Substituting these results into (\ref{eq:sqrtthetahat}) and applying Slutsky's theorem gives $\sqrt{n}(\hat{\theta}-\theta_{0})\xrightarrow{d}N(0,\Sigma_{\theta}),$ where $\Sigma_{\theta}=(\bar{D}^{\prime}\bar{\Xi}\bar{D})^{-1}\bar{D}^{\prime}\bar{\Xi}\bar{V}\bar{\Xi}\bar{D}(\bar{D}^{\prime}\bar{\Xi}\bar{D})^{-1}$.
\end{proof}

\subsection{Proof of Theorem \ref{thm:convergence_Sigma}}\label{subsec:Proof_Consistency_Sigma}

\subsubsection{Preliminary Lemmas}

Before proving the theorem in Section \ref{subsec:Proof-of-VC-estimation}, we establish two approximation results for the variance estimators and then combine them in Proposition \ref{prop:VC_convergence}. Let $\varsigma=(\varsigma^{\prime}_{1},\ldots,\varsigma^{\prime}_{R})^{\prime}$, with $\varsigma_{r}=(\sigma^{2}_{1r},\ldots,\sigma^{2}_{n_{r}r})^{\prime}$, and define $\tilde{\varsigma}=(\tilde{\varsigma}^{\prime}_{1},\ldots,\tilde{\varsigma}^{\prime}_{R})^{\prime}$ and $\hat{\varsigma}=(\hat{\varsigma}^{\prime}_{1},\ldots,\hat{\varsigma}^{\prime}_{R})^{\prime}$ analogously.
\begin{lem}
\label{lem:hat_tilde_ip}Suppose Assumptions \ref{assu:size} to \ref{assu:lim} hold.

(i) If $\alpha$ is a nonstochastic $n\times1$ vector that is UB, then $\alpha^{\prime}\left(\hat{\varsigma}-\tilde{\varsigma}\right)/n\xrightarrow{p}0$.

(ii) If $M$ is a symmetric, nonstochastic $n\times n$ matrix that is UBRC, then $\left(\hat{\varsigma}^{\prime}M\hat{\varsigma}-\tilde{\varsigma}^{\prime}M\tilde{\varsigma}\right)/n\xrightarrow{p}0$.

(iii) $\frac{1}{n}\left\Vert \hat{\varsigma}\right\Vert ^{2}=O_{p}(1)$.
\end{lem}
\begin{proof}
By Theorem \ref{thm:consistency_theta}, $\hat{\theta}\xrightarrow{p}\theta_{0}$. Recall that $\tilde{\varsigma}_{r}=P_{r}\left(\epsilon^{\ast}_{r}\odot\epsilon^{\ast}_{r}\right)$, where $P_{r}$ is defined in (\ref{eq:P_r}). Thus $\tilde{\varsigma}=P\left(\epsilon^{\ast}\odot\epsilon^{\ast}\right)$, and $\hat{\varsigma}-\tilde{\varsigma}=P\left(\hat{\epsilon}^{\ast}\odot\hat{\epsilon}^{\ast}-\epsilon^{\ast}\odot\epsilon^{\ast}\right)$, where $P=\diag^{R}_{r=1}\left\{ P_{r}\right\} $ is UBRC. From (\ref{eq:epstheta}),
\[
\hat{\epsilon}^{\ast}=\epsilon^{\ast}(\hat{\theta})=\epsilon^{\ast}-(\hat{\lambda}-\lambda_{0})I^{\ast}B\epsilon-I^{\ast}\tilde{X}(\hat{\theta}-\theta_{0}),
\]
where $B=W\left(I-\lambda_{0}W\right)^{-1}$ and $\epsilon^{\ast}=I^{\ast}\epsilon$. Define
\begin{align}
\ddot{\epsilon} & =\hat{\epsilon}^{\ast}-\epsilon^{\ast}=-(\hat{\lambda}-\lambda_{0})I^{\ast}B\epsilon-I^{\ast}\tilde{X}(\hat{\theta}-\theta_{0})\label{eq:ddoteps}\\
\Delta_{\epsilon} & =\hat{\epsilon}^{\ast}\odot\hat{\epsilon}^{\ast}-\epsilon^{\ast}\odot\epsilon^{\ast}=\left(\epsilon^{\ast}+\ddot{\epsilon}\right)\odot\left(\epsilon^{\ast}+\ddot{\epsilon}\right)-\epsilon^{\ast}\odot\epsilon^{\ast}=2\ddot{\epsilon}\odot\epsilon^{\ast}+\ddot{\epsilon}\odot\ddot{\epsilon}.\label{eq:Delta_eps}
\end{align}
Then $\hat{\varsigma}-\tilde{\varsigma}=P\Delta_{\epsilon}$. We first establish $\frac{1}{n}\|\Delta_{\epsilon}\|^{2}=o_{p}(1)$.

For any nonstochastic $n\times n$ matrix $M$ that is UBRC, $\sup_{i}\E|\left(M\epsilon\right)_{i}|^{4}<C<\infty$ for some constant $C>0$. Indeed, by the triangle inequality and Hölder's inequality, 
\[
\E\left|\left(M\epsilon\right)_{i}\right|^{4}\leqslant\E\left(\sum^{n}_{k=1}|M_{ik}||\epsilon_{k}|\right)^{4}\leqslant\left(\sum^{n}_{k=1}|M_{ik}|\right)^{3}\sum^{n}_{k=1}|M_{ik}|\E|\epsilon_{k}|^{4}\leqslant C_{\epsilon}\left(\sum^{n}_{k=1}\left|M_{ik}\right|\right)^{4},
\]
which is uniformly bounded because $M$ is UBRC. Since $I^{\ast}$ and $I^{\ast}B$ are UBRC, this implies that $\frac{1}{n}\sum^{n}_{i=1}(\epsilon^{\ast}_{i})^{4}=O_{p}(1)$ and $\frac{1}{n}\sum^{n}_{i=1}(I^{\ast}B\epsilon)^{4}_{i}=O_{p}(1)$. Moreover, $I^{\ast}\tilde{X}$ is UB, so there exists a constant $0<C_{X}<\infty$ such that $\left\Vert \left(I^{\ast}\tilde{X}\right)_{i,.}\right\Vert \leqslant C_{X}$ uniformly over $i$ and $n$. By (\ref{eq:ddoteps}),
\begin{align*}
\left|\ddot{\epsilon}_{i}\right| & =\left|(\hat{\lambda}-\lambda_{0})\left(I^{\ast}B\epsilon\right)_{i}+\left(I^{\ast}\tilde{X}\right)_{i,.}(\hat{\theta}-\theta_{0})\right|\\
 & \leqslant\left|\hat{\lambda}-\lambda_{0}\right|\left|\left(I^{\ast}B\epsilon\right)_{i}\right|+\left\Vert \left(I^{\ast}\tilde{X}\right)_{i,.}\right\Vert \left\Vert \hat{\theta}-\theta_{0}\right\Vert \\
 & \leqslant\left|\hat{\lambda}-\lambda_{0}\right|\left|\left(I^{\ast}B\epsilon\right)_{i}\right|+C_{X}\left\Vert \hat{\theta}-\theta_{0}\right\Vert .
\end{align*}
Using $(a+b)^{4}\leqslant8(a^{4}+b^{4})$ and $\hat{\theta}\xrightarrow{p}\theta_{0}$, we have
\begin{align*}
\frac{1}{n}\sum^{n}_{i=1}\ddot{\epsilon}^{4}_{i} & \leqslant8(\hat{\lambda}-\lambda_{0})^{4}\left[\frac{1}{n}\sum^{n}_{i=1}(I^{\ast}B\epsilon)^{4}_{i}\right]+8C^{4}_{X}\left\Vert \hat{\theta}-\theta_{0}\right\Vert ^{4}=o_{p}(1).
\end{align*}
By (\ref{eq:Delta_eps}) and $(a+b)^{2}\leqslant2(a^{2}+b^{2})$, $\frac{1}{n}\left\Vert \Delta_{\epsilon}\right\Vert ^{2}\leqslant\frac{8}{n}\left\Vert \ddot{\epsilon}\odot\epsilon^{\ast}\right\Vert ^{2}+\frac{2}{n}\left\Vert \ddot{\epsilon}\odot\ddot{\epsilon}\right\Vert ^{2}$. The Cauchy--Schwarz inequality gives
\begin{align*}
\frac{1}{n}\left\Vert \ddot{\epsilon}\odot\epsilon^{\ast}\right\Vert ^{2} & =\frac{1}{n}\sum^{n}_{i=1}\ddot{\epsilon}^{2}_{i}\epsilon^{\ast2}_{i}\leqslant\left(\frac{1}{n}\sum^{n}_{i=1}\ddot{\epsilon}^{4}_{i}\right)^{1/2}\left(\frac{1}{n}\sum^{n}_{i=1}\epsilon^{\ast4}_{i}\right)^{1/2}=o_{p}(1),
\end{align*}
and $\frac{1}{n}\left\Vert \ddot{\epsilon}\odot\ddot{\epsilon}\right\Vert ^{2}=\frac{1}{n}\sum^{n}_{i=1}\ddot{\epsilon}^{4}_{i}=o_{p}(1)$. Combining these bounds gives $n^{-1}\left\Vert \Delta_{\epsilon}\right\Vert ^{2}=o_{p}(1)$.

\emph{Proof of part (i).} Since $\hat{\varsigma}-\tilde{\varsigma}=P\Delta_{\epsilon}$, $\frac{1}{n}\alpha^{\prime}\left(\hat{\varsigma}-\tilde{\varsigma}\right)=\frac{1}{n}\alpha^{\prime}P\Delta_{\epsilon}=\frac{1}{n}\tilde{\alpha}^{\prime}\Delta_{\epsilon}$, where $\tilde{\alpha}=P^{\prime}\alpha$ is UB by Lemma \ref{lem:rub}. Therefore,
\[
\frac{1}{n}\left|\alpha^{\prime}\left(\hat{\varsigma}-\tilde{\varsigma}\right)\right|\leqslant\left(\frac{1}{n}\left\Vert \tilde{\alpha}\right\Vert ^{2}\right)^{1/2}\left(\frac{1}{n}\left\Vert \Delta_{\epsilon}\right\Vert ^{2}\right)^{1/2}=O(1)o_{p}(1)=o_{p}(1).
\]

\emph{Proof of part (ii).} Since $M$ is symmetric, $\hat{\varsigma}-\tilde{\varsigma}=P\Delta_{\epsilon}$, and $\tilde{\varsigma}=P\left(\epsilon^{\ast}\odot\epsilon^{\ast}\right)$, 
\begin{align*}
\left|\frac{1}{n}\left(\hat{\varsigma}^{\prime}M\hat{\varsigma}-\tilde{\varsigma}^{\prime}M\tilde{\varsigma}\right)\right| & =\left|\frac{1}{n}\left(\hat{\varsigma}-\tilde{\varsigma}\right)^{\prime}M(\hat{\varsigma}-\tilde{\varsigma})+\frac{2}{n}\tilde{\varsigma}^{\prime}M(\hat{\varsigma}-\tilde{\varsigma})\right|\\
 & \leqslant\left|\frac{1}{n}\left(P\Delta_{\epsilon}\right)^{\prime}M\left(P\Delta_{\epsilon}\right)\right|+\left|\frac{2}{n}\left(\epsilon^{\ast}\odot\epsilon^{\ast}\right)^{\prime}P^{\prime}MP\Delta_{\epsilon}\right|.
\end{align*}
Since $P$ and $M$ are UBRC, Lemma \ref{lem:rub} implies that $P^{\prime}MP$ is also UBRC. By Lemma \ref{lem:spectral}, there exists a constant $C_{M}$ such that 
\begin{align*}
\left|\frac{1}{n}\left(\hat{\varsigma}^{\prime}M\hat{\varsigma}-\tilde{\varsigma}^{\prime}M\tilde{\varsigma}\right)\right| & \leqslant C_{M}\left[\frac{1}{n}\left\Vert \Delta_{\epsilon}\right\Vert ^{2}+\frac{1}{n}\left\Vert \Delta_{\epsilon}\right\Vert \left\Vert \epsilon^{\ast}\odot\epsilon^{\ast}\right\Vert \right]\\
 & \leqslant C_{M}\left[\frac{1}{n}\left\Vert \Delta_{\epsilon}\right\Vert ^{2}+\left(\frac{1}{n}\left\Vert \Delta_{\epsilon}\right\Vert ^{2}\frac{1}{n}\left\Vert \epsilon^{\ast}\odot\epsilon^{\ast}\right\Vert ^{2}\right)^{1/2}\right].
\end{align*}
We have shown that $\frac{1}{n}\left\Vert \Delta_{\epsilon}\right\Vert ^{2}=o_{p}(1)$ and $\frac{1}{n}\left\Vert \epsilon^{\ast}\odot\epsilon^{\ast}\right\Vert ^{2}=\frac{1}{n}\sum_{i}\epsilon^{\ast4}_{i}=O_{p}(1)$. Hence $n^{-1}\left|\left(\hat{\varsigma}^{\prime}M\hat{\varsigma}-\tilde{\varsigma}^{\prime}M\tilde{\varsigma}\right)\right|=o_{p}(1)$.

\emph{Proof of part (iii).} Since $\hat{\varsigma}=P\left(\hat{\epsilon}^{\ast}\odot\hat{\epsilon}^{\ast}\right)$, we have
\[
\frac{1}{n}\left\Vert \hat{\varsigma}\right\Vert ^{2}=\frac{1}{n}\left\Vert P\left(\hat{\epsilon}^{\ast}\odot\hat{\epsilon}^{\ast}\right)\right\Vert ^{2}=\frac{1}{n}\left(\hat{\epsilon}^{\ast}\odot\hat{\epsilon}^{\ast}\right)^{\prime}P^{\prime}P\left(\hat{\epsilon}^{\ast}\odot\hat{\epsilon}^{\ast}\right).
\]
Because $P$ is UBRC, $P^{\prime}P$ is also UBRC. Therefore, by Lemma \ref{lem:spectral}, there exists a constant $C_{P}<\infty$ such that $\frac{1}{n}\left\Vert \hat{\varsigma}\right\Vert ^{2}\leqslant C_{P}\frac{1}{n}\left\Vert \hat{\epsilon}^{\ast}\odot\hat{\epsilon}^{\ast}\right\Vert ^{2}=C_{P}\frac{1}{n}\sum_{i}\hat{\epsilon}^{\ast4}_{i}$. From (\ref{eq:ddoteps}), $\hat{\epsilon}^{\ast}=\ddot{\epsilon}+\epsilon^{\ast}$, where $\frac{1}{n}\sum^{n}_{i=1}\ddot{\epsilon}^{4}_{i}=o_{p}(1)$ and $\frac{1}{n}\sum^{n}_{i=1}(\epsilon^{\ast}_{i})^{4}=O_{p}(1)$. Using $(a+b)^{4}\leqslant8(a^{4}+b^{4})$, $\frac{1}{n}\sum^{n}_{i=1}(\hat{\epsilon}^{\ast}_{i})^{4}\leqslant8\frac{1}{n}\sum^{n}_{i=1}\ddot{\epsilon}^{4}_{i}+8\frac{1}{n}\sum^{n}_{i=1}(\epsilon^{\ast}_{i})^{4}=O_{p}(1)$. Hence $\frac{1}{n}\left\Vert \hat{\varsigma}\right\Vert ^{2}=O_{p}(1)$.
\end{proof}

\begin{lem}
\label{lem:rxr}Suppose Assumptions \ref{assu:size} and \ref{assu:epsilon} hold. In addition, suppose $\sup_{n\geqslant1}\sup_{i,g,r}\E\left|\epsilon_{igr}\right|^{8}<\infty$.

(i) Suppose that $\alpha=(\alpha^{\prime}_{1},\ldots,\alpha^{\prime}_{R})^{\prime}$ is a nonstochastic $n\times1$ vector that is UB, where $\alpha_{r}$ is $n_{r}\times1$ for $r=1,\ldots,R$. Then $\alpha^{\prime}\left(\tilde{\varsigma}-\varsigma\right)/n\xrightarrow{p}0$.

(ii) Suppose that $M=\diag\{M_{1},\ldots,M_{R}\}$ is a nonstochastic $n\times n$ matrix, where each $M_{r}$ is UBRC, symmetric, and has zero diagonals. Then $\left[\tilde{\varsigma}^{\prime}M\tilde{\varsigma}-\varsigma^{\prime}f(M)\varsigma\right]/n\xrightarrow{p}0$, where $f(M)=\diag^{R}_{r=1}\left\{ f(M_{r})\right\} $, and
\begin{align}
(n_{r}-2)^{2}f(M_{r}) & =\left[4+(n_{r}-2)^{2}\right]M_{r}+2\frac{\mathbf{1}^{\prime}_{r}M_{r}\mathbf{1}_{r}}{\left(n_{r}-1\right)^{2}}(J_{r}-I_{r})\nonumber \\
 & -\frac{4}{\left(n_{r}-1\right)}\left[M_{r}J_{r}+\left(M_{r}J_{r}\right)^{\prime}-2\Diag(M_{r}J_{r})\right],\label{eq:f(x)}
\end{align}
and $\Diag(M_{r}J_{r})$ denotes the diagonal matrix with the same diagonals as $M_{r}J_{r}$.
\end{lem}
\begin{proof}
Recall that $\tilde{\varsigma}_{r}=P_{r}\left(\epsilon^{\ast}_{r}\odot\epsilon^{\ast}_{r}\right)=\frac{n_{r}}{n_{r}-2}\left[I_{r}-\frac{1}{n_{r}(n_{r}-1)}J_{r}\right]\left(\epsilon^{\ast}_{r}\odot\epsilon^{\ast}_{r}\right)$. We first derive a decomposition of $\tilde{\varsigma}_{r}-\varsigma_{r}$. Since $\epsilon^{\ast}_{r}=\epsilon_{r}-\bar{\epsilon}_{r}\mathbf{1}_{r}$, 
\begin{align*}
\epsilon^{\ast}_{r}\odot\epsilon^{\ast}_{r} & =\left(\epsilon_{r}-\bar{\epsilon}_{r}\mathbf{1}_{r}\right)\odot\left(\epsilon_{r}-\bar{\epsilon}_{r}\mathbf{1}_{r}\right)=\epsilon_{r}\odot\epsilon_{r}-2\bar{\epsilon}_{r}\epsilon_{r}+\bar{\epsilon}^{2}_{r}\mathbf{1}_{r}=\epsilon_{r}\odot\epsilon_{r}-\frac{2}{n_{r}}\epsilon_{r}\epsilon^{\prime}_{r}\mathbf{1}_{r}+\mathbf{1}_{r}\frac{\epsilon^{\prime}_{r}J_{r}\epsilon_{r}}{n^{2}_{r}},\\
J_{r}\left(\epsilon^{\ast}_{r}\odot\epsilon^{\ast}_{r}\right) & =\mathbf{1}_{r}\left[\mathbf{1}^{\prime}_{r}\left(\epsilon^{\ast}_{r}\odot\epsilon^{\ast}_{r}\right)\right]=\mathbf{1}_{r}\left(\epsilon^{\ast\prime}_{r}\epsilon^{\ast}_{r}\right)=\mathbf{1}_{r}\epsilon^{\prime}_{r}I^{\ast}_{r}\epsilon_{r}=\mathbf{1}_{r}\epsilon^{\prime}_{r}(I_{r}-\frac{J_{r}}{n_{r}})\epsilon_{r}.
\end{align*}
Consequently,
\begin{align}
\frac{n_{r}-2}{n_{r}}\left(\tilde{\varsigma}_{r}-\varsigma_{r}\right) & =\left(\epsilon^{\ast}_{r}\odot\epsilon^{\ast}_{r}\right)-\frac{J_{r}}{n_{r}(n_{r}-1)}\left(\epsilon^{\ast}_{r}\odot\epsilon^{\ast}_{r}\right)-\frac{n_{r}-2}{n_{r}}\varsigma_{r}\nonumber \\
 & =\left(\epsilon_{r}\odot\epsilon_{r}-\varsigma_{r}\right)-\frac{2}{n_{r}}\left(\epsilon_{r}\epsilon^{\prime}_{r}\mathbf{1}_{r}-\varsigma_{r}\right)+\frac{\mathbf{1}_{r}}{n_{r}(n_{r}-1)}\epsilon^{\prime}_{r}\left[\frac{n_{r}-1}{n_{r}}J_{r}-\left(I_{r}-\frac{J_{r}}{n_{r}}\right)\right]\epsilon_{r}\nonumber \\
 & =\eta^{(1)}_{r}-2\eta^{(2)}_{r}+\eta^{(3)}_{r},\label{eq:diff_gamma}
\end{align}
where $\eta^{(1)}_{r}=(\epsilon^{2}_{1r}-\sigma^{2}_{1r},\ldots,\epsilon^{2}_{n_{r}r}-\sigma^{2}_{n_{r}r})^{\prime}$, $\eta^{(2)}_{r}=n^{-1}_{r}\left(\epsilon_{r}\epsilon^{\prime}_{r}\mathbf{1}_{r}-\varsigma_{r}\right)$, and $\eta^{(3)}_{r}=\frac{1}{n_{r}(n_{r}-1)}\mathbf{1}_{r}\epsilon^{\prime}_{r}\left(J_{r}-I_{r}\right)\epsilon_{r}$. Under Assumption \ref{assu:epsilon}, $\E\left(\eta^{(1)}_{r}\right)=\E\left(\eta^{(2)}_{r}\right)=\E\left(\eta^{(3)}_{r}\right)=0$.

We next show that $\E\left(\left\Vert \eta^{(1)}_{r}\right\Vert ^{4}\right)=O(n^{2}_{r})$, $\E\left(\left\Vert \eta^{(2)}_{r}\right\Vert ^{4}\right)=O(1)$ and $\E\left(\left\Vert \eta^{(3)}_{r}\right\Vert ^{4}\right)=O(1/n^{2}_{r})$. First consider $\eta^{(1)}_{r}$. Since $\sup_{n\geqslant1}\sup_{i,g,r}\E|\epsilon_{igr}|^{8}<\infty$, the elements of $\eta^{(1)}_{r}$ are independent and have uniformly bounded fourth moments. Therefore
\begin{align*}
\E\left(\left\Vert \eta^{(1)}_{r}\right\Vert ^{4}\right) & =\E\left[\left(\sum_{j=1}(\epsilon^{2}_{jr}-\sigma^{2}_{jr})^{2}\right)^{2}\right]\\
 & =\sum_{j=1}\E\left[(\epsilon^{2}_{jr}-\sigma^{2}_{jr})^{4}\right]+\sum_{j=1}\sum_{j'\neq j}\E\left[(\epsilon^{2}_{jr}-\sigma^{2}_{jr})^{2}\right]\E\left[(\epsilon^{2}_{j'r}-\sigma^{2}_{j'r})^{2}\right].
\end{align*}
For some finite constant $C_{\eta}$, $\E\left[(\epsilon^{2}_{jr}-\sigma^{2}_{jr})^{4}\right]\leqslant8\left[\E\left(\epsilon^{8}_{jr}\right)+\sigma^{8}_{jr}\right]\leqslant C_{\eta}$ and $\E\left[(\epsilon^{2}_{jr}-\sigma^{2}_{jr})^{2}\right]\leqslant C_{\eta}$. Hence $\E\left(\left\Vert \eta^{(1)}_{r}\right\Vert ^{4}\right)\leqslant n_{r}C_{\eta}+n_{r}(n_{r}-1)C_{\eta}=O(n^{2}_{r})$. For $\eta^{(2)}_{r}$, its definition gives
\begin{align*}
\left\Vert \eta^{(2)}_{r}\right\Vert  & \leqslant\left\Vert \frac{1}{n_{r}}\epsilon_{r}\epsilon^{\prime}_{r}\mathbf{1}_{r}\right\Vert +\left\Vert \frac{\varsigma_{r}}{n_{r}}\right\Vert =\left|\frac{\epsilon^{\prime}_{r}\mathbf{1}_{r}}{n_{r}}\right|\left\Vert \epsilon_{r}\right\Vert +\left\Vert \frac{\varsigma_{r}}{n_{r}}\right\Vert .
\end{align*}
Thus, by Hölder's inequality, 
\[
\E\left(\left\Vert \eta^{(2)}_{r}\right\Vert ^{4}\right)\leqslant8\left[\E\left(\left|\frac{\epsilon^{\prime}_{r}\mathbf{1}_{r}}{n_{r}}\right|^{4}\left\Vert \epsilon_{r}\right\Vert ^{4}\right)+\left\Vert \frac{\varsigma_{r}}{n_{r}}\right\Vert ^{4}\right]\leqslant\frac{8}{n^{4}_{r}}\left(\E\left|\epsilon^{\prime}_{r}\mathbf{1}_{r}\right|^{8}\E\left\Vert \epsilon_{r}\right\Vert ^{8}\right)^{\frac{1}{2}}+\frac{8}{n^{4}_{r}}\left\Vert \varsigma_{r}\right\Vert ^{4}.
\]
Because the $\epsilon_{jr}$ are independent and mean zero, $\E\left|\epsilon^{\prime}_{r}\mathbf{1}_{r}\right|^{8}=\sum_{j_{1},\ldots,j_{8}}\E(\epsilon_{j_{1}r}\ldots\epsilon_{j_{8}r})=O(n^{4}_{r})$.\footnote{Note that $\E(\epsilon_{j_{1}r}\ldots\epsilon_{j_{8}r})=0$ whenever any index appears exactly once. Thus when $\E(\epsilon_{j_{1}r}\ldots\epsilon_{j_{8}r})\neq0$, every index must appear at least twice, so there can be at most four distinct indices.} Moreover, $\E\left\Vert \epsilon_{r}\right\Vert ^{8}=\E\left[(\sum^{n_{r}}_{j=1}\epsilon^{2}_{jr})^{4}\right]\leqslant n^{3}_{r}\sum_{j}\E\left|\epsilon_{jr}\right|^{8}=O(n^{4}_{r})$. Since the elements of $\varsigma_{r}$ are uniformly bounded, $\left\Vert \varsigma_{r}\right\Vert ^{4}=O(n^{2}_{r})$. Consequently, $\E\left(\left\Vert \eta^{(2)}_{r}\right\Vert ^{4}\right)=O(1)$. For $\eta^{(3)}_{r}$, 
\[
\E\left[\left(\epsilon^{\prime}_{r}(J_{r}-I_{r})\epsilon_{r}\right)^{4}\right]\leqslant8\left[\E\left((\epsilon^{\prime}_{r}J_{r}\epsilon_{r})^{4}\right)+\E\left((\epsilon^{\prime}_{r}\epsilon_{r})^{4}\right)\right]=8\left[\E\left|\epsilon^{\prime}_{r}\mathbf{1}_{r}\right|^{8}+\E\left\Vert \epsilon_{r}\right\Vert ^{8}\right]=O(n^{4}_{r}).
\]
Therefore, 
\[
\E\left(\left\Vert \eta^{(3)}_{r}\right\Vert ^{4}\right)=\frac{\left\Vert \mathbf{1}_{r}\right\Vert ^{4}}{n^{4}_{r}(n_{r}-1)^{4}}\E\left[\left(\epsilon^{\prime}_{r}(J_{r}-I_{r})\epsilon_{r}\right)^{4}\right]=\frac{n^{2}_{r}}{n^{4}_{r}(n_{r}-1)^{4}}O(n^{4}_{r})=O(\frac{1}{n^{2}_{r}}).
\]
\emph{Proof of part (i).} Observe that $\alpha^{\prime}\left(\tilde{\varsigma}-\varsigma\right)=\sum^{R}_{r=1}\alpha^{\prime}_{r}\left(\tilde{\varsigma}_{r}-\varsigma_{r}\right)$. Since $\E\tilde{\varsigma}_{r}=\varsigma_{r}$, $\E\alpha^{\prime}_{r}\left(\tilde{\varsigma}_{r}-\varsigma_{r}\right)=0$. Next, using (\ref{eq:diff_gamma}) and the Cauchy--Schwarz inequality,
\begin{align*}
\Var\left[\alpha^{\prime}_{r}\left(\tilde{\varsigma}_{r}-\varsigma_{r}\right)\right] & =\left(\frac{n_{r}}{n_{r}-2}\right)^{2}\left[\E\left(\alpha^{\prime}_{r}\eta^{(1)}_{r}-2\alpha^{\prime}_{r}\eta^{(2)}_{r}+\alpha^{\prime}_{r}\eta^{(3)}_{r}\right)^{2}\right]\\
 & \leqslant3\left(\frac{n_{r}}{n_{r}-2}\right)^{2}\left[\E\left((\alpha^{\prime}_{r}\eta^{(1)}_{r})^{2}\right)+4\E\left((\alpha^{\prime}_{r}\eta^{(2)}_{r})^{2}\right)+\E\left((\alpha^{\prime}_{r}\eta^{(3)}_{r})^{2}\right)\right].
\end{align*}
First, $\E\left((\alpha^{\prime}_{r}\eta^{(1)}_{r})^{2}\right)=\E\alpha^{\prime}_{r}\eta^{(1)}_{r}\eta^{(1)\prime}_{r}\alpha_{r}=\alpha^{\prime}_{r}\Sigma^{(1)}_{\eta,r}\alpha_{r}$, where $\Sigma^{(1)}_{\eta,r}=\E(\eta^{(1)}_{r}\eta^{(1)\prime}_{r})$ is diagonal and UB under Assumption \ref{assu:epsilon}. As a result, $\E\left((\alpha^{\prime}_{r}\eta^{(1)}_{r})^{2}\right)=\alpha^{\prime}_{r}\Sigma^{(1)}_{\eta,r}\alpha_{r}=O(n_{r})$. For $l=2,3$, $\E\left\Vert \eta^{(l)}_{r}\right\Vert ^{4}=O(1)$ so $\E\left\Vert \eta^{(l)}_{r}\right\Vert ^{2}=O(1)$. Since $\alpha_{r}$ is UB we have $\E\left((\alpha^{\prime}_{r}\eta^{(l)}_{r})^{2}\right)\leqslant\left\Vert \alpha_{r}\right\Vert ^{2}\E\left\Vert \eta^{(l)}_{r}\right\Vert ^{2}\leqslant O(n_{r})$. In all, $\Var\left[\alpha^{\prime}_{r}\left(\tilde{\varsigma}_{r}-\varsigma_{r}\right)\right]\leqslant n_{r}C_{V}$ for some constant $0<C_{V}<\infty$. With $\tilde{\varsigma}_{r}$ independent across $r$, 
\[
\Var\left[\frac{1}{n}\alpha^{\prime}\left(\tilde{\varsigma}-\varsigma\right)\right]=\frac{1}{n^{2}}\sum^{R}_{r=1}\Var\left(\alpha^{\prime}_{r}\tilde{\varsigma}_{r}\right)\leqslant\frac{1}{n^{2}}\sum n_{r}C_{V}=\frac{C_{V}}{n}\rightarrow0.
\]
By Chebyshev's inequality $\frac{1}{n}\alpha^{\prime}\left(\tilde{\varsigma}-\varsigma\right)\xrightarrow{p}0$.

\noindent\emph{Proof of part (ii).} We first show that $\E\left(\tilde{\varsigma}^{\prime}_{r}M_{r}\tilde{\varsigma}_{r}\right)=\varsigma^{\prime}_{r}f(M_{r})\varsigma_{r}$ and $\Var\left(\tilde{\varsigma}^{\prime}_{r}M_{r}\tilde{\varsigma}_{r}\right)=O(n_{r})$ uniformly. Observe that
\begin{align}
\tilde{\varsigma}^{\prime}_{r}M_{r}\tilde{\varsigma}_{r} & =\left[\left(\tilde{\varsigma}_{r}-\varsigma_{r}\right)+\varsigma_{r}\right]^{\prime}M_{r}\left[\left(\tilde{\varsigma}_{r}-\varsigma_{r}\right)+\varsigma_{r}\right]\nonumber \\
 & =\left(\tilde{\varsigma}_{r}-\varsigma_{r}\right)^{\prime}M_{r}\left(\tilde{\varsigma}_{r}-\varsigma_{r}\right)+2\varsigma^{\prime}_{r}M_{r}\left(\tilde{\varsigma}_{r}-\varsigma_{r}\right)+\varsigma^{\prime}_{r}M_{r}\varsigma_{r}.\label{eq:rXr}
\end{align}
For the first term in (\ref{eq:rXr}), using (\ref{eq:diff_gamma}), we have 
\[
\left(\frac{n_{r}-2}{n_{r}}\right)^{2}\left(\tilde{\varsigma}_{r}-\varsigma_{r}\right)^{\prime}M_{r}\left(\tilde{\varsigma}_{r}-\varsigma_{r}\right)=\left(\eta^{(1)}_{r}-2\eta^{(2)}_{r}+\eta^{(3)}_{r}\right)^{\prime}M_{r}\left(\eta^{(1)}_{r}-2\eta^{(2)}_{r}+\eta^{(3)}_{r}\right).
\]
Under Assumption \ref{assu:epsilon}, Lemma \ref{lem:varS}, and Remark \ref{rem:vars},\footnote{See Section \ref{subsec:Cal_EfM} of the Online Appendix for details of these calculations.} the expected values of the quadratic terms are $\E\left(\eta^{(1)\prime}_{r}M_{r}\eta^{(1)}_{r}\right)=0$, $\E\left(\eta^{(1)\prime}_{r}M_{r}\eta^{(2)}_{r}\right)=0$, $\E\left(\eta^{(1)\prime}_{r}M_{r}\eta^{(3)}_{r}\right)=0$, $\E\left(\eta^{(2)\prime}_{r}M_{r}\eta^{(2)}_{r}\right)=\varsigma^{\prime}_{r}M_{r}\varsigma_{r}/n^{2}_{r}$, 
\[
\E\left(\eta^{(2)\prime}_{r}M_{r}\eta^{(3)}_{r}\right)=\frac{1}{n^{2}_{r}\left(n_{r}-1\right)}\varsigma^{\prime}_{r}\left[J_{r}M_{r}+M_{r}J_{r}-\Diag(J_{r}M_{r})-\Diag(M_{r}J_{r})\right]\varsigma_{r},
\]
\[
\E\left(\eta^{(3)\prime}_{r}M_{r}\eta^{(3)}_{r}\right)=2\frac{\left(\mathbf{1}^{\prime}_{r}M_{r}\mathbf{1}_{r}\right)}{n^{2}_{r}\left(n_{r}-1\right)^{2}}\varsigma^{\prime}_{r}(J_{r}-I_{r})\varsigma_{r}.
\]
For the second term in (\ref{eq:rXr}), $\E\left[\varsigma^{\prime}_{r}M_{r}\left(\tilde{\varsigma}_{r}-\varsigma_{r}\right)\right]=0$. In all, 
\begin{align*}
 & \E\left(\tilde{\varsigma}^{\prime}_{r}M_{r}\tilde{\varsigma}_{r}\right)=\E\left[\left(\tilde{\varsigma}_{r}-\varsigma_{r}\right)^{\prime}M_{r}\left(\tilde{\varsigma}_{r}-\varsigma_{r}\right)\right]+\varsigma^{\prime}_{r}M_{r}\varsigma_{r}\\
= & \left(\frac{n_{r}}{n_{r}-2}\right)^{2}\left[\frac{4\varsigma^{\prime}_{r}M_{r}\varsigma_{r}}{n^{2}_{r}}-\frac{4\varsigma^{\prime}_{r}\left[J_{r}M_{r}+M_{r}J_{r}-2\Diag(J_{r}M_{r})\right]\varsigma_{r}}{n^{2}_{r}\left(n_{r}-1\right)}+\frac{2\left(\mathbf{1}^{\prime}_{r}M_{r}\mathbf{1}_{r}\right)\varsigma^{\prime}_{r}(J_{r}-I_{r})\varsigma_{r}}{n^{2}_{r}\left(n_{r}-1\right)^{2}}\right]+\varsigma^{\prime}_{r}M_{r}\varsigma_{r}\\
= & \frac{1}{(n_{r}-2)^{2}}\varsigma^{\prime}_{r}\left[\left[4+(n_{r}-2)^{2}\right]M_{r}-\frac{4\left[J_{r}M_{r}+M_{r}J_{r}-2\Diag(J_{r}M_{r})\right]}{\left(n_{r}-1\right)}+2\frac{\left(\mathbf{1}^{\prime}_{r}M_{r}\mathbf{1}_{r}\right)(J_{r}-I_{r})}{\left(n_{r}-1\right)^{2}}\right]\varsigma_{r}.
\end{align*}
As a result, $\E\left(\tilde{\varsigma}^{\prime}_{r}M_{r}\tilde{\varsigma}_{r}\right)=\varsigma^{\prime}_{r}f(M_{r})\varsigma_{r}$, where $f(M_{r})$ is defined in (\ref{eq:f(x)}). Here symmetry of $M_{r}$ implies $\Diag(J_{r}M_{r})=\Diag(M_{r}J_{r})$.

Next, since $M_{r}$ has zero diagonals, Remark \ref{rem:vars}(i) indicates that $\E\left[(\eta^{(1)\prime}_{r}M_{r}\eta^{(1)}_{r})^{2}\right]=2\tr(\Sigma^{(1)}_{\eta,r}M_{r}\Sigma^{(1)}_{\eta,r}M_{r})$ where $\Sigma^{(1)}_{\eta,r}=\E(\eta^{(1)}_{r}\eta^{(1)\prime}_{r})$ is diagonal with uniformly bounded elements. Since $M_{r}$ is UBRC, $\Sigma^{(1)}_{\eta,r}M_{r}\Sigma^{(1)}_{\eta,r}M_{r}$ is UBRC and thus $\E\left[(\eta^{(1)\prime}_{r}M_{r}\eta^{(1)}_{r})^{2}\right]$ is $O(n_{r})$. For $l,l'=1,2,3$ but either $l\neq1$ or $l'\neq1$, since $M_{r}$ is UBRC, $\left|\eta^{(l)\prime}_{r}M_{r}\eta^{(l')}_{r}\right|\leqslant C_{M}\left\Vert \eta^{(l)}_{r}\right\Vert \left\Vert \eta^{(l')}_{r}\right\Vert $ by Lemma \ref{lem:spectral}. Thus
\[
\E\left[(\eta^{(l)\prime}_{r}M_{r}\eta^{(l')}_{r})^{2}\right]\leqslant C^{2}_{M}\E\left[\left\Vert \eta^{(l)}_{r}\right\Vert ^{2}\left\Vert \eta^{(l')}_{r}\right\Vert ^{2}\right]\leqslant C^{2}_{M}\left[\E\left\Vert \eta^{(l)}_{r}\right\Vert ^{4}\E\left\Vert \eta^{(l')}_{r}\right\Vert ^{4}\right]^{\frac{1}{2}}\leqslant O(n_{r}),
\]
where the last inequality follows because $\E\left\Vert \eta^{(l)}_{r}\right\Vert ^{4}=O(1)$ for $l=2,3$, and $\E\left\Vert \eta^{(1)}_{r}\right\Vert ^{4}=O(n^{2}_{r})$, while at least one of $l,l'$ is not $1$. Thus $\left(\tilde{\varsigma}_{r}-\varsigma_{r}\right)^{\prime}M_{r}\left(\tilde{\varsigma}_{r}-\varsigma_{r}\right)$ in (\ref{eq:rXr}) is a sum of six terms whose second moments are uniformly $O(n_{r})$. For the second term in (\ref{eq:rXr}), $M_{r}\varsigma_{r}$ is UB because $M_{r}$ is UBRC and $\varsigma_{r}$ is UB. Applying the variance calculation in the proof of part (i) with $\alpha_{r}=M_{r}\varsigma_{r}$, $\Var\left[\varsigma^{\prime}_{r}M_{r}\left(\tilde{\varsigma}_{r}-\varsigma_{r}\right)\right]=\Var\left[(M_{r}\varsigma_{r})^{\prime}\left(\tilde{\varsigma}_{r}-\varsigma_{r}\right)\right]=O(n_{r})$. Since $\tilde{\varsigma}^{\prime}_{r}M_{r}\tilde{\varsigma}_{r}$ in (\ref{eq:rXr}) is the sum of a fixed number of terms whose variances are uniformly $O(n_{r})$, $\Var\left(\tilde{\varsigma}^{\prime}_{r}M_{r}\tilde{\varsigma}_{r}\right)=O(n_{r})$ uniformly.

Consequently, $\E\left(\tilde{\varsigma}^{\prime}M\tilde{\varsigma}\right)=\varsigma^{\prime}f(M)\varsigma$ and $\Var\left(\tilde{\varsigma}^{\prime}M\tilde{\varsigma}/n\right)=\sum^{R}_{r=1}\Var(\tilde{\varsigma}^{\prime}_{r}M_{r}\tilde{\varsigma}_{r})/n^{2}=O(n^{-1})$. By Chebyshev's inequality $\left[\tilde{\varsigma}^{\prime}M\tilde{\varsigma}-\varsigma^{\prime}f(M)\varsigma\right]/n\xrightarrow{p}0$.
\end{proof}

\begin{prop}
\label{prop:VC_convergence} Suppose Assumptions \ref{assu:size}-\ref{assu:lim} hold, and $\sup_{n\geqslant1}\sup_{i,g,r}\E\left|\epsilon_{igr}\right|^{8}<\infty$.

(i) Suppose that $\alpha=(\alpha^{\prime}_{1},\ldots,\alpha^{\prime}_{R})^{\prime}$ is a nonstochastic $n\times1$ vector that is UB, where $\alpha_{r}$ is $n_{r}\times1$ for $r=1,\ldots,R$. Then $\alpha^{\prime}\left(\hat{\varsigma}-\varsigma\right)/n\xrightarrow{p}0$.

(ii) Suppose $M=\diag\{M_{1},\ldots,M_{R}\}$ is a nonstochastic $n\times n$ matrix. Suppose each $M_{r}$ satisfies the conditions specified in Lemma \ref{lem:rxr} and has the form
\begin{equation}
M_{r}=\diag^{G_{r}}_{g=1}\{p_{gr}I^{\ast}_{gr}+q_{gr}J^{\ast}_{gr}\}+\varpi_{r}I^{\ast}_{r},\label{eq:pIqJcI}
\end{equation}
such that $\diag(M_{r})=0$, and $p_{gr}$, $q_{gr}$ and $\varpi_{r}$ are uniformly bounded. Let $g(M)=\diag^{R}_{r=1}g(M_{r})$, where
\begin{align}
g(M_{r}) & =\frac{(n_{r}-2)^{2}}{\left[4+(n_{r}-2)^{2}\right]}\left[M_{r}+\frac{4T_{r}}{(n_{r}-2)^{2}(n_{r}-1)+4}\right]+t_{r}(J_{r}-I_{r}),\label{eq:g_M}\\
T_{r} & =\left(\mathcal{D}_{r}-\frac{\tr(\mathcal{D}_{r})}{4\left(n_{r}-1\right)}I_{r}\right)(J_{r}-I_{r})+(J_{r}-I_{r})\left(\mathcal{D}_{r}-\frac{\tr(\mathcal{D}_{r})}{4\left(n_{r}-1\right)}I_{r}\right),\nonumber \\
t_{r} & =\frac{4(n_{r}-2)(3n_{r}-4)\tr(\mathcal{D}_{r})}{n_{r}(n_{r}-1)(n_{r}-3)\left[4+(n_{r}-2)^{2}\right]\left[(n_{r}-2)^{2}(n_{r}-1)+4\right]},\nonumber \\
\mathcal{D}_{r} & =\diag^{G_{r}}_{g=1}\{q_{gr}I_{gr}\}.\nonumber 
\end{align}
Then $\left[\hat{\varsigma}^{\prime}g(M)\hat{\varsigma}-\varsigma^{\prime}M\varsigma\right]/n\xrightarrow{p}0$.
\end{prop}
\begin{proof}
(i) Part (i) follows from Lemma \ref{lem:hat_tilde_ip}(i) and Lemma \ref{lem:rxr}(i).

(ii) We first verify that $g(M)$ satisfies the conditions in Lemmas \ref{lem:hat_tilde_ip}(ii) and \ref{lem:rxr}(ii). The matrices $M_{r}$, $J_{r}-I_{r}$, and $T_{r}$ are symmetric and have zero diagonals. Since $q_{gr}$ is uniformly bounded, $\tr(\mathcal{D}_{r})=O(n_{r})$ uniformly. It follows that $\mathcal{D}_{r}-\frac{\tr(\mathcal{D}_{r})}{4\left(n_{r}-1\right)}I_{r}$ is UBRC, and hence $T_{r}/n_{r}$ is UBRC. Also, $t_{r}=O(n^{-5}_{r})$. Thus $g(M_{r})$ is UBRC, symmetric, and has zero diagonal. Lemmas \ref{lem:hat_tilde_ip}(ii) and \ref{lem:rxr}(ii) give $\left[\hat{\varsigma}^{\prime}g(M)\hat{\varsigma}-\tilde{\varsigma}^{\prime}g(M)\tilde{\varsigma}\right]/n\xrightarrow{p}0$, $\left[\tilde{\varsigma}^{\prime}g(M)\tilde{\varsigma}-\varsigma^{\prime}f(g(M))\varsigma\right]/n\xrightarrow{p}0$, and hence $\left[\hat{\varsigma}^{\prime}g(M)\hat{\varsigma}-\varsigma^{\prime}f(g(M))\varsigma\right]/n\xrightarrow{p}0$. For the special form of $M_{r}$ in (\ref{eq:pIqJcI}), direct calculation gives $f(g(M_{r}))=M_{r}$.\footnote{See Section \ref{subsec:Proof-of-Proposition_g(x)} of the Online Appendix for details.} Therefore $\left(\hat{\varsigma}^{\prime}g(M)\hat{\varsigma}-\varsigma^{\prime}M\varsigma\right)/n\xrightarrow{p}0$.
\end{proof}

\subsubsection{Proof of the theorem}\label{subsec:Proof-of-VC-estimation}
\begin{proof}
Under the maintained assumptions, $\hat{\theta}\xrightarrow{p}\theta_{0}$ by Theorem \ref{thm:consistency_theta}. It suffices to show that $\hat{V}-V_{n}\xrightarrow{p}0$ and $\hat{D}-D_{n}\xrightarrow{p}0$, where $V_{n}$, $D_{n}$, $\hat{V}$ and $\hat{D}$ are defined in (\ref{eq:Vn}), (\ref{eq:Dn}), (\ref{eq:Vhat}), and (\ref{eq:Dhat}).

Observe that $A^{\ast}_{r}\odot A^{\ast}_{r}$ has the form in (\ref{eq:pIqJcI}), with $p_{gr}$, $q_{gr}$ and $\varpi_{r}$ all uniformly bounded and $q_{gr}=\frac{n_{r}-m_{gr}}{\left(m_{gr}-1\right)\left(n_{r}-1\right)}$.\footnote{See Section \ref{subsec:AodotA} of the Online Appendix for detailed calculations.} For $A^{\dagger}_{r}=g(A^{\ast}_{r}\odot A^{\ast}_{r})$, where $g(.)$ is defined in (\ref{eq:g_M}) and $\mathcal{D}_{r}=\diag^{G_{r}}_{g=1}\{\frac{n_{r}-m_{gr}}{(m_{gr}-1)(n_{r}-1)}I_{gr}\}$, Proposition \ref{prop:VC_convergence}(ii) gives $\frac{1}{n}\sum^{R}_{r=1}\left[\hat{\varsigma}^{\prime}_{r}A^{\dagger}_{r}\hat{\varsigma}_{r}-\varsigma^{\prime}_{r}(A^{\ast}_{r}\odot A^{\ast}_{r})\varsigma_{r}\right]\xrightarrow{p}0$, where $\varsigma^{\prime}_{r}(A^{\ast}_{r}\odot A^{\ast}_{r})\varsigma_{r}=\tr(A^{\ast}_{r}\Omega_{r}A^{\ast}_{r}\Omega_{r})$ by (\ref{eq:trAOAO}). For $H^{\ast\prime}\Omega H^{\ast}/n$, let $h^{\ast}_{l}$ be the $l$th column of $H^{\ast}$, so that $H^{\ast}=(h^{\ast}_{1},\ldots,h^{\ast}_{k_{h}})$. The $(i,j)$ element of $H^{\ast\prime}\Omega H^{\ast}$ is $\left(H^{\ast\prime}\Omega H^{\ast}\right)_{ij}=h^{\ast\prime}_{i}\Omega h^{\ast}_{j}=\diag(h^{\ast}_{j}h^{\ast\prime}_{i})^{\prime}\varsigma$, where $\diag(h^{\ast}_{j}h^{\ast\prime}_{i})$ is UB because $H^{\ast}=I^{\ast}H$ is UB. Proposition \ref{prop:VC_convergence}(i) therefore yields $\diag(h^{\ast}_{j}h^{\ast\prime}_{i})^{\prime}\left(\hat{\varsigma}-\varsigma\right)/n\xrightarrow{p}0$ for every fixed $i,j$. Since $k_{h}$ is fixed, $\frac{1}{n}\left(H^{\ast\prime}\hat{\Omega}H^{\ast}-H^{\ast\prime}\Omega H^{\ast}\right)\xrightarrow{p}0$, and hence $\hat{V}-V_{n}\xrightarrow{p}0$.

Next, for the $(1,1)$ element of $\hat{D}$ in (\ref{eq:Dhat}), observe that $\tr\left(A^{\ast}W(I-\lambda W)^{-1}\Omega\right)=\varsigma^{\prime}d(\lambda)$ for any $\lambda\in(-1,1)$, where $d(\lambda)=\diag\left[A^{\ast}W(I-\lambda W)^{-1}\right]$. We have
\begin{align*}
\hat{\varsigma}^{\prime}d(\hat{\lambda})-\varsigma^{\prime}d(\lambda_{0})= & (\hat{\varsigma}-\varsigma)^{\prime}d(\lambda_{0})+\hat{\varsigma}^{\prime}\left[d(\hat{\lambda})-d(\lambda_{0})\right].
\end{align*}
By Remark \ref{rem:rub1}(ii), $A^{\ast}W(I-\lambda_{0}W)^{-1}$ is UBRC so $d(\lambda_{0})$ is UB. Proposition \ref{prop:VC_convergence}(i) then gives $(\hat{\varsigma}-\varsigma)^{\prime}d(\lambda_{0})/n\xrightarrow{p}0$.

From the closed-form expression for $(I-\lambda W)^{-1}$ in (\ref{eq:I_lW_inv}) and compactness of $\Lambda\subset(-1,1)$, $(I-\lambda W)^{-1}$ is UBRC uniformly over $\lambda\in\Lambda$ and $n$. Hence $A^{\ast}W(I-\lambda W)^{-1}$ is also UBRC uniformly over $\lambda\in\Lambda$ and $n$ by Lemma \ref{lem:rub}. In addition,
\begin{align}
 & (I-\hat{\lambda}W)^{-1}-(I-\lambda_{0}W)^{-1}\nonumber \\
= & (I-\hat{\lambda}W)^{-1}(I-\lambda_{0}W)(I-\lambda_{0}W)^{-1}-(I-\hat{\lambda}W)^{-1}(I-\hat{\lambda}W)(I-\lambda_{0}W)^{-1}\nonumber \\
= & (\hat{\lambda}-\lambda_{0})(I-\hat{\lambda}W)^{-1}W(I-\lambda_{0}W)^{-1}.\label{eq:diff-I_lW_inv}
\end{align}
Thus 
\[
\frac{1}{n}\hat{\varsigma}^{\prime}\left[d(\hat{\lambda})-d(\lambda_{0})\right]=(\hat{\lambda}-\lambda_{0})\frac{1}{n}\hat{\varsigma}^{\prime}\ddot{d}(\hat{\lambda}),
\]
where $\ddot{d}(\hat{\lambda})=\diag\left[A^{\ast}W(I-\hat{\lambda}W)^{-1}W(I-\lambda_{0}W)^{-1}\right]$. By the preceding uniform UBRC bound and Lemma \ref{lem:rub}, $A^{\ast}W(I-\hat{\lambda}W)^{-1}W(I-\lambda_{0}W)^{-1}$ is UBRC uniformly over $\lambda\in\Lambda$ and $n$. Thus $\ddot{d}(\hat{\lambda})$ is bounded uniformly over $\hat{\lambda}\in\Lambda$ and $n$. 
\[
\frac{1}{n}\left|\hat{\varsigma}^{\prime}\ddot{d}(\hat{\lambda})\right|\leqslant\frac{1}{n}\sum_{i=1}\left|\hat{\varsigma}_{i}\ddot{d}(\hat{\lambda})_{i}\right|\leqslant C_{d}\frac{1}{n}\sum^{n}_{i=1}\left|\hat{\varsigma}_{i}\right|.
\]
By Lemma \ref{lem:hat_tilde_ip}(iii), $\frac{1}{n}\left\Vert \hat{\varsigma}\right\Vert ^{2}=O_{p}(1)$. Hence, by the Cauchy--Schwarz inequality, $\frac{1}{n}\sum^{n}_{i=1}\left|\hat{\varsigma}_{i}\right|=O_{p}(1)$. Therefore $\frac{1}{n}\left|\hat{\varsigma}^{\prime}\ddot{d}(\hat{\lambda})\right|=O_{p}(1)$ and $\left[\hat{\varsigma}^{\prime}d(\hat{\lambda})-\varsigma^{\prime}d(\lambda_{0})\right]/n\xrightarrow{p}0$ since $\hat{\lambda}\xrightarrow{p}\lambda_{0}$.

Finally, 
\begin{align*}
 & \frac{1}{n}\left[H^{\prime}I^{\ast}W(I-\hat{\lambda}W)^{-1}X\hat{\beta}-H^{\prime}I^{\ast}W(I-\lambda_{0}W)^{-1}X\beta_{0}\right]\\
= & \frac{1}{n}H^{\prime}I^{\ast}W(I-\hat{\lambda}W)^{-1}X\left(\hat{\beta}-\beta_{0}\right)+\frac{1}{n}H^{\prime}I^{\ast}W\left[(I-\hat{\lambda}W)^{-1}-(I-\lambda_{0}W)^{-1}\right]X\beta_{0}.
\end{align*}
By (\ref{eq:I_lW_inv}), the row and column sums of $(I-\lambda W)^{-1}$ are bounded uniformly in $\lambda\in\Lambda$ and $n$. Hence by Lemma \ref{lem:rub} and Remark \ref{rem:rub2}, $\frac{1}{n}H^{\prime}I^{\ast}W(I-\lambda W)^{-1}X$ is bounded uniformly in $\lambda\in\Lambda$ and $n$. The first term above converges in probability to 0 as $\hat{\beta}\xrightarrow{p}\beta_{0}$. Given (\ref{eq:diff-I_lW_inv}), the second term is $(\hat{\lambda}-\lambda_{0})\left[\frac{1}{n}H^{\prime}I^{\ast}W(I-\hat{\lambda}W)^{-1}W(I-\lambda_{0}W)^{-1}X\beta_{0}\right]$, where the bracketed term is bounded uniformly in $\hat{\lambda}\in\Lambda$ and $n$. Hence, as $\hat{\lambda}\xrightarrow{p}\lambda_{0}$, the second term also converges to 0 in probability. Therefore $\hat{D}-D_{n}\xrightarrow{p}0$.

By the conditions of Theorem \ref{thm:asym_theta}, $V_{n}\rightarrow\bar{V}$ and $D_{n}\rightarrow\bar{D}$, while Assumption \ref{assu:lim} gives $\Xi_{n}\xrightarrow{p}\bar{\Xi}$. Thus $\hat{V}\xrightarrow{p}\bar{V}$ and $\hat{D}\xrightarrow{p}\bar{D}$. Since $\bar{D}$ has full column rank and $\bar{\Xi}$ is positive definite, $\bar{D}^{\prime}\bar{\Xi}\bar{D}$ is nonsingular. Therefore $\hat{D}^{\prime}\Xi_{n}\hat{D}$ is nonsingular w.p.a.1, and the continuous mapping theorem gives $\hat{\Sigma}_{\theta}\xrightarrow{p}\Sigma_{\theta}$.
\end{proof}

\clearpage{}

\setcounter{section}{0}\setcounter{page}{1}\setcounter{table}{0}\setcounter{equation}{0}\setcounter{lem}{0}\setcounter{assumption}{0}\setcounter{thm}{0}\setcounter{footnote}{0}\setcounter{rem}{0}\setcounter{prop}{0}

\renewcommand{\theequation}{O.\arabic{equation}}
\renewcommand{\thetable}{O.\arabic{table}}
\renewcommand{\thesection}{O.\arabic{section}}
\renewcommand{\theassumption}{O.\arabic{assumption}}
\renewcommand{\thelem}{O.\arabic{lem}}
\renewcommand{\thethm}{O.\arabic{thm}}
\renewcommand{\therem}{O.\arabic{rem}}

\part*{Online Appendix}

This Online Appendix provides calculation details, reviews existing tests of random assignment to peer groups, and reports additional simulation results and an additional empirical application.

\section{Calculations}

\subsection{Diagonal Elements of $\tilde{A}(\lambda)$ and $A^{\ast}$}\label{subsec:diag_Atilde}

Since $A^{\ast}=\tilde{A}(\lambda_{0})$, we first derive the diagonal elements of $\tilde{A}(\lambda)$. By (\ref{eq:I_lW}), (\ref{eq:Ar_pIqJ}), and (\ref{eq:I_lW_inv}), the matrices $A_{r}$, $I_{r}-\lambda W_{r}$, and $(I_{r}-\lambda_{0}W_{r})^{-1}$ can all be written as $\diag^{G_{r}}_{g=1}\{p_{gr}I^{\ast}_{gr}+q_{gr}J^{\ast}_{gr}\}$. Thus, by (\ref{eq:prod_s}), 
\[
(I_{r}-\lambda W_{r})^{2}(I_{r}-\lambda_{0}W_{r})^{-2}A_{r}=\frac{1}{n_{r}-1}\diag^{G_{r}}_{g=1}\{\left(\frac{m_{gr}-1+\lambda}{m_{gr}-1+\lambda_{0}}\right)^{2}(-\frac{n_{r}-m_{gr}}{m_{gr}-1})I^{\ast}_{gr}+\left(\frac{1-\lambda}{1-\lambda_{0}}\right)^{2}n_{r}J^{\ast}_{gr}\}.
\]
By Lemma \ref{lem:calculation},
\begin{align*}
 & (n_{r}-1)I^{\ast}_{r}(I_{r}-\lambda W_{r})^{2}(I_{r}-\lambda_{0}W_{r})^{-2}A_{r}\\
= & \diag^{G_{r}}_{g=1}\{\left[\left(\frac{m_{gr}-1+\lambda}{m_{gr}-1+\lambda_{0}}\right)^{2}(-\frac{n_{r}-m_{gr}}{m_{gr}-1})-\left(\frac{1-\lambda}{1-\lambda_{0}}\right)^{2}n_{r}\right]I^{\ast}_{gr}\}+\left(\frac{1-\lambda}{1-\lambda_{0}}\right)^{2}n_{r}I^{\ast}_{r}.
\end{align*}
Hence, for individuals in peer group $g$, the corresponding diagonal elements of $\left(n_{r}-1\right)\tilde{A}_{r}(\lambda)$ are
\begin{align}
 & -\left(\frac{m_{gr}-1+\lambda}{m_{gr}-1+\lambda_{0}}\right)^{2}\frac{n_{r}-m_{gr}}{m_{gr}-1}(1-\frac{1}{m_{gr}})+\left(\frac{1-\lambda}{1-\lambda_{0}}\right)^{2}n_{r}\left[(1-\frac{1}{n_{r}})-(1-\frac{1}{m_{gr}})\right]\nonumber \\
= & -\left(\frac{m_{gr}-1+\lambda}{m_{gr}-1+\lambda_{0}}\right)^{2}\frac{n_{r}-m_{gr}}{m_{gr}}+\left(\frac{1-\lambda}{1-\lambda_{0}}\right)^{2}\frac{n_{r}-m_{gr}}{m_{gr}}\nonumber \\
= & -\frac{n_{r}-m_{gr}}{m_{gr}}\left[\left(\frac{m_{gr}-1+\lambda}{m_{gr}-1+\lambda_{0}}\right)^{2}-\left(\frac{1-\lambda}{1-\lambda_{0}}\right)^{2}\right]\nonumber \\
= & -\frac{n_{r}-m_{gr}}{m_{gr}}\left(\frac{m_{gr}-1+\lambda}{m_{gr}-1+\lambda_{0}}+\frac{1-\lambda}{1-\lambda_{0}}\right)\left(\frac{m_{gr}-1+\lambda_{0}+(\lambda-\lambda_{0})}{m_{gr}-1+\lambda_{0}}-\frac{1-\lambda_{0}-(\lambda-\lambda_{0})}{1-\lambda_{0}}\right)\nonumber \\
= & -(\lambda-\lambda_{0})\frac{n_{r}-m_{gr}}{m_{gr}}\varphi_{gr}(\lambda),\label{eq:diag_Atilde}
\end{align}
where
\[
\varphi_{gr}(\lambda)=\left(\frac{1}{1-\lambda_{0}}+\frac{1}{m_{gr}-1+\lambda_{0}}\right)\left(\frac{m_{gr}-1+\lambda}{m_{gr}-1+\lambda_{0}}+\frac{1-\lambda}{1-\lambda_{0}}\right).
\]
Consequently, the contribution of peer group $g$ to $\tr\left(\tilde{A}_{r}(\lambda)\Omega_{r}\right)$ is $-(\lambda-\lambda_{0})\frac{n_{r}-m_{gr}}{m_{gr}(n_{r}-1)}\varphi_{gr}(\lambda)\tr(\Omega_{gr})$.

Since $A^{\ast}=\tilde{A}(\lambda_{0})$, setting $\lambda=\lambda_{0}$ in (\ref{eq:diag_Atilde}) shows that all diagonal elements of $A^{\ast}$ are zero.

\subsection{Matrix $A^{\ast}_{r}\odot A^{\ast}_{r}$}\label{subsec:AodotA}

We express $A^{\ast}_{r}\odot A^{\ast}_{r}$ in the form $\diag\{p_{gr}I^{\ast}_{gr}+q_{gr}J^{\ast}_{gr}\}+\varpi_{r}I^{\ast}_{r}$ by deriving $p_{gr}$, $q_{gr}$, and $\varpi_{r}$. Using the representation of $A_{r}$ in (\ref{eq:Ar_pIqJ}) and Lemma \ref{lem:calculation},
\begin{align*}
I^{\ast}_{r}A_{r} & =\diag^{G_{r}}_{g=1}\left[\frac{1}{n_{r}-1}\left((-\frac{n_{r}-m_{gr}}{m_{gr}-1})-n_{r}\right)\right]I^{\ast}_{gr}+\frac{n_{r}}{n_{r}-1}I^{\ast}_{r}=\diag^{G_{r}}_{g=1}\left\{ -\frac{m_{gr}}{m_{gr}-1}I^{\ast}_{gr}\right\} +\frac{n_{r}}{n_{r}-1}I^{\ast}_{r}.
\end{align*}
The diagonal elements of $A^{\ast}_{r}$ are zero, as shown above. Hence, the diagonal elements of $A^{\ast}_{r}\odot A^{\ast}_{r}$ are also zero, which implies
\begin{align}
p_{gr}(1-\frac{1}{m_{gr}})+q_{gr}\frac{1}{m_{gr}}+\varpi_{r}(1-\frac{1}{n_{r}}) & =0.\label{eq:calc1}
\end{align}
For two distinct individuals in peer group $g$, the corresponding element of $A^{\ast}_{r}\odot A^{\ast}_{r}$ is the square of the corresponding off-diagonal element of $A^{\ast}_{r}$. Therefore,
\begin{align}
p_{gr}(-\frac{1}{m_{gr}})+q_{gr}\frac{1}{m_{gr}}+\varpi_{r}(-\frac{1}{n_{r}}) & =\left[-\frac{m_{gr}}{m_{gr}-1}(-\frac{1}{m_{gr}})+\frac{n_{r}}{n_{r}-1}(-\frac{1}{n_{r}})\right]^{2}\nonumber \\
 & =\left[\frac{1}{m_{gr}-1}-\frac{1}{n_{r}-1}\right]^{2}=\left[\frac{n_{r}-m_{gr}}{\left(m_{gr}-1\right)\left(n_{r}-1\right)}\right]^{2}.\label{eq:calc2}
\end{align}
For two individuals in different peer groups, the corresponding element is
\begin{equation}
\varpi_{r}(-\frac{1}{n_{r}})=\left[\frac{n_{r}}{n_{r}-1}(-\frac{1}{n_{r}})\right]^{2}=\frac{1}{(n_{r}-1)^{2}}.\label{eq:calc3}
\end{equation}

Equation (\ref{eq:calc3}) gives $\varpi_{r}=-\frac{n_{r}}{(n_{r}-1)^{2}}$. Substituting this expression into (\ref{eq:calc1}) and (\ref{eq:calc2}) gives
\begin{align}
p_{gr}(1-\frac{1}{m_{gr}})+q_{gr}\frac{1}{m_{gr}}-\frac{1}{(n_{r}-1)} & =0,\label{eq:calc5}\\
p_{gr}(-\frac{1}{m_{gr}})+q_{gr}\frac{1}{m_{gr}}+\frac{1}{(n_{r}-1)^{2}} & =\left[\frac{n_{r}-m_{gr}}{\left(m_{gr}-1\right)\left(n_{r}-1\right)}\right]^{2}.\label{eq:calc6}
\end{align}
Equations (\ref{eq:calc5})--(\ref{eq:calc6}) imply $p_{gr}-\frac{n_{r}}{(n_{r}-1)^{2}}=-\left[\frac{n_{r}-m_{gr}}{\left(m_{gr}-1\right)\left(n_{r}-1\right)}\right]^{2}$ and hence $p_{gr}=\frac{1}{(n_{r}-1)^{2}}\left[n_{r}-\left(\frac{n_{r}-m_{gr}}{m_{gr}-1}\right)^{2}\right]$. Adding (\ref{eq:calc5}) to $(m_{gr}-1)\times$(\ref{eq:calc6}) yields
\begin{align*}
\frac{q_{gr}}{m_{gr}}-\frac{1}{(n_{r}-1)}+(m_{gr}-1)\frac{q_{gr}}{m_{gr}}+\frac{m_{gr}-1}{(n_{r}-1)^{2}} & =\left[\frac{n_{r}-m_{gr}}{\left(m_{gr}-1\right)\left(n_{r}-1\right)}\right]^{2}(m_{gr}-1),\\
q_{gr}-\frac{n_{r}-m_{gr}}{(n_{r}-1)^{2}} & =\frac{(n_{r}-m_{gr})^{2}}{\left(m_{gr}-1\right)\left(n_{r}-1\right)^{2}}.
\end{align*}
Therefore,
\begin{align*}
q_{gr} & =\frac{n_{r}-m_{gr}}{\left(m_{gr}-1\right)\left(n_{r}-1\right)^{2}}(n_{r}-m_{gr}+m_{gr}-1)=\frac{n_{r}-m_{gr}}{\left(m_{gr}-1\right)\left(n_{r}-1\right)}.
\end{align*}
Consequently, $A^{\ast}_{r}\odot A^{\ast}_{r}=\diag\{p_{gr}I^{\ast}_{gr}+q_{gr}J^{\ast}_{gr}\}+\varpi_{r}I^{\ast}_{r}$, with $p_{gr}=\frac{1}{(n_{r}-1)^{2}}\left[n_{r}-\left(\frac{n_{r}-m_{gr}}{m_{gr}-1}\right)^{2}\right]$, $q_{gr}=\frac{n_{r}-m_{gr}}{\left(m_{gr}-1\right)\left(n_{r}-1\right)}$, and $\varpi_{r}=-\frac{n_{r}}{(n_{r}-1)^{2}}$. Since $A^{\ast}_{r}$ is symmetric, Lemma \ref{lem:calculation}(i) gives
\begin{align}
\tr(A^{\ast2}_{r})=\mathbf{1}^{\prime}_{r}\left(A^{\ast}_{r}\odot A^{\ast}_{r}\right)\mathbf{1}_{r} & =\tr\left[\left(A^{\ast}_{r}\odot A^{\ast}_{r}\right)J_{r}\right]=\tr\left[\diag(q_{gr})J_{r}\right]=\frac{1}{(n_{r}-1)}\sum^{G_{r}}_{g=1}\frac{(n_{r}-m_{gr})m_{gr}}{m_{gr}-1}.\label{eq:lAAl}
\end{align}

\subsection{Calculations for Lemma \ref{lem:rxr}(ii)}\label{subsec:Cal_EfM}

We verify the expression in (\ref{eq:f(x)}) term by term.
\begin{enumerate}
\item Since $\Sigma^{(1)}_{\eta,r}=\E(\eta^{(1)}_{r}\eta^{(1)\prime}_{r})$ is diagonal and $M_{r}$ has zero diagonals, $\E(\eta^{(1)\prime}_{r}M_{r}\eta^{(1)}_{r})=\tr(\Sigma^{(1)}_{\eta r}M_{r})=0$.
\item We have $\E\left(\eta^{(1)\prime}_{r}M_{r}\eta^{(2)}_{r}\right)=\E\left[\eta^{(1)\prime}_{r}M_{r}\left(\epsilon_{r}\epsilon^{\prime}_{r}\mathbf{1}_{r}-\varsigma_{r}\right)\right]/n_{r}$. For the first component, $\E\left(\eta^{(1)\prime}_{r}M_{r}\epsilon_{r}\epsilon^{\prime}_{r}\mathbf{1}_{r}\right)=\sum^{n_{r}}_{i,j,k}\E\left(M_{r,ij}\eta^{(1)}_{i}\epsilon_{j}\epsilon_{k}\right)$. Because $(\eta^{(1)}_{i},\epsilon_{i})$ is independent of $\epsilon_{j}$ whenever $j\neq i$, and $\E(\eta^{(1)}_{i})=\E(\epsilon_{i})=0$, $\E(M_{r,ij}\eta^{(1)}_{i}\epsilon_{j}\epsilon_{k})$ can be nonzero only if $i=j=k$. When $i=j=k$, however, $M_{r,ii}=0$ because $\diag(M_{r})=0$. Hence $\E(\eta^{(1)\prime}_{r}M_{r}\epsilon_{r}\epsilon^{\prime}_{r}\mathbf{1}_{r})=0$. Moreover, $\E\eta^{(1)}_{r}=0$ implies $\E\left(\eta^{(1)\prime}_{r}M_{r}\varsigma_{r}\right)=0$. Therefore $\E\left(\eta^{(1)\prime}_{r}M_{r}\eta^{(2)}_{r}\right)=0$.
\item For $\eta^{(3)}_{r}$, 
\begin{align*}
\E\left(\eta^{(1)\prime}_{r}M_{r}\eta^{(3)}_{r}\right) & =\E\left[\eta^{(1)\prime}_{r}M_{r}\mathbf{1}_{r}\epsilon^{\prime}_{r}(J_{r}-I_{r})\epsilon_{r}\right]/[n_{r}(n_{r}-1)]\\
 & =\sum_{i,j,k,l}\E\left[M_{r,ij}(J_{r}-I_{r})_{kl}\eta^{(1)}_{i}\epsilon_{k}\epsilon_{l}\right]/[n_{r}(n_{r}-1)].
\end{align*}
Because $(\eta^{(1)}_{i},\epsilon_{i})$ is independent of $\epsilon_{k}$ whenever $k\neq i$, and $\E(\eta^{(1)}_{i})=\E(\epsilon_{i})=0$, $\E\left[(J_{r}-I_{r})_{kl}\eta^{(1)}_{i}\epsilon_{k}\epsilon_{l}\right]$ can be nonzero only if $i=k=l$. When $i=k=l$, however, $(J_{r}-I_{r})_{kl}=0$. Hence $\E\left(\eta^{(1)\prime}_{r}M_{r}\eta^{(3)}_{r}\right)=0$.
\item For $\eta^{(2)}_{r}$, using the symmetry of $M_{r}$, 
\begin{align*}
\eta^{(2)\prime}_{r}M_{r}\eta^{(2)}_{r} & =\left(\epsilon_{r}\epsilon^{\prime}_{r}\mathbf{1}_{r}-\varsigma_{r}\right)^{\prime}M_{r}\left(\epsilon_{r}\epsilon^{\prime}_{r}\mathbf{1}_{r}-\varsigma_{r}\right)/n^{2}_{r}\\
 & =\left(\mathbf{1}^{\prime}_{r}\epsilon_{r}\epsilon^{\prime}_{r}M_{r}\epsilon_{r}\epsilon^{\prime}_{r}\mathbf{1}_{r}-2\times\mathbf{1}^{\prime}_{r}\epsilon_{r}\epsilon^{\prime}_{r}M_{r}\varsigma_{r}+\varsigma^{\prime}_{r}M_{r}\varsigma_{r}\right)/n^{2}_{r}.
\end{align*}
Since $\mathbf{1}^{\prime}_{r}\epsilon_{r}\epsilon^{\prime}_{r}M_{r}\epsilon_{r}\epsilon^{\prime}_{r}\mathbf{1}_{r}=\epsilon^{\prime}_{r}\mathbf{1}_{r}\mathbf{1}^{\prime}_{r}\epsilon_{r}\epsilon^{\prime}_{r}M_{r}\epsilon_{r}=\left(\epsilon^{\prime}_{r}J_{r}\epsilon_{r}\right)\left(\epsilon^{\prime}_{r}M_{r}\epsilon_{r}\right)$, and $\diag(M_{r})=0$, Remark \ref{rem:vars}(i) gives $\E\left(\epsilon^{\prime}_{r}J_{r}\epsilon_{r}\right)\left(\epsilon^{\prime}_{r}M_{r}\epsilon_{r}\right)=2\tr(\Omega_{r}J_{r}\Omega_{r}M_{r})=2\mathbf{1}^{\prime}_{r}\Omega_{r}M_{r}\Omega_{r}\mathbf{1}_{r}=2\varsigma^{\prime}_{r}M_{r}\varsigma_{r}$. Moreover, $\E\left(\mathbf{1}^{\prime}_{r}\epsilon_{r}\epsilon^{\prime}_{r}M_{r}\varsigma_{r}\right)=\mathbf{1}^{\prime}_{r}\Omega_{r}M_{r}\varsigma_{r}=\varsigma^{\prime}_{r}M_{r}\varsigma_{r}$. Hence $\ensuremath{\E\left[\left(\epsilon_{r}\epsilon^{\prime}_{r}\mathbf{1}_{r}-\varsigma_{r}\right)^{\prime}M_{r}\left(\epsilon_{r}\epsilon^{\prime}_{r}\mathbf{1}_{r}-\varsigma_{r}\right)\right]/n^{2}_{r}=\left(2\varsigma^{\prime}_{r}M_{r}\varsigma_{r}-2\varsigma^{\prime}_{r}M_{r}\varsigma_{r}+\varsigma^{\prime}_{r}M_{r}\varsigma_{r}\right)/n^{2}_{r}=\varsigma^{\prime}_{r}M_{r}\varsigma_{r}/n^{2}_{r}}.$
\item We have
\begin{align*}
 & n^{2}_{r}(n_{r}-1)\eta^{(2)\prime}_{r}M_{r}\eta^{(3)}_{r}=\left(\epsilon_{r}\epsilon^{\prime}_{r}\mathbf{1}_{r}-\varsigma_{r}\right)^{\prime}M_{r}\mathbf{1}_{r}\epsilon^{\prime}_{r}(J_{r}-I_{r})\epsilon_{r}\\
= & \mathbf{1}^{\prime}_{r}\epsilon_{r}\epsilon^{\prime}_{r}M_{r}\mathbf{1}_{r}\epsilon^{\prime}_{r}(J_{r}-I_{r})\epsilon_{r}-\varsigma^{\prime}_{r}M_{r}\mathbf{1}_{r}\epsilon^{\prime}_{r}(J_{r}-I_{r})\epsilon_{r}\\
= & \left(\epsilon^{\prime}_{r}M_{r}\mathbf{1}_{r}\mathbf{1}^{\prime}_{r}\epsilon_{r}\right)\left[\epsilon^{\prime}_{r}(J_{r}-I_{r})\epsilon_{r}\right]-\left(\varsigma^{\prime}_{r}M_{r}\mathbf{1}_{r}\right)\left[\epsilon^{\prime}_{r}(J_{r}-I_{r})\epsilon_{r}\right]\\
= & \left(\epsilon^{\prime}_{r}M_{r}J_{r}\epsilon_{r}\right)\left[\epsilon^{\prime}_{r}(J_{r}-I_{r})\epsilon_{r}\right]-\left(\varsigma^{\prime}_{r}M_{r}\mathbf{1}_{r}\right)\left[\epsilon^{\prime}_{r}(J_{r}-I_{r})\epsilon_{r}\right].
\end{align*}
Since $\diag(J_{r}-I_{r})=0$, Remark \ref{rem:vars}(i) and (\ref{eq:gamma_quadratic}) imply
\begin{align*}
 & \E\left(\epsilon^{\prime}_{r}M_{r}J_{r}\epsilon_{r}\right)\left[\epsilon^{\prime}_{r}(J_{r}-I_{r})\epsilon_{r}\right]\\
= & 2\tr\left[J_{r}M_{r}\Omega_{r}(J_{r}-I_{r})\Omega_{r}\right]=2\varsigma^{\prime}_{r}\left[\left(J_{r}M_{r}\right)\odot(J_{r}-I_{r})\right]\varsigma_{r}\\
= & 2\varsigma^{\prime}_{r}\left[J_{r}M_{r}-\Diag(J_{r}M_{r})\right]\varsigma_{r}=\varsigma^{\prime}_{r}\left[J_{r}M_{r}+M_{r}J_{r}-\Diag(J_{r}M_{r})-\Diag(M_{r}J_{r})\right]\varsigma_{r},
\end{align*}
where the last equality uses the symmetry of $M_{r}$. For the second term, $\E\left[\epsilon^{\prime}_{r}(J_{r}-I_{r})\epsilon_{r}\right]=0$. Therefore,
\[
\E\eta^{(2)\prime}_{r}M_{r}\eta^{(3)}_{r}=\frac{1}{n^{2}_{r}\left(n_{r}-1\right)}\varsigma^{\prime}_{r}\left[J_{r}M_{r}+M_{r}J_{r}-\Diag(J_{r}M_{r})-\Diag(M_{r}J_{r})\right]\varsigma_{r}.
\]
\item Since $J_{r}-I_{r}$ has zero diagonals, Remark \ref{rem:vars}(i) and (\ref{eq:gamma_quadratic}) give
\begin{align*}
n^{2}_{r}\left(n_{r}-1\right)^{2}\E\left(\eta^{(3)\prime}_{r}M_{r}\eta^{(3)}_{r}\right) & =\E\left[\epsilon^{\prime}_{r}(J_{r}-I_{r})\epsilon_{r}\left(\mathbf{1}^{\prime}_{r}M_{r}\mathbf{1}_{r}\right)\epsilon^{\prime}_{r}(J_{r}-I_{r})\epsilon_{r}\right]\\
 & =2\left(\mathbf{1}^{\prime}_{r}M_{r}\mathbf{1}_{r}\right)\tr\left[\Omega_{r}(J_{r}-I_{r})\Omega_{r}(J_{r}-I_{r})\right]\\
 & =2\left(\mathbf{1}^{\prime}_{r}M_{r}\mathbf{1}_{r}\right)\varsigma^{\prime}_{r}\left[(J_{r}-I_{r})\odot(J_{r}-I_{r})\right]\varsigma_{r}.
\end{align*}
Because the elements of $J_{r}-I_{r}$ are either $0$ or $1$, $(J_{r}-I_{r})\odot(J_{r}-I_{r})=J_{r}-I_{r}$. Thus,
\[
\E\left(\eta^{(3)\prime}_{r}M_{r}\eta^{(3)}_{r}\right)=2\frac{\left(\mathbf{1}^{\prime}_{r}M_{r}\mathbf{1}_{r}\right)}{n^{2}_{r}\left(n_{r}-1\right)^{2}}\varsigma^{\prime}_{r}(J_{r}-I_{r})\varsigma_{r}.
\]
\end{enumerate}

\subsection{Calculation for Proposition \ref{prop:VC_convergence}}\label{subsec:Proof-of-Proposition_g(x)}

By (\ref{eq:pIqJcI}) and Lemma \ref{lem:calculation}(i), $M_{r}J_{r}=\mathcal{D}_{r}J_{r}$, where $\mathcal{D}_{r}=\diag^{G_{r}}_{g=1}\left\{ q_{gr}I_{gr}\right\} $. Therefore,
\begin{align*}
M_{r}J_{r}-\Diag(M_{r}J_{r}) & =\mathcal{D}_{r}(J_{r}-I_{r}),\\
J_{r}M_{r}-\Diag(J_{r}M_{r}) & =(J_{r}-I_{r})\mathcal{D}_{r},
\end{align*}
 and $\mathbf{1}^{\prime}_{r}M_{r}\mathbf{1}_{r}=\tr(M_{r}J_{r})=\tr(\mathcal{D}_{r})$. Substituting these identities into (\ref{eq:f(x)}) gives
\begin{align*}
 & (n_{r}-2)^{2}f(M_{r})\\
= & \left[4+(n_{r}-2)^{2}\right]M_{r}+\frac{2\tr(\mathcal{D}_{r})}{(n_{r}-1)^{2}}(J_{r}-I_{r})-\frac{4}{(n_{r}-1)}\left[\mathcal{D}_{r}(J_{r}-I_{r})+(J_{r}-I_{r})\mathcal{D}_{r}\right]\\
= & \left[4+(n_{r}-2)^{2}\right]M_{r}-\frac{4}{n_{r}-1}\left[\left(\mathcal{D}_{r}-\frac{\tr(\mathcal{D}_{r})}{4\left(n_{r}-1\right)}I_{r}\right)(J_{r}-I_{r})+(J_{r}-I_{r})\left(\mathcal{D}_{r}-\frac{\tr(\mathcal{D}_{r})}{4\left(n_{r}-1\right)}I_{r}\right)\right].
\end{align*}

Define
\begin{eqnarray}
a_{r} & = & \frac{(n_{r}-2)^{2}}{\left[4+(n_{r}-2)^{2}\right]},\label{eq:a_r}
\end{eqnarray}
and let $\mathcal{\ddot{D}}_{r}=\mathcal{D}_{r}-\frac{\tr(\mathcal{D}_{r})}{4\left(n_{r}-1\right)}I_{r}$ and 
\begin{equation}
T_{r}=\mathcal{\ddot{D}}_{r}(J_{r}-I_{r})+(J_{r}-I_{r})\mathcal{\ddot{D}}_{r}.\label{eq:T_r}
\end{equation}
Then
\begin{align}
a_{r}f(M_{r}) & =M_{r}-a_{r}\frac{4}{\left(n_{r}-1\right)(n_{r}-2)^{2}}T_{r}.\label{eq:fx_1}
\end{align}

Since $\mathcal{\ddot{D}}_{r}$ is diagonal, $J_{r}\mathcal{\ddot{D}}_{r}J_{r}=\mathbf{1}_{r}\left(\mathbf{1}^{\prime}_{r}\mathcal{\ddot{D}}_{r}\mathbf{1}_{r}\right)\mathbf{1}^{\prime}_{r}=\tr(\mathcal{\ddot{D}}_{r})J_{r}$. Consequently,
\begin{align*}
T_{r}J_{r} & =(n_{r}-1)\mathcal{\ddot{D}}_{r}J_{r}+\left[\tr(\mathcal{\ddot{D}}_{r})I_{r}-\mathcal{\ddot{D}}_{r}\right]J_{r}=\left[(n_{r}-2)\mathcal{\ddot{D}}_{r}+\tr(\mathcal{\ddot{D}}_{r})I_{r}\right]J_{r},\\
\mathbf{1}^{\prime}_{r}T_{r}\mathbf{1}_{r} & =\tr(T_{r}J_{r})=2\left(n_{r}-1\right)\tr(\mathcal{\ddot{D}}_{r}).
\end{align*}
Applying (\ref{eq:f(x)}) to $T_{r}$ gives
\begin{align*}
(n_{r}-2)^{2}f(T_{r}) & =\left[4+(n_{r}-2)^{2}\right]T_{r}+\frac{4\tr(\mathcal{\ddot{D}}_{r})}{(n_{r}-1)}(J_{r}-I_{r})-\frac{4}{n_{r}-1}\left\{ (n_{r}-2)T_{r}+2\tr(\mathcal{\ddot{D}}_{r})(J_{r}-I_{r})\right\} \\
 & =\left[(n_{r}-2)^{2}+\frac{4}{n_{r}-1}\right]T_{r}-\frac{4\tr(\mathcal{\ddot{D}}_{r})}{(n_{r}-1)}(J_{r}-I_{r}),
\end{align*}
and hence
\begin{align*}
a_{r}f(T_{r}) & =\frac{a_{r}}{(n_{r}-2)^{2}}\left[(n_{r}-2)^{2}+\frac{4}{n_{r}-1}\right]T_{r}-\frac{4a_{r}\tr(\mathcal{\ddot{D}}_{r})}{(n_{r}-1)(n_{r}-2)^{2}}(J_{r}-I_{r})
\end{align*}

Define
\begin{equation}
b_{r}=\frac{4}{(n_{r}-1)}\frac{1}{\left[(n_{r}-2)^{2}+\frac{4}{n_{r}-1}\right]}=\frac{4}{(n_{r}-2)^{2}(n_{r}-1)+4}.\label{eq:b_r}
\end{equation}
Then 
\begin{equation}
a_{r}b_{r}f(T_{r})=\frac{a_{r}}{(n_{r}-2)^{2}}\frac{4}{(n_{r}-1)}T_{r}-\frac{4a_{r}b_{r}\tr(\mathcal{\ddot{D}}_{r})}{(n_{r}-1)(n_{r}-2)^{2}}(J_{r}-I_{r}).\label{eq:fx_2}
\end{equation}

The preceding expression for $f(T_{r})$ holds for any diagonal $\mathcal{\ddot{D}}_{r}$. Setting $\mathcal{\ddot{D}}_{r}=\frac{1}{2}I_{r}$ gives $T_{r}=J_{r}-I_{r}$ and hence
\begin{align*}
(n_{r}-2)^{2}f(J_{r}-I_{r}) & =\left[(n_{r}-2)^{2}+\frac{4}{n_{r}-1}\right](J_{r}-I_{r})-\frac{2n_{r}}{n_{r}-1}(J_{r}-I_{r})\\
 & =\left[(n_{r}-2)^{2}-\frac{2(n_{r}-2)}{n_{r}-1}\right](J_{r}-I_{r}).\\
f(J_{r}-I_{r}) & =\frac{n_{r}(n_{r}-3)}{(n_{r}-2)(n_{r}-1)}(J_{r}-I_{r}).
\end{align*}

Moreover,
\[
\tr(\mathcal{\ddot{D}}_{r})=\tr(\mathcal{D}_{r})-\tr(\mathcal{D}_{r})\frac{n_{r}}{4(n_{r}-1)}=\frac{3n_{r}-4}{4(n_{r}-1)}\tr(\mathcal{D}_{r}).
\]
The coefficient $t_{r}$ in (\ref{eq:g_M}) is therefore
\begin{align}
t_{r} & =\frac{(n_{r}-2)(n_{r}-1)}{n_{r}(n_{r}-3)}\frac{4a_{r}b_{r}\tr(\mathcal{\ddot{D}}_{r})}{(n_{r}-1)(n_{r}-2)^{2}}\nonumber \\
 & =\frac{(n_{r}-2)(n_{r}-1)}{n_{r}(n_{r}-3)}\frac{4}{(n_{r}-1)(n_{r}-2)^{2}}\frac{(n_{r}-2)^{2}}{\left[4+(n_{r}-2)^{2}\right]}\frac{4}{(n_{r}-2)^{2}(n_{r}-1)+4}\frac{3n_{r}-4}{4(n_{r}-1)}\tr(\mathcal{D}_{r})\nonumber \\
 & =\frac{4(n_{r}-2)(3n_{r}-4)\tr(\mathcal{D}_{r})}{n_{r}(n_{r}-1)(n_{r}-3)\left[4+(n_{r}-2)^{2}\right]\left[(n_{r}-2)^{2}(n_{r}-1)+4\right]}.\label{eq:t_r}
\end{align}
Then 
\begin{equation}
t_{r}f(J_{r}-I_{r})=\frac{4b_{r}\tr(\mathcal{\ddot{D}}_{r})}{\left[4+(n_{r}-2)^{2}\right]\left(n_{r}-1\right)}(J_{r}-I_{r}).\label{eq:fx_3}
\end{equation}

Combining (\ref{eq:fx_1}), (\ref{eq:fx_2}), and (\ref{eq:fx_3}) gives $a_{r}f(M_{r})+a_{r}b_{r}f(T_{r})+t_{r}f(J_{r}-I_{r})=M_{r}$. Since $f$ is linear and the definition of $g(M_{r})$ in (\ref{eq:g_M}) is equivalently $g(M_{r})=a_{r}M_{r}+a_{r}b_{r}T_{r}+t_{r}(J_{r}-I_{r})$, it follows that $f(g(M_{r}))=M_{r}$.

\section{Review of Tests of Random Assignment to Peer Groups}\label{sec:Review of Tests}

This section reviews existing tests of random assignment to peer groups. We first summarize the main procedures and their key features, and then provide a detailed theoretical review of each test.

Testing random peer-group assignment within urns has received considerable attention \citep{sacerdote_peer_2001,guryan_peer_2009,wang_peer_2010,stevenson_tests_2015,jochmans_testing_2023,caeyers_exclusion_2024}. This literature tests the implication of random assignment through the null hypothesis $H_{0}:\lambda_{0}=0$ in the model
\begin{equation}
Y_{igr}=\alpha_{r}+\lambda_{0}\bar{Y}_{(-i)gr}+\epsilon_{igr},\label{eq:base}
\end{equation}
where $Y_{igr}$ is a predetermined characteristic of individual $i$ in peer group $g$ of urn $r$, $\alpha_{r}$ is an urn fixed effect, $\bar{Y}_{(-i)gr}$ is the average of $Y$ among $i$'s peers, and $\epsilon_{igr}$ is an idiosyncratic term. Under random assignment within urns, predetermined characteristics are uncorrelated among peers conditional on urn fixed effects, implying $\lambda_{0}=0$.

Several procedures have been proposed for testing this implication, each with different strengths and limitations. \citet{sacerdote_peer_2001} estimates (\ref{eq:base}) by OLS and tests $H_{0}:\lambda_{0}=0$ using the conventional $t$ statistic. The OLS estimator is biased, however, and the test can therefore have distorted size. This bias was first noted by \citet{guryan_peer_2009} and subsequently studied by \citet{wang_peer_2010}, \citet{stevenson_tests_2015}, \citet{jochmans_testing_2023}, and \citet{caeyers_exclusion_2024}. \citet{guryan_peer_2009} retain Sacerdote's OLS $t$ statistic but add the urn-level leave-out mean of $Y$ as a control. This GKN test is widely used in empirical work (see, e.g., \citealp{amodio_input_2018,feld_understanding_2016,golsteyn_impact_2021,hampole_peer_2026,radbruch_interview_2025}). Under the conditions stated below, it has asymptotically correct size. It nevertheless requires variation in urn size to avoid a perfect-fit degeneracy and can have low power when such variation is limited \citep{stevenson_tests_2015,jochmans_testing_2023}.

Other procedures take different approaches. The $F$ test of \citet{wang_peer_2010} uses the joint significance of peer-group dummies after controlling for urn dummies. The test has power only against $\lambda_{0}>0$ \citep{stevenson_tests_2015}. \citet{stevenson_tests_2015} repeatedly splits each peer group, estimates a split-sample regression, and averages the estimates. We establish asymptotically correct size and consistency against fixed alternatives for the corresponding single-split statistic. Averaging reduces split-specific variation but increases computational cost. \citet{caeyers_exclusion_2024} instead estimate $\lambda_{0}$ from the variance-covariance structure under independent and homoskedastic innovations. The test of \citet{jochmans_testing_2023}, by contrast, permits heteroskedastic innovations and has formally established asymptotically correct size and consistency when the number of urns diverges and urn sizes remain uniformly bounded. It has been adopted in recent empirical applications, including \citet{carneiro_effect_2025} and \citet{radbruch_interview_2025}.

\subsection{\citet{sacerdote_peer_2001}}

In studying peer effects among college roommates, \citet{sacerdote_peer_2001} uses data from Dartmouth College, where first-year students are randomly assigned to dormitories and roommates within blocks defined by gender and room-matching characteristics. Table II reports OLS regressions corresponding to (\ref{eq:scalar}), where $Y_{igr}$ is a predetermined student characteristic and $\alpha_{r}$ is a block fixed effect. Statistical insignificance of $\lambda$ is interpreted as evidence of random room assignment conditional on blocks. However, the OLS estimator is generally biased, as discussed by \citet{guryan_peer_2009}, \citet{stevenson_tests_2015}, \citet{jochmans_testing_2023}, and \citet{caeyers_exclusion_2024}. The result below characterizes its asymptotic bias and formalizes the point in \citet{jochmans_testing_2023} that the bias arises from the within-urn demeaning used to absorb urn fixed effects.

Let $\hat{\lambda}^{ols}$ be the OLS estimator of $\lambda$ in (\ref{eq:scalar}).
\begin{prop}
\label{prop:sacerdote} Suppose Assumptions \ref{assu:lambda} to \ref{assu:epsilon} hold. When $\lambda_{0}=0$, as $n\rightarrow\infty$,
\begin{equation}
\hat{\lambda}^{ols}-\frac{\sum^{R}_{r=1}\tr\left[I^{\ast}_{r}W_{r}\Omega_{r}\right]}{\sum^{R}_{r=1}\tr\left[W^{\prime}_{r}I^{\ast}_{r}W_{r}\Omega_{r}\right]}\xrightarrow{p}0.\label{eq:sacerdote_general}
\end{equation}
For $W$ defined in (\ref{eq:W}), (\ref{eq:sacerdote_general}) is equivalent to
\begin{equation}
\hat{\lambda}^{ols}+\frac{\sum^{R}_{r=1}\frac{1}{n_{r}}\tr(\Omega_{r})}{\sum^{R}_{r=1}\left[\sum^{G_{r}}_{g=1}\left(\frac{1}{m_{gr}-1}-\frac{1}{n_{r}}\right)\tr(\Omega_{gr})\right]}\xrightarrow{p}0.\label{eq:sacerdote_wbar}
\end{equation}
Under homoskedasticity, (\ref{eq:sacerdote_wbar}) becomes
\[
\hat{\lambda}^{ols}+\frac{1}{\frac{1}{R}\sum^{R}_{r=1}\sum^{G_{r}}_{g=1}\frac{m_{gr}}{m_{gr}-1}-1}\xrightarrow{p}0.
\]
If, further, $m_{gr}=m$ for all $g$ and $r$, so that the average urn size is $\bar{n}=n/R=mG/R$, then
\begin{equation}
\hat{\lambda}^{ols}+\frac{m-1}{\bar{n}-(m-1)}\xrightarrow{p}0.\label{eq:sacerdote_equalsize}
\end{equation}
\end{prop}
The proof is provided at the end of this section.

Equation (\ref{eq:sacerdote_general}) shows that the asymptotic OLS bias under the null is induced by the within-urn demeaning used to absorb urn fixed effects, a point first noted by \citet{jochmans_testing_2023}. To see this, suppose instead that the true model in (\ref{eq:SAR_urn}) has no urn fixed effects, so that $\alpha_{r}=0$. Under $\lambda_{0}=0$, one regresses $Y$ directly on $WY$ without demeaning. The argument used below to establish (\ref{eq:sacerdote_general}), with $I^{\ast}_{r}$ omitted, gives
\[
\hat{\lambda}^{ols}-\frac{\sum^{R}_{r=1}\tr\left[W_{r}\Omega_{r}\right]}{\sum^{R}_{r=1}\tr\left[W^{\prime}_{r}W_{r}\Omega_{r}\right]}\xrightarrow{p}0.
\]
Since $\diag(W_{r})=0$ and $\Omega_{r}$ is diagonal, we have $\tr(W_{r}\Omega_{r})=0$, and hence $\hat{\lambda}^{ols}\xrightarrow{p}0$. Therefore, without urn fixed effects, the OLS estimator is consistent under the null.

Equation (\ref{eq:sacerdote_equalsize}) agrees with the result in \citet{jochmans_testing_2023}.\footnote{His $m$ denotes the number of peers, whereas our $m$ denotes group size and is therefore one larger.} When $R=1$ and $G\rightarrow\infty$, it corresponds to equation (4.2) of \citet{caeyers_exclusion_2024}.
\begin{proof}
By the Frisch--Waugh--Lovell theorem, the OLS estimator of $\lambda$ in (\ref{eq:SAR_urn}) is equivalent to the OLS estimator based on (\ref{eq:SAR_urn_star}). Thus, 
\[
\hat{\lambda}^{ols}=\frac{\sum_{r}Y^{\prime}_{r}I^{\ast}_{r}W_{r}Y_{r}}{\sum_{r}Y^{\prime}_{r}W^{\prime}_{r}I^{\ast}_{r}W_{r}Y_{r}}.
\]
When $\lambda_{0}=0$, $Y_{r}=\alpha_{r}\mathbf{1}_{r}+\epsilon_{r}$ and $W_{r}Y_{r}=\alpha_{r}\mathbf{1}_{r}+W_{r}\epsilon_{r}$, since $W_{r}\mathbf{1}_{r}=\mathbf{1}_{r}$. Thus
\[
\hat{\lambda}^{ols}=\frac{\sum_{r}\epsilon^{\prime}_{r}I^{\ast}_{r}W_{r}\epsilon_{r}}{\sum_{r}\epsilon^{\prime}_{r}W^{\prime}_{r}I^{\ast}_{r}W_{r}\epsilon_{r}}=\frac{\frac{1}{n}\epsilon^{\prime}I^{\ast}W\epsilon}{\frac{1}{n}\epsilon^{\prime}W^{\prime}I^{\ast}W\epsilon}.
\]
By Remark \ref{rem:rub1}(ii) and Lemma \ref{lem:rub}, $I^{\ast}W$ and $W^{\prime}I^{\ast}W$ are UBRC. Lemma \ref{lem:convergence_ip} therefore gives $\frac{1}{n}[\epsilon^{\prime}I^{\ast}W\epsilon-\tr(I^{\ast}W\Omega)]\xrightarrow{p}0$ and $\frac{1}{n}\left[\epsilon^{\prime}W^{\prime}I^{\ast}W\epsilon-\tr\left(W^{\prime}I^{\ast}W\Omega\right)\right]\xrightarrow{p}0$ as $n\rightarrow\infty$. As shown below, the diagonal elements of $W^{\prime}_{r}I^{\ast}_{r}W_{r}$ are $\frac{1}{m_{gr}-1}-\frac{1}{n_{r}}$. Hence $\tr(W^{\prime}_{r}I^{\ast}_{r}W_{r})=\sum^{G_{r}}_{g=1}\left(\frac{1}{m_{gr}-1}-\frac{1}{n_{r}}\right)m_{gr}\geqslant G_{r}-1$. Assumptions \ref{assu:size} and \ref{assu:epsilon} therefore imply
\begin{align*}
\frac{1}{n}\tr\left(W^{\prime}I^{\ast}W\Omega\right) & \geqslant\frac{c_{\sigma}}{n}\sum^{R}_{r=1}\tr(W^{\prime}_{r}I^{\ast}_{r}W_{r})\geqslant\frac{c_{\sigma}}{n}\sum^{R}_{r=1}(G_{r}-1)=\frac{c_{\sigma}(G-R)}{n}\geqslant\frac{c_{\sigma}}{2C_{m}}>0.
\end{align*}
Since $\tr(I^{\ast}W\Omega)/n$ is bounded, equation (\ref{eq:sacerdote_general}) follows.

Using (\ref{eq:W}) and Lemma \ref{lem:calculation}(ii), 
\[
I^{\ast}_{r}W_{r}=\diag\{\left[\left(-\frac{1}{m_{gr}-1}\right)-1\right]I^{\ast}_{gr}\}+I^{\ast}_{r},
\]
whose diagonal elements are
\[
-\frac{m_{gr}}{m_{gr}-1}(1-\frac{1}{m_{gr}})+1-\frac{1}{n_{r}}=-\frac{1}{n_{r}},
\]
and
\[
W^{\prime}_{r}I^{\ast}_{r}W_{r}=I^{\ast}_{r}W^{2}_{r}=\diag\{\left[\left(-\frac{1}{m_{gr}-1}\right)^{2}-1\right]I^{\ast}_{gr}\}+I^{\ast}_{r}
\]
has diagonal elements
\begin{align*}
\left[\left(\frac{1}{m_{gr}-1}\right)^{2}-1\right](1-\frac{1}{m_{gr}})+1-\frac{1}{n_{r}} & =\frac{m_{gr}}{m_{gr}-1}(\frac{1}{m_{gr}-1}-1)\frac{m_{gr}-1}{m_{gr}}+1-\frac{1}{n_{r}}=\frac{1}{m_{gr}-1}-\frac{1}{n_{r}}.
\end{align*}
Therefore, (\ref{eq:sacerdote_wbar}) holds. Under homoskedasticity across groups and urns, 
\[
\frac{\sum^{R}_{r=1}\frac{1}{n_{r}}\tr(\Omega_{r})}{\sum^{R}_{r=1}\left[\sum^{G_{r}}_{g=1}\left(\frac{1}{m_{gr}-1}-\frac{1}{n_{r}}\right)\tr(\Omega_{gr})\right]}=\frac{\sum^{R}_{r=1}1}{\sum^{R}_{r=1}\sum^{G_{r}}_{g=1}\left(\frac{1}{m_{gr}-1}-\frac{1}{n_{r}}\right)m_{gr}}=\frac{R}{\sum^{R}_{r=1}\sum^{G_{r}}_{g=1}\left(\frac{m_{gr}}{m_{gr}-1}\right)-R}.
\]
Thus, $\hat{\lambda}^{ols}+\frac{1}{\frac{1}{R}\sum^{R}_{r=1}\sum^{G_{r}}_{g=1}\frac{m_{gr}}{m_{gr}-1}-1}\xrightarrow{p}0.$

Finally, suppose $m_{gr}=m$ for all $g$, $r$. Since $G=\sum^{R}_{r=1}G_{r}$, the total sample size is $n=mG$, and the average urn size is $\bar{n}=n/R=mG/R$. Therefore,
\[
\frac{R}{\sum^{R}_{r=1}\sum^{G_{r}}_{g=1}\left(\frac{m_{gr}}{m_{gr}-1}\right)-R}=\frac{1}{\frac{1}{R}G\frac{m}{m-1}-1}=\frac{m-1}{Gm/R-(m-1)}=\frac{m-1}{\bar{n}-(m-1)}.
\]
Hence $\hat{\lambda}^{ols}+\frac{m-1}{\bar{n}-(m-1)}\xrightarrow{p}0$.
\end{proof}

\subsection{\citet{guryan_peer_2009}}

\citet{guryan_peer_2009} note the bias in the test used by \citet{sacerdote_peer_2001} and recommend correcting it by adding the urn-level leave-out mean of the outcome as a control. Subsequent work shows that the test of $\lambda^{GKN}=0$ can have low power when urn-size variation is limited \citep{stevenson_tests_2015,jochmans_testing_2023}. The proposition below establishes that the test has asymptotically correct size under the null and provides a sufficient condition under which the power of the usual two-sided test converges to one against fixed alternatives.

\citet{guryan_peer_2009} test random assignment within urns using the $t$ test of $H_{0}:\lambda^{GKN}=0$ in the regression
\begin{equation}
y_{igr}=\alpha_{r}+\lambda^{GKN}\bar{y}_{(-i)gr}+\phi\bar{y}_{(-i)r}+u_{igr},\label{eq:guryan}
\end{equation}
where $\bar{y}_{(-i)r}$ is the average of $y$ in urn $r$, excluding individual $i$. The matrix form of (\ref{eq:guryan}) is
\begin{align*}
Y_{r} & =\alpha_{r}\mathbf{1}_{r}+\lambda^{GKN}W_{r}Y_{r}+\phi\frac{J_{r}-I_{r}}{n_{r}-1}Y_{r}+u_{r},
\end{align*}
where all variables are defined as in Section \ref{sec:Motivating-Example}. By the Frisch--Waugh--Lovell theorem, the same OLS estimator of $\lambda^{GKN}$ is obtained after partialling out the urn fixed effects:
\begin{align}
I^{\ast}_{r}Y_{r} & =\lambda^{GKN}I^{\ast}_{r}W_{r}Y_{r}-\frac{\phi}{n_{r}-1}I^{\ast}_{r}Y_{r}+I^{\ast}_{r}u_{r}.\label{eq:gkn_demeaned}
\end{align}
Equation (\ref{eq:gkn_demeaned}) shows why urn-size variation is essential. If $n_{r}=\bar{n}$ for every $r$, the second regressor equals $-I^{\ast}_{r}Y_{r}/(\bar{n}-1)$ and is therefore proportional to the dependent variable. If the regressor matrix has full column rank, the OLS estimates are $\hat{\lambda}^{GKN}=0$ and $\hat{\phi}=-(\bar{n}-1)$, with standard errors equal to zero. Thus the associated $t$ statistic is not defined. Variation in $n_{r}$ is therefore necessary to avoid this perfect-fit degeneracy. The following proposition formalizes this result and gives the size and power conclusions when urn sizes vary.
\begin{prop}
\label{prop:GKN} Suppose Assumptions \ref{assu:lambda}--\ref{assu:epsilon} hold and $n_{r}$ is uniformly bounded. Let $t^{GKN}_{\lambda}=\hat{\lambda}^{GKN}/\hat{se}(\hat{\lambda}^{GKN})$, where $\hat{\lambda}^{GKN}$ is the OLS estimator of $\lambda^{GKN}$ in (\ref{eq:gkn_demeaned}) and $\hat{se}(\hat{\lambda}^{GKN})$ is its urn-clustered standard error. Then:
\begin{enumerate}
\item If $n_{r}=\bar{n}$ for every $r$ and the regressor matrix in (\ref{eq:gkn_demeaned}) has full column rank, $\hat{\lambda}^{GKN}=\hat{se}(\hat{\lambda}^{GKN})=0$, and $t^{GKN}_{\lambda}$ is not defined.
\item Suppose $n_{r}$ varies across $r$. Let $\tilde{X}^{\ast}_{r}=I^{\ast}_{r}[W_{r}Y_{r},-\frac{1}{n_{r}-1}Y_{r}]$, and suppose $\varkappa_{X}=\lim_{R\rightarrow\infty}\frac{1}{R}\sum^{R}_{r=1}\E\left(\tilde{X}^{\ast\prime}_{r}\tilde{X}^{\ast}_{r}\right)$ exists, is finite and nonsingular. Suppose the variance limit $V_{\lambda}$ defined in the proof exists and satisfies $0<V_{\lambda}<\infty$, and suppose $R\left[\hat{se}(\hat{\lambda}^{GKN})\right]^{2}\xrightarrow{p}V_{\lambda}$. Then:
\begin{enumerate}
\item Under $H_{0}:\lambda_{0}=0$, $t^{GKN}_{\lambda}\xrightarrow{d}N(0,1)$ as $R\rightarrow\infty$.
\item Under $H_{1}:\lambda_{0}\neq0$, let $a_{r}=Y^{\prime}_{r}I^{\ast}_{r}Y_{r}$, $b_{r}=Y^{\prime}_{r}W^{\prime}_{r}I^{\ast}_{r}Y_{r}$ and define
\begin{equation}
\Gamma=\frac{1}{R^{2}}\sum^{R}_{r=1}\sum^{R}_{s=1}\left[\left(\frac{1}{n_{r}-1}-\frac{1}{n_{s}-1}\right)\frac{\E a_{r}}{(n_{r}-1)}\E b_{s}\right].\label{eq:power_gkn}
\end{equation}
 If $\liminf_{R\rightarrow\infty}|\Gamma|>0$, then $|t^{GKN}_{\lambda}|\xrightarrow{p}\infty$, and the power of the usual two-sided test converges to one. If $\Gamma\rightarrow0$, then $\hat{\lambda}^{GKN}\xrightarrow{p}0$ even when $\lambda_{0}\neq0$.
\end{enumerate}
\end{enumerate}
\end{prop}
Part (i) shows that variation in urn size is necessary for the GKN test to be well defined. Part (ii)(a) establishes asymptotically correct size under the null. Part (ii)(b) gives $\liminf_{R\rightarrow\infty}|\Gamma|>0$ as a sufficient condition for consistency against fixed alternatives. Variation in $n_{r}$ is necessary for this condition, because $\Gamma=0$ when $n_{r}$ is constant, but it is not sufficient.\footnote{To interpret the condition, let $d_{r}=1/(n_{r}-1)$. Antisymmetrizing the double sum gives 
\[
\Gamma=\frac{1}{2R^{2}}\sum^{R}_{r=1}\sum^{R}_{s=1}(d_{r}-d_{s})\left[d_{r}\E a_{r}\E b_{s}-d_{s}\E a_{s}\E b_{r}\right].
\]
Thus, the sufficient condition for consistency in part (ii)(b) requires this weighted antisymmetric sum to remain bounded away from zero. Under $H_{0}:\lambda_{0}=0$, $d_{r}\E a_{r}=-\E b_{r}$ for every $r$, so each bracket vanishes and $\Gamma=0$. Under fixed alternatives, variation in urn size, group sizes, and $\Omega_{r}$ may make the sum nonzero, but variation in these quantities alone does not guarantee that $\Gamma$ is bounded away from zero.}
\begin{proof}
Part (i) follows from the discussion preceding the proposition. For notational simplicity, omit the superscript $GKN$ throughout the proof. Let $\tilde{X}^{\ast}_{r}=I^{\ast}_{r}[W_{r}Y_{r},-\frac{1}{n_{r}-1}Y_{r}]$ and recall that $Y^{\ast}_{r}=I^{\ast}_{r}Y_{r}$. The OLS estimator of (\ref{eq:gkn_demeaned}) is
\[
\hat{\psi}=\left(\begin{array}{c}
\hat{\lambda}\\
\hat{\phi}
\end{array}\right)=\left(\sum^{R}_{r=1}\tilde{X}^{\ast\prime}_{r}\tilde{X}^{\ast}_{r}\right)^{-1}\sum^{R}_{r=1}\left(\tilde{X}^{\ast\prime}_{r}Y^{\ast}_{r}\right).
\]
Define the pseudo-true value
\[
\bar{\psi}=\left(\begin{array}{c}
\bar{\lambda}\\
\bar{\phi}
\end{array}\right)=\left[\sum^{R}_{r=1}\E\left(\tilde{X}^{\ast\prime}_{r}\tilde{X}^{\ast}_{r}\right)\right]^{-1}\left[\sum^{R}_{r=1}\E\left(\tilde{X}^{\ast\prime}_{r}Y^{\ast}_{r}\right)\right].
\]
Its dependence on $R$ is suppressed in the notation.

We first show that $\hat{\psi}-\bar{\psi}\xrightarrow{p}0$. Define $a_{r}=Y^{\prime}_{r}I^{\ast}_{r}Y_{r}$, $b_{r}=Y^{\prime}_{r}W^{\prime}_{r}I^{\ast}_{r}Y_{r}$, and $c_{r}=Y^{\prime}_{r}W^{\prime}_{r}I^{\ast}_{r}W_{r}Y_{r}$. Then
\begin{align*}
\tilde{X}^{\ast\prime}_{r}\tilde{X}^{\ast}_{r} & =\left(\begin{array}{cc}
Y^{\prime}_{r}W^{\prime}_{r}I^{\ast}_{r}W_{r}Y_{r} & -\frac{1}{n_{r}-1}Y^{\prime}_{r}W^{\prime}_{r}I^{\ast}_{r}Y_{r}\\
-\frac{1}{n_{r}-1}Y^{\prime}_{r}W^{\prime}_{r}I^{\ast}_{r}Y_{r} & \frac{1}{(n_{r}-1)^{2}}Y^{\prime}_{r}I^{\ast}_{r}Y_{r}
\end{array}\right)=\left(\begin{array}{cc}
c_{r} & -\frac{1}{n_{r}-1}b_{r}\\
-\frac{1}{n_{r}-1}b_{r} & \frac{1}{(n_{r}-1)^{2}}a_{r}
\end{array}\right),\\
\tilde{X}^{\ast\prime}_{r}Y^{\ast}_{r} & =\left(\begin{array}{c}
Y^{\prime}_{r}W^{\prime}_{r}I^{\ast}_{r}Y_{r}\\
-\frac{1}{n_{r}-1}Y^{\prime}_{r}I^{\ast}_{r}Y_{r}
\end{array}\right)=\left(\begin{array}{c}
b_{r}\\
-\frac{1}{n_{r}-1}a_{r}
\end{array}\right).
\end{align*}
From (\ref{eq:SAR_urn_star}), $I^{\ast}_{r}Y_{r}=I^{\ast}_{r}(I_{r}-\lambda_{0}W_{r})^{-1}\epsilon_{r}$, and hence $a_{r}=\epsilon^{\prime}_{r}(I_{r}-\lambda_{0}W^{\prime}_{r})^{-1}I^{\ast}_{r}(I_{r}-\lambda_{0}W_{r})^{-1}\epsilon_{r}$. Analogous quadratic-form representations hold for $b_{r}$ and $c_{r}$. By Remark \ref{rem:rub1}(ii) and Lemma \ref{lem:rub}, the coefficient matrices of these three quadratic forms are UBRC. Since $n_{r}$ is uniformly bounded, $n/R$ is uniformly bounded. Lemma \ref{lem:convergence_ip} therefore gives $R^{-1}\sum_{r}\left(a_{r}-\E a_{r}\right)\xrightarrow{p}0$, $R^{-1}\sum_{r}\left(b_{r}-\E b_{r}\right)\xrightarrow{p}0$, and $R^{-1}\sum_{r}\left(c_{r}-\E c_{r}\right)\xrightarrow{p}0$. It follows that $\frac{1}{R}\sum^{R}_{r=1}\tilde{X}^{\ast\prime}_{r}\tilde{X}^{\ast}_{r}-\frac{1}{R}\sum^{R}_{r=1}\E\left(\tilde{X}^{\ast\prime}_{r}\tilde{X}^{\ast}_{r}\right)\xrightarrow{p}0$, and $\frac{1}{R}\sum^{R}_{r=1}\tilde{X}^{\ast\prime}_{r}Y^{\ast}_{r}-\frac{1}{R}\sum^{R}_{r=1}\E\left(\tilde{X}^{\ast\prime}_{r}Y^{\ast}_{r}\right)\xrightarrow{p}0$. The corresponding $R^{-1}$ normalized expectations are uniformly bounded. Since $\varkappa_{X}$ is nonsingular, $\frac{1}{R}\sum^{R}_{r=1}\E\left(\tilde{X}^{\ast\prime}_{r}\tilde{X}^{\ast}_{r}\right)$ is invertible with a uniformly bounded inverse for all sufficiently large $R$. Hence $\hat{\psi}-\bar{\psi}\xrightarrow{p}0$ and $\hat{\lambda}-\bar{\lambda}\xrightarrow{p}0$.

Next, we derive $\bar{\lambda}$ and show that $\bar{\lambda}=0$ when $\lambda_{0}=0$. Using the inverse formula for a $2\times2$ matrix, 
\begin{align*}
\bar{\psi} & =\left\{ \det\left[\sum^{R}_{r=1}\E\left(\tilde{X}^{\ast\prime}_{r}\tilde{X}^{\ast}_{r}\right)\right]\right\} ^{-1}\left(\begin{array}{cc}
\sum^{R}_{r=1}\frac{1}{(n_{r}-1)^{2}}\E a_{r} & \sum^{R}_{r=1}\frac{1}{n_{r}-1}\E b_{r}\\
\sum^{R}_{r=1}\frac{1}{n_{r}-1}\E b_{r} & \sum^{R}_{r=1}\E c_{r}
\end{array}\right)\left(\begin{array}{c}
\sum^{R}_{r=1}\E b_{r}\\
-\sum^{R}_{r=1}\frac{1}{n_{r}-1}\E a_{r}
\end{array}\right),\\
\bar{\lambda} & =(1,0)\bar{\psi}=\left\{ \det\left[\frac{1}{R}\sum^{R}_{r=1}\E\left(\tilde{X}^{\ast\prime}_{r}\tilde{X}^{\ast}_{r}\right)\right]\right\} ^{-1}\Gamma,
\end{align*}
where $\Gamma$ is defined in (\ref{eq:power_gkn}). Because $\varkappa_{X}$ is finite and nonsingular, there are constants $c$ and $C$ such that $0<c\leqslant\det\left[\frac{1}{R}\sum^{R}_{r=1}\E\left(\tilde{X}^{\ast\prime}_{r}\tilde{X}^{\ast}_{r}\right)\right]\leqslant C<\infty$ for sufficiently large $R$. If $\liminf_{R\rightarrow\infty}|\Gamma|>0$, then $\liminf_{R\rightarrow\infty}|\bar{\lambda}|>0$. Since $\hat{\lambda}-\bar{\lambda}\xrightarrow{p}0$ and the assumed standard-error condition implies $\hat{se}(\hat{\lambda})=O_{p}(R^{-1/2})$, it follows that $|t^{GKN}_{\lambda}|\xrightarrow{p}\infty$. If $\Gamma\rightarrow0$, the determinant bound implies that $\bar{\lambda}\rightarrow0$, and hence $\hat{\lambda}\xrightarrow{p}0$. This proves part (ii)(b).

To evaluate $\bar{\lambda}$ under the null, observe that by Lemma \ref{lem:calculation},
\begin{align*}
\E a_{r} & =\tr\left[(I_{r}-\lambda_{0}W_{r})^{-1}I^{\ast}_{r}(I_{r}-\lambda_{0}W_{r})^{-1}\Omega_{r}\right]\\
 & =\sum_{g}\left[(\frac{m_{gr}-1}{m_{gr}-1+\lambda_{0}})^{2}(1-\frac{1}{m_{gr}})+\frac{1}{(1-\lambda_{0})^{2}}(\frac{1}{m_{gr}}-\frac{1}{n_{r}})\right]\tr(\Omega_{gr}),
\end{align*}
\begin{align*}
\E b_{r} & =\tr\left[(I_{r}-\lambda_{0}W_{r})^{-1}I^{\ast}_{r}(I_{r}-\lambda_{0}W_{r})^{-1}W_{r}\Omega_{r}\right]\\
 & =\sum_{g}\left[-\frac{1}{m_{gr}}(\frac{m_{gr}-1}{m_{gr}-1+\lambda_{0}})^{2}+\frac{1}{(1-\lambda_{0})^{2}}(\frac{1}{m_{gr}}-\frac{1}{n_{r}})\right]\tr(\Omega_{gr}).
\end{align*}
When $\lambda_{0}=0$, $\E a_{r}=(1-\frac{1}{n_{r}})\tr(\Omega_{r})$ and $\E b_{r}=-\frac{1}{n_{r}}\tr(\Omega_{r})$, so that $\frac{\E a_{r}}{n_{r}-1}=-\E b_{r}$. Substitution into (\ref{eq:power_gkn}) gives 
\[
\Gamma=-\frac{1}{R^{2}}\sum^{R}_{r=1}\sum^{R}_{s=1}\left[\left(\frac{1}{n_{r}-1}-\frac{1}{n_{s}-1}\right)\E b_{r}\E b_{s}\right].
\]
This expression is zero by antisymmetry because interchanging $r$ and $s$ changes the sign of the first factor and leaves $\E b_{r}\E b_{s}$ unchanged. Therefore, $\Gamma=0$ and $\bar{\lambda}=0$ under $H_{0}$.

It remains to establish the limiting distribution under $H_{0}$. Define $\xi^{\ast}_{r}=\tilde{X}^{\ast\prime}_{r}Y^{\ast}_{r}-\tilde{X}^{\ast\prime}_{r}\tilde{X}^{\ast}_{r}\bar{\psi}$. By construction, $\sum^{R}_{r=1}\E\xi^{\ast}_{r}=0$. Moreover, 
\[
\sqrt{R}(\hat{\psi}-\bar{\psi})=\left(\frac{1}{R}\sum^{R}_{r=1}\tilde{X}^{\ast\prime}_{r}\tilde{X}^{\ast}_{r}\right)^{-1}\frac{1}{\sqrt{R}}\sum^{R}_{r=1}\xi^{\ast}_{r}.
\]
Let $e_{1}=(1,0)^{\prime}$ and define 
\[
V_{\lambda}=\lim_{R\rightarrow\infty}e^{\prime}_{1}\varkappa^{-1}_{X}\left[\frac{1}{R}\sum^{R}_{r=1}\Var(\xi^{\ast}_{r})\right]\varkappa^{-1}_{X}e_{1}.
\]
Under $H_{0}$, $\bar{\lambda}=0$, $Y^{\ast}_{r}=I^{\ast}_{r}\epsilon_{r}$, and $\tilde{X}^{\ast}_{r}=\left[I^{\ast}_{r}W_{r}\epsilon_{r},-\frac{1}{n_{r}-1}I^{\ast}_{r}\epsilon_{r}\right]$. The preceding bounds imply that $\bar{\psi}$ is uniformly bounded for all sufficiently large $R$. Hence $\xi^{\ast}_{r}$ is a vector of quadratic forms in $\epsilon_{r}$ with uniformly bounded coefficient matrices. Uniform boundedness of $n_{r}$ and the uniform $4+c_{\epsilon}$ moment condition imply
\[
\max_{1\leq r\leq R}\E\left|e^{\prime}_{1}\varkappa^{-1}_{X}\left[\xi^{\ast}_{r}-\E\left(\xi^{\ast}_{r}\right)\right]\right|^{2+\delta}\leq C<\infty
\]
for some $0<\delta\leqslant c_{\epsilon}/2$ and $C<\infty$, for all sufficiently large $R$. Since $\xi^{\ast}_{r}$ is independent across $r$, this uniform moment bound and $0<V_{\lambda}<\infty$ verify the Lyapunov condition for $e^{\prime}_{1}\varkappa^{-1}_{X}\left[\xi^{\ast}_{r}-\E\left(\xi^{\ast}_{r}\right)\right]$. Hence, since $\sum^{R}_{r=1}\E\xi^{\ast}_{r}=0$, the Lyapunov central limit theorem gives 
\[
e^{\prime}_{1}\varkappa^{-1}_{X}R^{-1/2}\sum^{R}_{r=1}\xi^{\ast}_{r}\xrightarrow{d}N(0,V_{\lambda}).
\]
The same quadratic-form bounds and independence across urns imply $R^{-1/2}\sum^{R}_{r=1}\xi^{\ast}_{r}=O_{p}(1)$. Since $\frac{1}{R}\sum^{R}_{r=1}\tilde{X}^{\ast\prime}_{r}\tilde{X}^{\ast}_{r}\xrightarrow{p}\varkappa_{X},$ Slutsky's theorem gives $\sqrt{R}\hat{\lambda}\xrightarrow{d}N(0,V_{\lambda}).$ The assumed consistency of the urn-clustered standard error and Slutsky's theorem therefore give $t^{GKN}_{\lambda}=\sqrt{R}\hat{\lambda}/(\sqrt{R}\hat{se}(\hat{\lambda}))\xrightarrow{d}N(0,1).$
\end{proof}

\subsection{\citet{stevenson_tests_2015}}\label{subsec:review_stevenson}

\citet{stevenson_tests_2015} uses simulations to study potential limitations of the tests proposed by \citet{sacerdote_peer_2001}, \citet{guryan_peer_2009}, and \citet{wang_peer_2010}. She proposes a testing method that first splits each peer group into two subsamples, with one supplying the dependent variable and the other supplying the peer characteristic. The original analysis evaluates power against alternatives with peer group random effects but does not provide a formal asymptotic justification for the test.

The test proposed by \citet{stevenson_tests_2015} proceeds as follows. We randomly split each peer group $g$ of urn $r$ into two pools (subgroups), with $k^{a}_{gr}$ members in the analysis pool and $k^{b}_{gr}$ members in the remaining peer pool, where $k^{a}_{gr}\geqslant1$, $k^{b}_{gr}\geqslant1$, and $k^{a}_{gr}+k^{b}_{gr}=m_{gr}$.\footnote{\citet{stevenson_tests_2015} sets $k^{a}_{gr}=1$ for all $g,r$.} Let $\bar{Y}^{a}_{gr}$ and $\bar{Y}^{b}_{gr}$ denote the average outcomes in the analysis and peer pools respectively. \citet{stevenson_tests_2015} proposes estimating the regression 
\begin{equation}
\bar{Y}^{b}_{gr}=\lambda^{split}\bar{Y}^{a}_{gr}+\phi_{r}+u_{gr},\label{eq:ybarb_ybara}
\end{equation}
where $\phi_{r}$ is an urn fixed effect, and then testing $H_{0}:\lambda^{split}=0$ for random assignment within urns. To reduce variation induced by sample splitting, \citet{stevenson_tests_2015} suggests repeating the procedure and averaging the estimates. The proposition below establishes asymptotically correct size and consistency against fixed alternatives $\lambda_{0}\neq0$ in (\ref{eq:scalar}) for a single random split, rather than for the averaged estimator obtained from repeated splits.
\begin{prop}
\label{prop:stevenson} Let $\hat{\lambda}^{split}$ be the OLS estimator of (\ref{eq:ybarb_ybara}) and let $\widehat{\Var}(\hat{\lambda}^{split})$ be the urn-clustered variance estimator. Suppose the true model is (\ref{eq:scalar}) and Assumptions \ref{assu:lambda} to \ref{assu:epsilon} hold, and $m_{gr}$ and $G_{r}$ are uniformly bounded with $G_{r}\geqslant2$. Suppose the random splits are drawn mutually independently across peer groups and independently of all innovations. Suppose the limits $V_{ab}$ and $V_{aa}$ defined in the proof exist, with $0<V_{aa}<\infty$. Under $H_{0}$, suppose the limit $\Sigma_{ab}$ defined in the proof exists, with $0<\Sigma_{ab}<\infty$, and suppose $R\widehat{\Var}(\hat{\lambda}^{split})\xrightarrow{p}V^{-2}_{aa}\Sigma_{ab}$.

(i) Under $H_{0}:\lambda_{0}=0$, as $R\rightarrow\infty$,
\[
\frac{\hat{\lambda}^{split}}{\sqrt{\widehat{\Var}(\hat{\lambda}^{split})}}\xrightarrow{d}N(0,1).
\]
(ii) Under $H_{1}:\lambda_{0}\neq0$, as $R\rightarrow\infty$, $\hat{\lambda}^{split}\xrightarrow{p}\bar{\lambda}^{split}=V_{ab}/V_{aa}\neq0$. Moreover, $\left|\hat{\lambda}^{split}/\sqrt{\widehat{\Var}(\hat{\lambda}^{split})}\right|\xrightarrow{p}\infty$, and the power of the usual two-sided test converges to one.
\end{prop}
\begin{proof}
From (\ref{eq:SAR_group}), $Y_{gr}=(I_{gr}-\lambda_{0}W_{gr})^{-1}(\alpha_{r}\mathbf{1}_{gr}+\epsilon_{gr})$. Without loss of generality, let $Y_{gr}=(Y^{a\prime}_{gr},Y^{b\prime}_{gr})^{\prime}$ and $\epsilon_{gr}=(\epsilon^{a\prime}_{gr},\epsilon^{b\prime}_{gr})^{\prime}$. Observe that $(I_{gr}-\lambda_{0}W_{gr})^{-1}\mathbf{1}_{gr}=\mathbf{1}_{gr}/(1-\lambda_{0})$. Direct calculation gives $(I_{gr}-\lambda_{0}W_{gr})^{-1}=c_{gr}I_{gr}+d_{gr}J_{gr}$, where $c_{gr}=\frac{m_{gr}-1}{m_{gr}-1+\lambda_{0}}$, $d_{gr}=\frac{\lambda_{0}}{\left(1-\lambda_{0}\right)\left(m_{gr}-1+\lambda_{0}\right)}$, $J_{gr}=\mathbf{1}_{gr}\mathbf{1}^{\prime}_{gr}$. Consequently,
\begin{align*}
\left(\begin{array}{c}
Y^{a}_{gr}\\
Y^{b}_{gr}
\end{array}\right) & =\frac{\alpha_{r}}{1-\lambda_{0}}\mathbf{1}_{gr}+c_{gr}\left(\begin{array}{c}
\epsilon^{a}_{gr}\\
\epsilon^{b}_{gr}
\end{array}\right)+d_{gr}\mathbf{1}_{gr}(k^{a}_{gr}\bar{\epsilon}^{a}_{gr}+k^{b}_{gr}\bar{\epsilon}^{b}_{gr}),
\end{align*}
where $\bar{\epsilon}^{a}_{gr}$ and $\bar{\epsilon}^{b}_{gr}$ are the means of $\epsilon_{igr}$ in the analysis and peer pools respectively. As a result, $\bar{Y}^{a}_{gr}=\frac{\alpha_{r}}{1-\lambda_{0}}+\zeta^{a}_{gr}$ and $\bar{Y}^{b}_{gr}=\frac{\alpha_{r}}{1-\lambda_{0}}+\zeta^{b}_{gr}$, where
\begin{align*}
\zeta^{a}_{gr} & =c_{gr}\bar{\epsilon}^{a}_{gr}+d_{gr}\left(k^{a}_{gr}\bar{\epsilon}^{a}_{gr}+k^{b}_{gr}\bar{\epsilon}^{b}_{gr}\right),\\
\zeta^{b}_{gr} & =c_{gr}\bar{\epsilon}^{b}_{gr}+d_{gr}\left(k^{a}_{gr}\bar{\epsilon}^{a}_{gr}+k^{b}_{gr}\bar{\epsilon}^{b}_{gr}\right).
\end{align*}
Let $\bar{\zeta}^{a}_{r}=G^{-1}_{r}\sum^{G_{r}}_{g=1}\zeta^{a}_{gr}$ and $\bar{\zeta}^{b}_{r}=G^{-1}_{r}\sum^{G_{r}}_{g=1}\zeta^{b}_{gr}$. The OLS estimator in (\ref{eq:ybarb_ybara}) is
\begin{align*}
\hat{\lambda}^{split} & =\frac{\sum^{R}_{r=1}\sum^{G_{r}}_{g=1}\left(\zeta^{a}_{gr}-\bar{\zeta}^{a}_{r}\right)\left(\zeta^{b}_{gr}-\bar{\zeta}^{b}_{r}\right)}{\sum^{R}_{r=1}\sum^{G_{r}}_{g=1}\left(\zeta^{a}_{gr}-\bar{\zeta}^{a}_{r}\right)^{2}}=\frac{\sum^{R}_{r=1}\tau^{ab}_{r}}{\sum^{R}_{r=1}\tau^{aa}_{r}},
\end{align*}
where $\tau^{ab}_{r}=\sum^{G_{r}}_{g=1}\left(\zeta^{a}_{gr}-\bar{\zeta}^{a}_{r}\right)\left(\zeta^{b}_{gr}-\bar{\zeta}^{b}_{r}\right)$, and $\tau^{aa}_{r}=\sum^{G_{r}}_{g=1}\left(\zeta^{a}_{gr}-\bar{\zeta}^{a}_{r}\right)^{2}$. All expectations below are taken over both the innovations and the random split.

We first prove part (ii). Under Assumptions \ref{assu:lambda}--\ref{assu:epsilon}, $\E\zeta^{a}_{gr}=\E\zeta^{b}_{gr}=0$. The pairs $\{\zeta^{a}_{gr},\zeta^{b}_{gr}\}$ are independent across $g$ and $r$. Because $m_{gr}$ and $G_{r}$ are uniformly bounded, and the innovations have uniformly bounded fourth moments, there exists some $0<c_{\tau}\leqslant c_{\epsilon}/2$ such that $\E|\tau^{ab}_{r}|^{2+c_{\tau}}$ and $\E|\tau^{aa}_{r}|^{2+c_{\tau}}$ are uniformly bounded. The pairs $\{\tau^{ab}_{r},\tau^{aa}_{r}\}$ are therefore independent across $r$ with uniformly bounded second moments. Hence, 
\begin{align*}
\frac{1}{R}\sum_{r}\tau^{ab}_{r} & \xrightarrow{p}V_{ab}=\lim_{R\rightarrow\infty}\frac{1}{R}\sum_{r}\E\tau^{ab}_{r},\\
\frac{1}{R}\sum_{r}\tau^{aa}_{r} & \xrightarrow{p}V_{aa}=\lim_{R\rightarrow\infty}\frac{1}{R}\sum_{r}\E\tau^{aa}_{r}.
\end{align*}
By assumption $V_{aa}>0$. Thus $\hat{\lambda}^{split}\xrightarrow{p}\bar{\lambda}^{split}=V_{ab}/V_{aa}$. It remains to show that $V_{ab}\neq0$ under the fixed alternative $\lambda_{0}\neq0$. Independence across peer groups gives
\begin{align*}
\E\tau^{ab}_{r} & =\E\sum^{G_{r}}_{g=1}\left(\zeta^{a}_{gr}-\bar{\zeta}^{a}_{r}\right)\left(\zeta^{b}_{gr}-\bar{\zeta}^{b}_{r}\right)=(1-\frac{1}{G_{r}})\sum^{G_{r}}_{g=1}\E\left(\zeta^{a}_{gr}\zeta^{b}_{gr}\right).
\end{align*}
Conditional on the realized split, $\bar{\epsilon}^{a}_{gr}$ and $\bar{\epsilon}^{b}_{gr}$ use disjoint sets of innovations, so $\E\left(\bar{\epsilon}^{a}_{gr}\bar{\epsilon}^{b}_{gr}\mid\text{split}\right)=0$ and hence $\E\left(\bar{\epsilon}^{a}_{gr}\bar{\epsilon}^{b}_{gr}\right)=0$. With $d_{gr}=\frac{\lambda_{0}}{\left(1-\lambda_{0}\right)\left(m_{gr}-1+\lambda_{0}\right)}$, 
\begin{align*}
\E\zeta^{a}_{gr}\zeta^{b}_{gr} & =\left(c_{gr}+d_{gr}k^{a}_{gr}\right)k^{a}_{gr}d_{gr}\E\left(\bar{\epsilon}^{a}_{gr}\right)^{2}+\left(c_{gr}+d_{gr}k^{b}_{gr}\right)k^{b}_{gr}d_{gr}\E\left(\bar{\epsilon}^{b}_{gr}\right)^{2}=\lambda_{0}f_{\zeta,gr},
\end{align*}
where 
\[
f_{\zeta,gr}=\frac{\left(c_{gr}+d_{gr}k^{a}_{gr}\right)k^{a}_{gr}\E\left(\bar{\epsilon}^{a}_{gr}\right)^{2}+\left(c_{gr}+d_{gr}k^{b}_{gr}\right)k^{b}_{gr}\E\left(\bar{\epsilon}^{b}_{gr}\right)^{2}}{\left(1-\lambda_{0}\right)\left(m_{gr}-1+\lambda_{0}\right)}.
\]
Under Assumptions \ref{assu:lambda}--\ref{assu:epsilon}, $1-\lambda_{0}>0$, $m_{gr}-1+\lambda_{0}>0$, $c_{gr}=\frac{m_{gr}-1}{m_{gr}-1+\lambda_{0}}>0$, $k^{a}_{gr}\geqslant1$, and $k^{b}_{gr}\geqslant1$. Also $c_{gr}+d_{gr}k^{a}_{gr}=\frac{\left(m_{gr}-1\right)\left(1-\lambda_{0}\right)+k^{a}_{gr}\lambda_{0}}{\left(1-\lambda_{0}\right)\left(m_{gr}-1+\lambda_{0}\right)}=\frac{\left(m_{gr}-1\right)-\left(m_{gr}-1-k^{a}_{gr}\right)\lambda_{0}}{\left(1-\lambda_{0}\right)\left(m_{gr}-1+\lambda_{0}\right)}>0$ because $0\leqslant m_{gr}-1-k^{a}_{gr}<m_{gr}-1$ and $|\lambda_{0}|<1$. Similarly, $c_{gr}+d_{gr}k^{b}_{gr}>0$. Moreover, by Assumption \ref{assu:epsilon}, $\E\left(\bar{\epsilon}^{a}_{gr}\right)^{2}\geqslant c_{\sigma}/k^{a}_{gr}$ and $\E\left(\bar{\epsilon}^{b}_{gr}\right)^{2}\geqslant c_{\sigma}/k^{b}_{gr}$. Since $\Lambda$ is a compact subset of $(-1,1)$ and $m_{gr}$ is uniformly bounded, the positive coefficients in the numerator of $f_{\zeta,gr}$ are uniformly bounded away from zero and its denominator is uniformly bounded above. Hence, $f_{\zeta,gr}\geqslant c_{f}>0$ for some constant $c_{f}$. As a result, $\E\tau^{ab}_{r}=(1-\frac{1}{G_{r}})\lambda_{0}\sum^{G_{r}}_{g=1}f_{\zeta,gr}$ and 
\[
V_{ab}=\lim_{R\rightarrow\infty}\frac{1}{R}\sum^{R}_{r=1}\E\tau^{ab}_{r}=\lambda_{0}\left(\lim_{R\rightarrow\infty}\frac{1}{R}\sum_{r}(1-\frac{1}{G_{r}})\sum^{G_{r}}_{g=1}f_{\zeta,gr}\right).
\]
Thus, if $\lambda_{0}\neq0$, then $V_{ab}\neq0$ and $\bar{\lambda}^{split}=V_{ab}/V_{aa}\neq0$. Under the fixed alternative, the uniformly bounded second moments established above, $\hat{\lambda}^{split}=O_{p}(1)$, and $R^{-1}\sum_{r}\tau^{aa}_{r}\xrightarrow{p}V_{aa}>0$ imply that the usual urn-clustered variance estimator is $O_{p}(R^{-1})$. Hence $\sqrt{\widehat{\Var}(\hat{\lambda}^{split})}=O_{p}(R^{-1/2})$. Together with $\hat{\lambda}^{split}\xrightarrow{p}\bar{\lambda}^{split}\neq0$, this yields $\left|\hat{\lambda}^{split}/\sqrt{\widehat{\Var}(\hat{\lambda}^{split})}\right|\xrightarrow{p}\infty$, so the power of the usual two-sided test converges to one.

For part (i), under $H_{0}:\lambda_{0}=0$ we have $d_{gr}=0$ and hence $\zeta^{a}_{gr}=\bar{\epsilon}^{a}_{gr}$ and $\zeta^{b}_{gr}=\bar{\epsilon}^{b}_{gr}$. Conditional on the realized split, $\zeta^{a}_{gr}$ and $\zeta^{b}_{gr}$ are independent and have zero means. Thus $\E\tau^{ab}_{r}=0$. The uniformly bounded $(2+c_{\tau})$ moments of $\tau^{ab}_{r}$ established above, together with independence across urns, imply the Lyapunov condition, so $R^{-1/2}\sum^{R}_{r=1}\tau^{ab}_{r}\xrightarrow{d}N(0,\Sigma_{ab})$, where $\Sigma_{ab}=\lim_{R\rightarrow\infty}\frac{1}{R}\sum^{R}_{r=1}\Var\left(\tau^{ab}_{r}\right)$. Consequently $\sqrt{R}\hat{\lambda}^{split}\xrightarrow{d}N(0,V_{\lambda})$, where $V_{\lambda}=V^{-2}_{aa}\Sigma_{ab}$. Since under $H_{0}$ $R\widehat{\Var}(\hat{\lambda}^{split})\xrightarrow{p}V_{\lambda}$ by assumption, part (i) follows from Slutsky's theorem.
\end{proof}

\subsection{\citet{jochmans_testing_2023}}\label{subsec:review_Jochmans}

\citet{jochmans_testing_2023} reviews tests of random assignment to peer groups, quantifies the bias in \citet{sacerdote_peer_2001}, and proposes a bias-corrected LM-type test. The homoskedastic correction motivates our choice of the instrument matrix $A$. The procedure also permits unknown heteroskedasticity, and its asymptotic validity is established as $R\rightarrow\infty$ with uniformly bounded urn sizes.

Below we relate the homoskedastic version of the test to the quadratic moment underlying our scalar GMM estimator. \citet{jochmans_testing_2023} standardizes this moment using an urn-clustered standard error. Our analytic variance estimator corrects the finite-urn bias induced by within-urn demeaning and permits other asymptotic sequences, including a bounded number of urns and a growing number of peer groups within urns. We use his homoskedastic notation for this comparison.\footnote{Our quadratic-moment condition remains valid under heteroskedasticity, although \citet{jochmans_testing_2023} equation (3.8) uses a different correction in the heteroskedastic case.}

\citet{jochmans_testing_2023} uses $g$ to index urns and $x$ for the variable of interest. To match our notation, we replace $x$ with $y$ and $g$ with $r$. Let $y_{r,i}$ denote the predetermined variable of individual $i$ in urn $r$, and let $Y_{r}$, $W_{r}$, and $I^{\ast}_{r}$ be as defined in Section \ref{subsec:Model-without-Covariates} of the main text. Thus, $Y_{r}$ collects the $y_{r,i}$, $W_{r}Y_{r}$ collects the peer averages $\bar{y}_{r,[i]}$, and $I^{\ast}_{r}Y_{r}$ and $I^{\ast}_{r}W_{r}Y_{r}$ collect their urn-demeaned counterparts, whose entries are denoted by $\tilde{y}_{r,i}$ and $\widetilde{\bar{y}}_{r,[i]}$, respectively. Let
\begin{align*}
q^{HO} & =\sum^{R}_{r=1}\sum^{n_{r}}_{i=1}\tilde{y}_{r,i}(\bar{y}_{r,[i]}+\frac{y_{r,i}}{n_{r}-1})=\sum^{R}_{r=1}\left(I^{\ast}_{r}Y_{r}\right)^{\prime}\left(W_{r}Y_{r}+\frac{1}{n_{r}-1}Y_{r}\right)\\
 & =\sum^{R}_{r=1}Y^{\prime}_{r}I^{\ast}_{r}\left(W_{r}+\frac{1}{n_{r}-1}I_{r}\right)Y_{r}
\end{align*}
Note that $A_{r}=W_{r}+\frac{1}{n_{r}-1}I_{r}$ and $I^{\ast}_{r}A_{r}I^{\ast}_{r}=I^{\ast}_{r}A_{r}$ by Remark \ref{rem:calculation}, so $q^{HO}=\sum^{R}_{r=1}Y^{\prime}_{r}A^{\ast}_{r}Y_{r}$. Also, let
\[
s^{HO}=\sqrt{\sum^{R}_{r=1}\left(\sum^{n_{r}}_{i=1}\tilde{y}_{r,i}(\bar{y}_{r,[i]}+\frac{y_{r,i}}{n_{r}-1})\right)^{2}}=\sqrt{\sum^{R}_{r=1}\left(Y^{\prime}_{r}A^{\ast}_{r}Y_{r}\right)^{2}}.
\]

Under $H_{0}:\lambda_{0}=0$, $Y_{r}=\alpha_{r}\mathbf{1}_{r}+\epsilon_{r}$. Since $A^{\ast}_{r}\mathbf{1}_{r}=0$, $Y^{\prime}_{r}A^{\ast}_{r}Y_{r}=\epsilon^{\prime}_{r}A^{\ast}_{r}\epsilon_{r}$. Moreover, $\diag(A^{\ast}_{r})=0$ and $\Omega_{r}$ is diagonal under Assumption \ref{assu:epsilon}, so $\E(Y^{\prime}_{r}A^{\ast}_{r}Y_{r})=\tr(A^{\ast}_{r}\Omega_{r})=0$. The urn-level quadratic forms $Y^{\prime}_{r}A^{\ast}_{r}Y_{r}$ are independent across $r$. Under the additional conditions imposed in Theorem 1 of \citet{jochmans_testing_2023}, the corresponding self-normalized statistic satisfies
\[
t^{HO}=\frac{q^{HO}}{s^{HO}}=\frac{\sum^{R}_{r=1}Y^{\prime}_{r}A^{\ast}_{r}Y_{r}}{\sqrt{\sum^{R}_{r=1}\left(Y^{\prime}_{r}A^{\ast}_{r}Y_{r}\right)^{2}}}\xrightarrow{d}N(0,1)
\]
as $R$ goes to infinity with uniformly bounded urn sizes.

Recall that in the model without covariates, the quadratic moment function for urn $r$ is $h^{q}_{r}(\lambda)=\epsilon^{\ast}_{r}(\lambda)^{\prime}A_{r}\epsilon^{\ast}_{r}(\lambda)$, where $\epsilon^{\ast}_{r}(\lambda)=I^{\ast}_{r}(I_{r}-\lambda W_{r})Y_{r}$. Evaluating this moment at the null value gives $h^{q}_{r}(0)=Y^{\prime}_{r}I^{\ast}_{r}A_{r}I^{\ast}_{r}Y_{r}$. Therefore, for the homoskedastic statistic of \citet{jochmans_testing_2023}, $t^{HO}=\sum^{R}_{r=1}h^{q}_{r}(0)/\sqrt{\sum^{R}_{r=1}h^{q}_{r}(0)^{2}}$. Since $h^{q}_{n}(0)=n^{-1}\sum^{R}_{r=1}h^{q}_{r}(0)$, its numerator coincides, up to the normalization by $n$, with the sample quadratic moment underlying our scalar GMM estimator. This yields an LM-type interpretation of the test, with an urn-clustered standardization.

\subsection{\citet{caeyers_exclusion_2024}}\label{subsec:review_CF2024}

While prior work has focused on testing $\lambda_{0}=0$, \citet{caeyers_exclusion_2024} propose a minimum-distance estimator of $\lambda$ under independent and homoskedastic innovations. The estimation method exploits the variance-covariance structure of the error terms. Specifically, if Assumptions \ref{assu:lambda}--\ref{assu:epsilon} hold, and additionally $\Var(\epsilon_{igr})=\sigma^{2}_{0}$ for all $i,g,r$, then (\ref{eq:SAR_urn_star}) implies $\E(I^{\ast}_{r}Y_{r}Y^{\prime}_{r}I^{\ast}_{r})=\sigma^{2}_{0}I^{\ast}_{r}(I_{r}-\lambda_{0}W_{r})^{-2}I^{\ast}_{r}$. With covariates omitted for simplicity, we consider the following reweighted adaptation of their estimator,\footnote{\citet{caeyers_exclusion_2024} do not include the factor $1/n_{r}$. Under uniformly bounded urn sizes, this factor can be omitted without affecting the consistency argument. We include it so that the consistency argument also permits unbounded urn sizes.}
\[
(\hat{\lambda}^{CF},\hat{\sigma}^{2})=\argmin_{\lambda,\sigma^{2}}Q^{CF}(\lambda,\sigma^{2}),
\]
where 
\[
Q^{CF}(\lambda,\sigma^{2})=\frac{1}{R}\sum^{R}_{r=1}\frac{1}{n_{r}}\left\Vert I^{\ast}_{r}Y_{r}Y^{\prime}_{r}I^{\ast}_{r}-\sigma^{2}I^{\ast}_{r}(I_{r}-\lambda W_{r})^{-2}I^{\ast}_{r}\right\Vert ^{2}_{F},
\]
and for any $n_{r}\times n_{r}$ matrix $M_{r}$, $\left\Vert M_{r}\right\Vert _{F}=\left(\sum^{n_{r}}_{i=1}\sum^{n_{r}}_{j=1}M^{2}_{r,ij}\right)^{1/2}=\left[\tr(M^{\prime}_{r}M_{r})\right]^{1/2}$ denotes the Frobenius norm. \citet{caeyers_exclusion_2024} motivate this estimator through a method-of-moments argument. Proposition \ref{prop:cf} provides a formal consistency proof.
\begin{prop}
\label{prop:cf} Let (\ref{eq:scalar}) be the true model and suppose Assumptions \ref{assu:lambda} to \ref{assu:epsilon} hold. In addition, assume that $\Var(\epsilon_{igr})=\sigma^{2}_{0}$ for all $i,g,r$. Suppose the parameter space for $(\lambda,\sigma^{2})$ is $\Lambda\times[c_{\sigma},C_{\sigma}]$, where $\Lambda$ is defined in \ref{assu:lambda} and $0<c_{\sigma}<C_{\sigma}<\infty$. Suppose further that $\sup_{n\geqslant1}\sup_{g,r}m_{gr}\leqslant\bar{m}<\infty$. Then as $R\rightarrow\infty$, $\hat{\lambda}^{CF}-\lambda_{0}\xrightarrow{p}0$ and $\hat{\sigma}^{2}-\sigma^{2}_{0}\xrightarrow{p}0$.
\end{prop}
The minimum-distance criterion $Q^{CF}(\lambda,\sigma^{2})$ is an average of independent urn-level contributions, providing a natural starting point for an asymptotic-normality argument. Establishing such a result would nevertheless require additional regularity conditions and a separate derivation, which we do not pursue.
\begin{proof}
Define $K_{0r}=I^{\ast}_{r}(I_{r}-\lambda_{0}W_{r})^{-1}$, and $K_{r}(\lambda)=I^{\ast}_{r}(I_{r}-\lambda W_{r})^{-1}$. By (\ref{eq:I_lW_inv}), Lemma \ref{lem:calculation}(ii), and (\ref{eq:prod_s}), $K_{0r}$ and $K_{r}(\lambda)$ are symmetric and commute. Since $I^{\ast}_{r}$ and $W_{r}$ commute, (\ref{eq:SAR_urn_star}) gives $I^{\ast}_{r}Y_{r}=K_{0r}\epsilon_{r}$. Therefore,
\begin{align}
Q^{CF}(\lambda,\sigma^{2}) & =\frac{1}{R}\sum^{R}_{r=1}\frac{1}{n_{r}}\left\Vert K_{0r}\epsilon_{r}\epsilon^{\prime}_{r}K_{0r}-\sigma^{2}K^{2}_{r}(\lambda)\right\Vert ^{2}_{F}\nonumber \\
 & =\frac{1}{R}\sum^{R}_{r=1}\frac{1}{n_{r}}\tr\left[K_{0r}\epsilon_{r}\epsilon^{\prime}_{r}K^{2}_{0r}\epsilon_{r}\epsilon^{\prime}_{r}K_{0r}-2\sigma^{2}K^{2}_{r}(\lambda)K_{0r}\epsilon_{r}\epsilon^{\prime}_{r}K_{0r}+\sigma^{4}K^{4}_{r}(\lambda)\right]\nonumber \\
 & =\underbrace{\frac{1}{R}\sum^{R}_{r=1}\frac{1}{n_{r}}\left(\epsilon^{\prime}_{r}K^{2}_{0r}\epsilon_{r}\right)^{2}}_{term\,1}-2\underbrace{\sigma^{2}\frac{1}{R}\sum^{R}_{r=1}\frac{1}{n_{r}}\epsilon^{\prime}_{r}K^{2}_{0r}K^{2}_{r}(\lambda)\epsilon_{r}}_{term\,2}+\underbrace{\frac{1}{R}\sum^{R}_{r=1}\frac{\sigma^{4}}{n_{r}}\tr\left(K^{4}_{r}(\lambda)\right)}_{term\,3}.\label{eq:Q_CF}
\end{align}

Since term 1 in (\ref{eq:Q_CF}) does not depend on $(\lambda,\sigma^{2})$, it does not affect the minimizer. Let $\widetilde{Q}^{CF}(\lambda,\sigma^{2})$ denote $Q^{CF}(\lambda,\sigma^{2})$ with term 1 removed. We first show that $\sup_{(\lambda,\sigma^{2})}\left|\widetilde{Q}^{CF}(\lambda,\sigma^{2})-\E\widetilde{Q}^{CF}(\lambda,\sigma^{2})\right|\xrightarrow{p}0$. Term 3 is nonstochastic and therefore contributes nothing to $\widetilde{Q}^{CF}-\E\widetilde{Q}^{CF}$. For term 2, let $\xi_{r}(\lambda)=\frac{1}{n_{r}}\epsilon^{\prime}_{r}K^{2}_{0r}K^{2}_{r}(\lambda)\epsilon_{r}$. Since $K^{2}_{0r}K^{2}_{r}(\lambda)$ is UBRC uniformly over $\lambda\in\Lambda$, Remark \ref{rem:vars}(iii) implies that $\Var\left(\xi_{r}(\lambda)\right)$ is $O(1/n_{r})$ uniformly. Hence for each fixed $\lambda$, $\Var\left(\frac{1}{R}\sum^{R}_{r=1}\xi_{r}(\lambda)\right)=\frac{1}{R^{2}}\sum_{r}O(1/n_{r})=O(\frac{1}{R})\rightarrow0$ and $\frac{1}{R}\sum^{R}_{r=1}\left(\xi_{r}(\lambda)-\E\xi_{r}(\lambda)\right)\xrightarrow{p}0$. We now verify uniformity in $\lambda$. By (\ref{eq:I_lW_inv}) and Lemma \ref{lem:calculation}(ii),
\[
K^{2}_{r}(\lambda)=\diag^{G_{r}}_{g=1}\{\left[\left(\frac{m_{gr}-1}{m_{gr}-1+\lambda}\right)^{2}-\left(\frac{1}{1-\lambda}\right)^{2}\right]I^{\ast}_{gr}\}+\left(\frac{1}{1-\lambda}\right)^{2}I^{\ast}_{r}.
\]
By the mean value theorem, 
\[
\left|\frac{1}{R}\sum^{R}_{r=1}\xi_{r}(\lambda^{\prime})-\frac{1}{R}\sum^{R}_{r=1}\xi_{r}(\lambda)\right|\leqslant|\lambda^{\prime}-\lambda|\sup_{\lambda\in\Lambda}\left|\frac{1}{R}\sum^{R}_{r=1}\frac{1}{n_{r}}\epsilon^{\prime}_{r}K^{2}_{0r}\frac{dK^{2}_{r}(\lambda)}{d\lambda}\epsilon_{r}\right|.
\]
Since $m_{gr}\geqslant2$ by Assumption \ref{assu:size} and $\Lambda$ is a compact subset of $(-1,1)$ by Assumption \ref{assu:lambda}, $dK^{2}_{r}(\lambda)/d\lambda$ is UBRC uniformly over $\lambda\in\Lambda$. Hence, by Lemma \ref{lem:rub}, $K^{2}_{0r}dK^{2}_{r}(\lambda)/d\lambda$ is also UBRC uniformly over $\lambda\in\Lambda$. Consequently $\sup_{\lambda,r}\left\Vert K^{2}_{0r}\frac{dK^{2}_{r}(\lambda)}{d\lambda}\right\Vert _{2}\leqslant C$ for some constant $C<\infty$ that does not depend on $n$, $r$, or $\lambda$. Therefore, $\sup_{\lambda\in\Lambda}\left|\frac{1}{n_{r}}\epsilon^{\prime}_{r}K^{2}_{0r}\frac{dK^{2}_{r}(\lambda)}{d\lambda}\epsilon_{r}\right|\leqslant\frac{C}{n_{r}}\epsilon^{\prime}_{r}\epsilon_{r}$ and hence
\[
\sup_{\lambda\in\Lambda}\left|\frac{1}{R}\sum^{R}_{r=1}\frac{1}{n_{r}}\epsilon^{\prime}_{r}K^{2}_{0r}\frac{dK^{2}_{r}(\lambda)}{d\lambda}\epsilon_{r}\right|\leqslant\frac{C}{R}\sum^{R}_{r=1}\frac{\epsilon^{\prime}_{r}\epsilon_{r}}{n_{r}}=O_{p}(1),
\]
where the last equality follows because Assumption \ref{assu:epsilon} implies $\E\epsilon^{\prime}_{r}\epsilon_{r}/n_{r}\leqslant C_{\sigma}$ for some constant $C_{\sigma}<\infty$. This establishes stochastic equicontinuity of $\frac{1}{R}\sum^{R}_{r=1}\xi_{r}(\lambda)$. In addition, $\frac{1}{R}\sum^{R}_{r=1}\E(\xi_{r}(\lambda))=\sigma^{2}_{0}\frac{1}{R}\sum^{R}_{r=1}\frac{1}{n_{r}}\tr\left[K^{2}_{0r}K^{2}_{r}(\lambda)\right]$ is uniformly equicontinuous because $K^{2}_{0r}\frac{dK^{2}_{r}(\lambda)}{d\lambda}$ is UBRC uniformly over $\lambda\in\Lambda$ as shown above and hence $\sup_{\lambda\in\Lambda}\left|\frac{d}{d\lambda}\left[\frac{1}{R}\sum^{R}_{r=1}\E(\xi_{r}(\lambda))\right]\right|\leqslant C$. Together with compactness of $\Lambda$, the preceding arguments verify the conditions of Corollary 2.2 of \citet{newey_uniform_1991}. Thus $\sup_{\lambda\in\Lambda}\left|\frac{1}{R}\sum^{R}_{r=1}\left[\xi_{r}(\lambda)-\E\xi_{r}(\lambda)\right]\right|\xrightarrow{p}0.$ Since $\sigma^{2}\in[c_{\sigma},C_{\sigma}]$, uniformity over $\sigma^{2}$ for term 2 follows immediately. Combining the preceding arguments gives $\sup_{(\lambda,\sigma^{2})}\left|\widetilde{Q}^{CF}(\lambda,\sigma^{2})-\E\widetilde{Q}^{CF}(\lambda,\sigma^{2})\right|\xrightarrow{p}0$.

It remains to show identification. Minimizing $\E\tilde{Q}^{CF}(\lambda,\sigma^{2})$ is equivalent to minimizing 
\begin{align*}
\frac{1}{R}\sum^{R}_{r=1}\frac{1}{n_{r}}\left[-2\sigma^{2}\E\left(\epsilon^{\prime}_{r}K^{2}_{0r}K^{2}_{r}\epsilon_{r}\right)+\tr\left(\sigma^{4}K^{4}_{r}\right)\right]= & \frac{1}{R}\sum^{R}_{r=1}\frac{1}{n_{r}}\tr\left(-2\sigma^{2}_{0}\sigma^{2}K^{2}_{0r}K^{2}_{r}+\sigma^{4}K^{4}_{r}\right).
\end{align*}
Since $\sigma^{4}_{0}\tr(K^{4}_{0r})$ is free of $(\lambda,\sigma^{2})$, subtracting the population criterion at $(\lambda_{0},\sigma^{2}_{0})$ yields 
\begin{equation}
\Delta(\lambda,\sigma^{2})=\E\widetilde{Q}^{CF}(\lambda,\sigma^{2})-\E\widetilde{Q}^{CF}(\lambda_{0},\sigma^{2}_{0})=\frac{1}{R}\sum^{R}_{r=1}\frac{1}{n_{r}}\tr\left[\left(\sigma^{2}K^{2}_{r}(\lambda)-\sigma^{2}_{0}K^{2}_{0r}\right)^{2}\right].\label{eq:Eq_CE}
\end{equation}
It remains to show that $\Delta(\lambda,\sigma^{2})$ is uniformly bounded away from zero outside every neighborhood of $(\lambda_{0},\sigma^{2}_{0})$, uniformly in $R$. Observe that 
\begin{align*}
\sigma^{2}_{0}K^{2}_{0r}-\sigma^{2}K^{2}_{r} & (\lambda)=I^{\ast}_{r}\diag\left\{ p_{K,gr}(\lambda,\sigma^{2})I^{\ast}_{gr}+q_{K}(\lambda,\sigma^{2})J^{\ast}_{gr}\right\} ,
\end{align*}
where $p_{K,gr}(\lambda,\sigma^{2})=\left(\frac{m_{gr}-1}{m_{gr}-1+\lambda_{0}}\right)^{2}\sigma^{2}_{0}-\left(\frac{m_{gr}-1}{m_{gr}-1+\lambda}\right)^{2}\sigma^{2}$, and $q_{K}(\lambda,\sigma^{2})=\left(\frac{1}{1-\lambda_{0}}\right)^{2}\sigma^{2}_{0}-\left(\frac{1}{1-\lambda}\right)^{2}\sigma^{2}$. Hence
\[
\left(\sigma^{2}K^{2}_{r}(\lambda)-\sigma^{2}_{0}K^{2}_{0r}\right)^{2}=I^{\ast}_{r}\diag\left\{ p^{2}_{K,gr}I^{\ast}_{gr}+q^{2}_{K}J^{\ast}_{gr}\right\} =\diag\{\left(p^{2}_{K,gr}-q^{2}_{K}\right)I^{\ast}_{gr}\}+q^{2}_{K}I^{\ast}_{r},
\]
\begin{align*}
\frac{1}{n_{r}}\tr\left[\left(\sigma^{2}K^{2}_{r}(\lambda)-\sigma^{2}_{0}K^{2}_{0r}\right)^{2}\right] & =\frac{1}{n_{r}}\left[\sum_{g}\left(p^{2}_{K,gr}-q^{2}_{K}\right)(m_{gr}-1)+q^{2}_{K}(n_{r}-1)\right]\\
 & =\frac{1}{n_{r}}\left[\sum_{g}p^{2}_{K,gr}(m_{gr}-1)+q^{2}_{K}(G_{r}-1)\right].
\end{align*}
The right-hand side is nonnegative. Since $2\leqslant m_{gr}\leqslant\bar{m}$, $\sum_{g}(m_{gr}-1)/n_{r}\geqslant1/2\geqslant1/\bar{m}$. Moreover, $n_{r}\leqslant\bar{m}G_{r}$ and $G_{r}\geqslant2$, so $(G_{r}-1)/n_{r}\geqslant1/(2\bar{m})$. Therefore, 
\[
\frac{1}{n_{r}}\tr\left[\left(\sigma^{2}K^{2}_{r}(\lambda)-\sigma^{2}_{0}K^{2}_{0r}\right)^{2}\right]\geqslant\frac{1}{\bar{m}}\min_{m}p^{2}_{K,m}(\lambda,\sigma^{2})+\frac{1}{2\bar{m}}q^{2}_{K}(\lambda,\sigma^{2}),
\]
where $p_{K,m}(\lambda,\sigma^{2})$ denotes $p_{K,gr}(\lambda,\sigma^{2})$ evaluated at an admissible group size $m$. For every admissible $m$, $p_{K,m}(\lambda,\sigma^{2})=q_{K}(\lambda,\sigma^{2})=0$ implies $\frac{\sigma^{2}_{0}}{\sigma^{2}}=\left(\frac{1-\lambda_{0}}{1-\lambda}\right)^{2}=\left(\frac{m-1+\lambda_{0}}{m-1+\lambda}\right)^{2}$. Since $\Lambda\subset(-1,1)$, all quantities in the ratios are positive, so $\frac{m-1+\lambda_{0}}{1-\lambda_{0}}=\frac{m-1+\lambda}{1-\lambda}$. Equivalently, $\frac{m}{1-\lambda_{0}}-1=\frac{m}{1-\lambda}-1$, which implies $\lambda=\lambda_{0}$. Then $q_{K}(\lambda,\sigma^{2})=0$ implies $\sigma^{2}=\sigma^{2}_{0}$. Hence the preceding lower bound is zero if and only if $(\lambda,\sigma^{2})=(\lambda_{0},\sigma^{2}_{0})$. Because the admissible group sizes form a finite set, the lower bound is continuous in $(\lambda,\sigma^{2})$. Therefore, by compactness, for every neighborhood $\mathcal{N}$ of $(\lambda_{0},\sigma^{2}_{0})$ there exists $\eta_{\mathcal{N}}>0$ such that $\inf_{(\lambda,\sigma^{2})\notin\mathcal{N}}\Delta(\lambda,\sigma^{2})\geqslant\eta_{\mathcal{N}}$ for every $R$. Thus $\E\widetilde{Q}^{CF}(\lambda,\sigma^{2})$ is uniformly separated from its minimum outside every neighborhood of $(\lambda_{0},\sigma^{2}_{0})$ and is uniquely minimized at $(\lambda_{0},\sigma^{2}_{0})$.

The uniform convergence result, the uniform identification result, and compactness of the parameter space establish consistency.
\end{proof}

\subsection{\citet{wang_peer_2010}}

\citet{wang_peer_2010} uses an $F$ test for the joint significance of peer group dummies in
\[
y_{igr}=\alpha_{r}+\phi_{gr}+\epsilon_{igr},
\]
where $\phi_{gr}$ denotes the peer-group fixed effect, and $\alpha_{r}$ denotes the urn fixed effect. \citet{stevenson_tests_2015} provides simulation evidence that the test fails to reject under negative correlation, while \citet{jochmans_testing_2023} questions its size control. The results below show that, under i.i.d. normal innovations, the test has exact size and is consistent against $\lambda_{0}>0$, but its rejection probability converges to zero under fixed alternatives $\lambda_{0}<0$.
\begin{prop}
\label{prop:wang} Let (\ref{eq:scalar}) be the true model and suppose Assumptions \ref{assu:lambda} to \ref{assu:epsilon} hold. In addition, suppose $\epsilon_{igr}\sim N(0,\sigma^{2}_{0})$ independently across $i,g,r$. Consider the usual upper-tail $F$ test for the joint significance of peer-group dummies, comparing the restricted regression with urn dummies only to the unrestricted regression with peer-group dummies. Under $H_{0}:\lambda_{0}=0$, the test has exact finite-sample size. Under any fixed alternative $\lambda_{0}>0$, the test is consistent as $n\rightarrow\infty$. Under any fixed alternative $\lambda_{0}<0$, its rejection probability converges to zero as $n\rightarrow\infty$.
\end{prop}
\begin{proof}
Because each urn dummy is a linear combination of the peer-group dummies within that urn, the unrestricted regression has the same column space as a regression on all peer-group dummies. Its residuals are therefore obtained by demeaning within peer groups. Hence the unrestricted residual sum of squares for urn $r$ is $SSR_{U,r}=Y^{\prime}_{r}\diag\{I^{\ast}_{gr}\}Y_{r}$. The restricted model includes only urn dummies, so its residual sum of squares for urn $r$ is $SSR_{R,r}=Y^{\prime}_{r}I^{\ast}_{r}Y_{r}$.

From (\ref{eq:SAR_urn}) we have $Y_{r}=(I_{r}-\lambda_{0}W_{r})^{-1}(\alpha_{r}\mathbf{1}_{r}+\epsilon_{r})=\frac{\alpha_{r}}{1-\lambda_{0}}\mathbf{1}_{r}+(I_{r}-\lambda_{0}W_{r})^{-1}\epsilon_{r}$, and $\diag\{I^{\ast}_{gr}\}\mathbf{1}_{r}=I^{\ast}_{r}\mathbf{1}_{r}=0$. Using (\ref{eq:I_lW_inv}) and (\ref{eq:prod_s}), we have 
\[
(I_{r}-\lambda_{0}W_{r})^{-2}=\diag\left\{ \left(\frac{m_{gr}-1}{m_{gr}-1+\lambda_{0}}\right)^{2}I^{\ast}_{gr}+\left(\frac{1}{1-\lambda_{0}}\right)^{2}J^{\ast}_{gr}\right\} .
\]
Therefore 
\begin{align*}
SSR_{U,r} & =Y^{\prime}_{r}\diag\{I^{\ast}_{gr}\}Y_{r}=\epsilon^{\prime}_{r}(I_{r}-\lambda_{0}W_{r})^{-1}\diag\{I^{\ast}_{gr}\}(I_{r}-\lambda_{0}W_{r})^{-1}\epsilon_{r}\\
 & =\epsilon^{\prime}_{r}\diag\left\{ \left(\frac{m_{gr}-1}{m_{gr}-1+\lambda_{0}}\right)^{2}I^{\ast}_{gr}\right\} \epsilon_{r}.
\end{align*}
Similarly, by Lemma \ref{lem:calculation}(ii), 
\begin{align*}
SSR_{R,r} & =Y^{\prime}_{r}I^{\ast}_{r}Y_{r}=\epsilon^{\prime}_{r}(I_{r}-\lambda_{0}W_{r})^{-1}I^{\ast}_{r}(I_{r}-\lambda_{0}W_{r})^{-1}\epsilon_{r}\\
 & =\epsilon^{\prime}_{r}\left[\diag\left\{ \left[\left(\frac{m_{gr}-1}{m_{gr}-1+\lambda_{0}}\right)^{2}-\left(\frac{1}{1-\lambda_{0}}\right)^{2}\right]I^{\ast}_{gr}\right\} +\left(\frac{1}{1-\lambda_{0}}\right)^{2}I^{\ast}_{r}\right]\epsilon_{r}.
\end{align*}
It follows that
\[
SSR_{R,r}-SSR_{U,r}=\left(\frac{1}{1-\lambda_{0}}\right)^{2}\epsilon^{\prime}_{r}\left[I^{\ast}_{r}-\diag\{I^{\ast}_{gr}\}\right]\epsilon_{r}.
\]

Because the unrestricted model contains $G$ peer-group dummies and the restricted model contains $R$ urn dummies, the usual $F$ statistic is
\begin{align}
F_{\lambda} & =\frac{\sum^{R}_{r=1}\left(\frac{1}{1-\lambda_{0}}\right)^{2}\epsilon^{\prime}_{r}\left[I^{\ast}_{r}-\diag\{I^{\ast}_{gr}\}\right]\epsilon_{r}/(G-R)}{\sum^{R}_{r=1}\epsilon^{\prime}_{r}\diag\left\{ \left(\frac{m_{gr}-1}{m_{gr}-1+\lambda_{0}}\right)^{2}I^{\ast}_{gr}\right\} \epsilon_{r}/(n-G)}\nonumber \\
 & =\frac{\sum^{R}_{r=1}\epsilon^{\prime}_{r}\left[I^{\ast}_{r}-\diag\{I^{\ast}_{gr}\}\right]\epsilon_{r}/(G-R)}{\sum^{R}_{r=1}\epsilon^{\prime}_{r}\diag\left\{ \left[\frac{(1-\lambda_{0})(m_{gr}-1)}{m_{gr}-1+\lambda_{0}}\right]^{2}I^{\ast}_{gr}\right\} \epsilon_{r}/(n-G)}.\label{eq:F}
\end{align}
Under the null, $\left[\frac{(1-\lambda_{0})(m_{gr}-1)}{m_{gr}-1+\lambda_{0}}\right]^{2}=1$, so
\[
F_{\lambda}=\frac{\sum^{R}_{r=1}\epsilon_{r}'\left[I^{*}_{r}-\diag_{g}\{I^{*}_{gr}\}\right]\epsilon_{r}/(G-R)}{\sum^{R}_{r=1}\epsilon_{r}'\diag_{g}\{I^{*}_{gr}\}\epsilon_{r}/(n-G)}.
\]
The matrices $I^{*}_{r}-\diag_{g}\{I^{*}_{gr}\}=\diag_{g}\{J^{*}_{gr}\}-J^{*}_{r}$ and $\diag_{g}\{I^{*}_{gr}\}$ are orthogonal projection matrices with ranks $G_{r}-1$ and $n_{r}-G_{r}$, respectively. Moreover, $\left[I^{*}_{r}-\diag_{g}\{I^{*}_{gr}\}\right]\diag_{g}\{I^{*}_{gr}\}=0$, so the two projection matrices are mutually orthogonal. Their block-diagonal counterparts across urns therefore have ranks $G-R$ and $n-G$, respectively. Since $\epsilon_{igr}\sim N(0,\sigma^{2}_{0})$ independently across $i,g,r$, the numerator and denominator quadratic forms are independent chi-square variables scaled by $\sigma^{2}_{0}$.\footnote{The standard non-normal large-sample justification for the OLS F-test applies when the number of tested restrictions is fixed. Here, the number of restrictions is $G-R$, which increases with the sample size. Moreover, because the framework permits peer-group sizes to remain bounded, the corresponding projection matrix can have non-negligible diagonal elements. As a result, without Gaussianity, the null distribution of the numerator quadratic form may depend on fourth moments of the innovations, and the conventional $F$-critical values need not yield correct asymptotic size.} Consequently, $F_{\lambda}\sim F_{G-R,n-G}$ under $\lambda_{0}=0$, and the usual $F$-test has exact size.

Now consider a fixed alternative $\lambda_{0}\neq0$. In the denominator of (\ref{eq:F}), $a_{gr}(\lambda_{0})=\left(\frac{(1-\lambda_{0})(m_{gr}-1)}{m_{gr}-1+\lambda_{0}}\right)^{2}=\left(1-\frac{\lambda_{0}m_{gr}}{m_{gr}-1+\lambda_{0}}\right)^{2}$. For fixed $\lambda_{0}\in(-1,1)$, $a_{gr}(\lambda_{0})$ is uniformly bounded over $g,r,n$. Moreover, the matrices $\diag^{R}_{r=1}\left\{ I^{\ast}_{r}-\diag_{g}\{I^{\ast}_{gr}\}\right\} $ and $\diag^{R}_{r=1}\left\{ \diag_{g}\{a_{gr}(\lambda_{0})I^{\ast}_{gr}\}\right\} $ are UBRC. Since $G_{r}\geqslant2$ and $m_{gr}\geqslant2$, $G-R\geqslant G/2$ and $n-G\geqslant G$. Assumption \ref{assu:size} therefore implies that $n/(G-R)$ and $n/(n-G)$ are uniformly bounded and that both $G-R$ and $n-G$ diverge. Applying Lemma \ref{lem:convergence_ip}(i) and using $\tr\left[\diag^{R}_{r=1}\left\{ I^{\ast}_{r}-\diag_{g}\{I^{\ast}_{gr}\}\right\} \right]=G-R$ gives 
\[
\frac{\sum^{R}_{r=1}\epsilon_{r}'\left[I^{*}_{r}-\diag_{g}\{I^{*}_{gr}\}\right]\epsilon_{r}}{G-R}\xrightarrow{p}\sigma^{2}_{0}.
\]
Similarly, since $\tr\left[\diag^{R}_{r=1}\left\{ \diag_{g}\{a_{gr}(\lambda_{0})I^{\ast}_{gr}\}\right\} \right]=\sum^{R}_{r=1}\sum^{G_{r}}_{g=1}a_{gr}(\lambda_{0})(m_{gr}-1),$ Lemma \ref{lem:convergence_ip}(i) gives 
\[
\frac{\sum^{R}_{r=1}\epsilon_{r}'\diag_{g}\{a_{gr}(\lambda_{0})I^{\ast}_{gr}\}\epsilon_{r}}{n-G}-\sigma^{2}_{0}\bar{a}_{n}(\lambda_{0})\xrightarrow{p}0,
\]
where
\[
\bar{a}_{n}(\lambda_{0})=\frac{\sum^{R}_{r=1}\sum^{G_{r}}_{g=1}a_{gr}(\lambda_{0})(m_{gr}-1)}{n-G}.
\]
If $\lambda_{0}>0$, then $0<\frac{\lambda_{0}m_{gr}}{m_{gr}-1+\lambda_{0}}<1$ and $0<a_{gr}(\lambda_{0})<(1-\lambda_{0})^{2}<1.$ Hence $\bar{a}_{n}(\lambda_{0})\leqslant(1-\lambda_{0})^{2}<1$ uniformly in $n$. The preceding convergences therefore imply that there exists $\delta>0$ such that $\Pr(F_{\lambda}>1+\delta)\rightarrow1.$ The usual upper-tail critical value converges to one as both degrees of freedom diverge. Therefore the rejection probability converges to one, so the test is consistent against fixed alternatives $\lambda_{0}>0$.

By contrast, if $\lambda_{0}<0$, then $a_{gr}(\lambda_{0})>(1-\lambda_{0})^{2}>1$. Thus $\bar{a}_{n}(\lambda_{0})\geqslant(1-\lambda_{0})^{2}>1$ uniformly in $n$. The preceding convergences therefore imply that there exists $\delta>0$ such that $\Pr(F_{\lambda}<1-\delta)\rightarrow1$. Since the usual upper-tail critical value converges to one, the rejection probability converges to zero under any fixed $\lambda_{0}<0$.
\end{proof}

\section{Additional Monte Carlo Results}

The additional simulations consider several departures from the baseline design. We report results for $\lambda$ and omit those for $\beta_{1}$ and $\beta_{2}$. As in the main text, the left panel of each table fixes $G_{r}=2$ and increases $R$, whereas the right panel fixes $R=2$ and increases $G_{r}$. The changes to the baseline specification are described below.

Table \ref{tab:mc-nox} considers the model without covariates, so that $\lambda$ is identified by the quadratic moment alone. The finite-sample performance is similar to that in the baseline design: bias is small and the rejection frequencies approach the nominal level as $G$ increases. Table \ref{tab:mc-chisq} instead retains the baseline heteroskedasticity but replaces the normal innovations with
\[
\epsilon_{igr}=\sqrt{m_{gr}+1}\frac{\xi_{igr}-3}{\sqrt{6}},\qquad\xi_{igr}\stackrel{\mathrm{i.i.d.}}{\sim}\chi^{2}_{3}.
\]
The non-normal design produces more pronounced over-rejection at small $G$, particularly when $G_{r}=2$ and the number of urns increases. This distortion declines steadily with $G$, and the rejection frequencies are close to the nominal level in the larger designs.

Tables \ref{tab:mc-m3} and \ref{tab:mc-m46} examine the role of peer-group size. Table \ref{tab:mc-m3} sets $m_{gr}=3$ for all groups and uses homoskedastic $N(0,1)$ innovations.\footnote{In unreported simulations, we retain $m_{gr}\in\{2,4\}$ but replace the baseline heteroskedastic innovations with homoskedastic $N(0,1)$ innovations. The resulting changes are generally smaller than those obtained when group size is additionally fixed at $m_{gr}=3$, particularly for Monte Carlo dispersion and average estimated standard errors. Because the two designs differ along more than one dimension, this comparison is only suggestive about the contribution of group-size variation.} The estimator continues to exhibit decreasing bias and approximately correct inference as $G$ increases, providing a finite-sample illustration of the theoretical result that identification does not require variation in peer-group size. Finite-sample bias and dispersion are nevertheless somewhat larger than in the baseline design when $G$ is small. Table \ref{tab:mc-m46}, which instead increases group sizes to $m_{gr}\in\{4,6\}$, yields a similar pattern: finite-sample bias and dispersion are larger in the smaller designs, but the differences diminish as the number of peer groups increases.

Overall, the additional simulations reinforce the main Monte Carlo findings. Performance improves along both asymptotic sequences, including when identification relies only on the quadratic moment and when all peer groups have the same size. The main finite-sample qualification is that Wald inference can over-reject when the number of peer groups is small, with the distortion most pronounced under non-normal innovations and when the number of groups per urn remains small.

% Generated by simulations/export_mc_tables.R.
% Do not edit this file by hand.

\begin{table}[!htbp]
\centering
\begin{threeparttable}
\caption{Monte Carlo results without covariates}
\label{tab:mc-nox}
\begin{tabular}{llccc@{\hspace{1.2em}}c@{\hspace{1.2em}}ccc}
\hline\hline
 & & \multicolumn{3}{c}{$G_r=2$, growing $R$} & & \multicolumn{3}{c}{$R=2$, growing $G_r$} \\
\cline{3-5}\cline{7-9}
$G$ & Statistic & $\lambda_0=-0.4$ & $\lambda_0=0$ & $\lambda_0=0.4$ & & $\lambda_0=-0.4$ & $\lambda_0=0$ & $\lambda_0=0.4$ \\
\hline
20 & Bias & -0.014 & -0.026 & -0.025 &  & -0.010 & -0.012 & -0.014 \\
 & MC SD & 0.154 & 0.164 & 0.129 &  & 0.134 & 0.132 & 0.102 \\
 & Avg. SE & 0.134 & 0.140 & 0.110 &  & 0.124 & 0.125 & 0.094 \\
 & Rej. freq. & 0.105 & 0.112 & 0.106 &  & 0.079 & 0.073 & 0.076 \\
\noalign{\smallskip}
40 & Bias & -0.008 & -0.011 & -0.012 &  & -0.001 & -0.006 & -0.007 \\
 & MC SD & 0.104 & 0.111 & 0.087 &  & 0.093 & 0.092 & 0.069 \\
 & Avg. SE & 0.098 & 0.103 & 0.081 &  & 0.089 & 0.088 & 0.067 \\
 & Rej. freq. & 0.081 & 0.083 & 0.072 &  & 0.063 & 0.063 & 0.055 \\
\noalign{\smallskip}
100 & Bias & -0.003 & -0.004 & -0.005 &  & -0.001 & -0.002 & -0.003 \\
 & MC SD & 0.065 & 0.069 & 0.054 &  & 0.056 & 0.058 & 0.043 \\
 & Avg. SE & 0.064 & 0.067 & 0.052 &  & 0.056 & 0.056 & 0.042 \\
 & Rej. freq. & 0.060 & 0.062 & 0.062 &  & 0.045 & 0.059 & 0.057 \\
\noalign{\smallskip}
200 & Bias & -0.001 & -0.002 & -0.002 &  & -0.001 & -0.002 & -0.002 \\
 & MC SD & 0.046 & 0.049 & 0.038 &  & 0.040 & 0.040 & 0.030 \\
 & Avg. SE & 0.045 & 0.048 & 0.037 &  & 0.040 & 0.040 & 0.030 \\
 & Rej. freq. & 0.057 & 0.057 & 0.055 &  & 0.053 & 0.054 & 0.047 \\
\noalign{\smallskip}
400 & Bias & -0.001 & -0.001 & -0.001 &  & -0.001 & -0.001 & -0.001 \\
 & MC SD & 0.033 & 0.034 & 0.027 &  & 0.028 & 0.028 & 0.022 \\
 & Avg. SE & 0.032 & 0.034 & 0.027 &  & 0.028 & 0.028 & 0.021 \\
 & Rej. freq. & 0.060 & 0.051 & 0.053 &  & 0.049 & 0.054 & 0.058 \\
\noalign{\smallskip}
1,000 & Bias & -0.000 & -0.000 & -0.001 &  & 0.000 & 0.000 & -0.000 \\
 & MC SD & 0.021 & 0.022 & 0.017 &  & 0.018 & 0.018 & 0.013 \\
 & Avg. SE & 0.021 & 0.022 & 0.017 &  & 0.018 & 0.018 & 0.013 \\
 & Rej. freq. & 0.052 & 0.057 & 0.051 &  & 0.057 & 0.051 & 0.049 \\
\noalign{\hrule height 1pt}
\end{tabular}
\begin{tablenotes}[flushleft]
\footnotesize
\item[] \textit{Notes:} The first column reports the total number of peer groups, $G=\sum_r G_r$. For each value of $G$, the four rows report bias, Monte Carlo standard deviation, average estimated standard error, and the rejection frequency of the nominal 5 percent two-sided Wald test of the corresponding true parameter value. All other features follow the baseline design, except that covariates are omitted from both the data-generating process and estimation. Only $\lambda$ is estimated. Each design cell has 5,000 replications.
\end{tablenotes}
\end{threeparttable}
\end{table}

% Generated by simulations/export_mc_tables.R.
% Do not edit this file by hand.

\begin{table}[!htbp]
\centering
\begin{threeparttable}
\caption{Monte Carlo results with non-normal innovations}
\label{tab:mc-chisq}
\begin{tabular}{llccc@{\hspace{1.2em}}c@{\hspace{1.2em}}ccc}
\hline\hline
 & & \multicolumn{3}{c}{$G_r=2$, growing $R$} & & \multicolumn{3}{c}{$R=2$, growing $G_r$} \\
\cline{3-5}\cline{7-9}
$G$ & Statistic & $\lambda_0=-0.4$ & $\lambda_0=0$ & $\lambda_0=0.4$ & & $\lambda_0=-0.4$ & $\lambda_0=0$ & $\lambda_0=0.4$ \\
\hline
20 & Bias & -0.043 & -0.041 & -0.032 &  & -0.026 & -0.026 & -0.016 \\
 & MC SD & 0.157 & 0.166 & 0.127 &  & 0.135 & 0.133 & 0.099 \\
 & Avg. SE & 0.124 & 0.130 & 0.102 &  & 0.116 & 0.117 & 0.087 \\
 & Rej. freq. & 0.174 & 0.168 & 0.145 &  & 0.127 & 0.107 & 0.101 \\
\noalign{\smallskip}
40 & Bias & -0.024 & -0.023 & -0.015 &  & -0.013 & -0.012 & -0.009 \\
 & MC SD & 0.105 & 0.111 & 0.085 &  & 0.092 & 0.089 & 0.067 \\
 & Avg. SE & 0.092 & 0.096 & 0.075 &  & 0.085 & 0.084 & 0.063 \\
 & Rej. freq. & 0.124 & 0.124 & 0.107 &  & 0.091 & 0.082 & 0.079 \\
\noalign{\smallskip}
100 & Bias & -0.009 & -0.007 & -0.005 &  & -0.005 & -0.005 & -0.004 \\
 & MC SD & 0.066 & 0.067 & 0.052 &  & 0.056 & 0.055 & 0.042 \\
 & Avg. SE & 0.061 & 0.064 & 0.049 &  & 0.055 & 0.054 & 0.040 \\
 & Rej. freq. & 0.087 & 0.074 & 0.075 &  & 0.065 & 0.062 & 0.062 \\
\noalign{\smallskip}
200 & Bias & -0.004 & -0.003 & -0.003 &  & -0.002 & -0.004 & -0.002 \\
 & MC SD & 0.046 & 0.047 & 0.037 &  & 0.040 & 0.039 & 0.030 \\
 & Avg. SE & 0.044 & 0.046 & 0.036 &  & 0.039 & 0.039 & 0.029 \\
 & Rej. freq. & 0.068 & 0.066 & 0.066 &  & 0.067 & 0.060 & 0.059 \\
\noalign{\smallskip}
400 & Bias & -0.002 & -0.002 & -0.001 &  & -0.002 & -0.001 & -0.001 \\
 & MC SD & 0.032 & 0.034 & 0.026 &  & 0.028 & 0.028 & 0.020 \\
 & Avg. SE & 0.031 & 0.033 & 0.025 &  & 0.028 & 0.028 & 0.020 \\
 & Rej. freq. & 0.057 & 0.061 & 0.064 &  & 0.055 & 0.053 & 0.047 \\
\noalign{\smallskip}
1,000 & Bias & -0.001 & -0.001 & -0.000 &  & -0.001 & -0.000 & -0.001 \\
 & MC SD & 0.020 & 0.021 & 0.016 &  & 0.018 & 0.017 & 0.013 \\
 & Avg. SE & 0.020 & 0.021 & 0.016 &  & 0.018 & 0.017 & 0.013 \\
 & Rej. freq. & 0.057 & 0.058 & 0.049 &  & 0.050 & 0.051 & 0.052 \\
\noalign{\hrule height 1pt}
\end{tabular}
\begin{tablenotes}[flushleft]
\footnotesize
\item[] \textit{Notes:} The first column reports the total number of peer groups, $G=\sum_r G_r$. For each value of $G$, the four rows report bias, Monte Carlo standard deviation, average estimated standard error, and the rejection frequency of the nominal 5 percent two-sided Wald test of the corresponding true parameter value. All other features follow the baseline design, except that $\epsilon_{igr}=\sqrt{m_{gr}+1}(\xi_{igr}-3)/\sqrt{6}$, where $\xi_{igr}\sim\chi^2_3$. Only results for $\lambda$ are reported. Each design cell has 5,000 replications.
\end{tablenotes}
\end{threeparttable}
\end{table}

% Generated by simulations/export_mc_tables.R.
% Do not edit this file by hand.

\begin{table}[!htbp]
\centering
\begin{threeparttable}
\caption{Monte Carlo results with $m_{gr}=3$ and homoskedastic innovations}
\label{tab:mc-m3}
\begin{tabular}{llccc@{\hspace{1.2em}}c@{\hspace{1.2em}}ccc}
\hline\hline
 & & \multicolumn{3}{c}{$G_r=2$, growing $R$} & & \multicolumn{3}{c}{$R=2$, growing $G_r$} \\
\cline{3-5}\cline{7-9}
$G$ & Statistic & $\lambda_0=-0.4$ & $\lambda_0=0$ & $\lambda_0=0.4$ & & $\lambda_0=-0.4$ & $\lambda_0=0$ & $\lambda_0=0.4$ \\
\hline
20 & Bias & -0.063 & -0.069 & -0.053 &  & -0.028 & -0.029 & -0.024 \\
 & MC SD & 0.203 & 0.190 & 0.144 &  & 0.156 & 0.142 & 0.106 \\
 & Avg. SE & 0.178 & 0.163 & 0.121 &  & 0.143 & 0.130 & 0.095 \\
 & Rej. freq. & 0.116 & 0.108 & 0.092 &  & 0.088 & 0.083 & 0.082 \\
\noalign{\smallskip}
40 & Bias & -0.028 & -0.030 & -0.026 &  & -0.008 & -0.013 & -0.011 \\
 & MC SD & 0.136 & 0.124 & 0.091 &  & 0.107 & 0.096 & 0.070 \\
 & Avg. SE & 0.128 & 0.116 & 0.085 &  & 0.102 & 0.091 & 0.066 \\
 & Rej. freq. & 0.075 & 0.074 & 0.064 &  & 0.067 & 0.070 & 0.065 \\
\noalign{\smallskip}
100 & Bias & -0.013 & -0.011 & -0.009 &  & -0.004 & -0.004 & -0.003 \\
 & MC SD & 0.085 & 0.076 & 0.055 &  & 0.065 & 0.060 & 0.042 \\
 & Avg. SE & 0.083 & 0.074 & 0.054 &  & 0.064 & 0.058 & 0.042 \\
 & Rej. freq. & 0.066 & 0.060 & 0.055 &  & 0.057 & 0.060 & 0.052 \\
\noalign{\smallskip}
200 & Bias & -0.004 & -0.006 & -0.005 &  & -0.003 & -0.003 & -0.002 \\
 & MC SD & 0.059 & 0.053 & 0.038 &  & 0.046 & 0.042 & 0.029 \\
 & Avg. SE & 0.059 & 0.053 & 0.038 &  & 0.046 & 0.041 & 0.029 \\
 & Rej. freq. & 0.061 & 0.054 & 0.052 &  & 0.052 & 0.058 & 0.057 \\
\noalign{\smallskip}
400 & Bias & -0.003 & -0.003 & -0.003 &  & -0.001 & -0.001 & -0.001 \\
 & MC SD & 0.043 & 0.037 & 0.027 &  & 0.032 & 0.030 & 0.021 \\
 & Avg. SE & 0.042 & 0.037 & 0.027 &  & 0.032 & 0.029 & 0.021 \\
 & Rej. freq. & 0.056 & 0.049 & 0.052 &  & 0.050 & 0.054 & 0.052 \\
\noalign{\smallskip}
1,000 & Bias & -0.001 & -0.001 & -0.001 &  & -0.001 & -0.000 & -0.000 \\
 & MC SD & 0.027 & 0.024 & 0.017 &  & 0.020 & 0.018 & 0.013 \\
 & Avg. SE & 0.026 & 0.024 & 0.017 &  & 0.020 & 0.018 & 0.013 \\
 & Rej. freq. & 0.054 & 0.049 & 0.049 &  & 0.046 & 0.048 & 0.051 \\
\noalign{\hrule height 1pt}
\end{tabular}
\begin{tablenotes}[flushleft]
\footnotesize
\item[] \textit{Notes:} The first column reports the total number of peer groups, $G=\sum_r G_r$. For each value of $G$, the four rows report bias, Monte Carlo standard deviation, average estimated standard error, and the rejection frequency of the nominal 5 percent two-sided Wald test of the corresponding true parameter value. All other features follow the baseline design, except that $m_{gr}=3$ and $\epsilon_{igr}\sim N(0,1)$. Only results for $\lambda$ are reported. Each design cell has 5,000 replications.
\end{tablenotes}
\end{threeparttable}
\end{table}

% Generated by simulations/export_mc_tables.R.
% Do not edit this file by hand.

\begin{table}[!htbp]
\centering
\begin{threeparttable}
\caption{Monte Carlo results with $m_{gr} \in \{4,6\}$}
\label{tab:mc-m46}
\begin{tabular}{llccc@{\hspace{1.2em}}c@{\hspace{1.2em}}ccc}
\hline\hline
 & & \multicolumn{3}{c}{$G_r=2$, growing $R$} & & \multicolumn{3}{c}{$R=2$, growing $G_r$} \\
\cline{3-5}\cline{7-9}
$G$ & Statistic & $\lambda_0=-0.4$ & $\lambda_0=0$ & $\lambda_0=0.4$ & & $\lambda_0=-0.4$ & $\lambda_0=0$ & $\lambda_0=0.4$ \\
\hline
20 & Bias & -0.113 & -0.099 & -0.068 &  & -0.064 & -0.052 & -0.032 \\
 & MC SD & 0.250 & 0.229 & 0.157 &  & 0.198 & 0.164 & 0.107 \\
 & Avg. SE & 0.237 & 0.194 & 0.132 &  & 0.187 & 0.151 & 0.100 \\
 & Rej. freq. & 0.098 & 0.090 & 0.076 &  & 0.076 & 0.072 & 0.058 \\
\noalign{\smallskip}
40 & Bias & -0.056 & -0.049 & -0.031 &  & -0.025 & -0.023 & -0.016 \\
 & MC SD & 0.179 & 0.146 & 0.098 &  & 0.132 & 0.107 & 0.072 \\
 & Avg. SE & 0.166 & 0.135 & 0.091 &  & 0.129 & 0.103 & 0.069 \\
 & Rej. freq. & 0.074 & 0.068 & 0.060 &  & 0.061 & 0.057 & 0.055 \\
\noalign{\smallskip}
100 & Bias & -0.023 & -0.019 & -0.012 &  & -0.010 & -0.007 & -0.006 \\
 & MC SD & 0.107 & 0.086 & 0.057 &  & 0.082 & 0.065 & 0.043 \\
 & Avg. SE & 0.105 & 0.084 & 0.056 &  & 0.080 & 0.064 & 0.043 \\
 & Rej. freq. & 0.059 & 0.053 & 0.051 &  & 0.056 & 0.052 & 0.055 \\
\noalign{\smallskip}
200 & Bias & -0.012 & -0.008 & -0.006 &  & -0.006 & -0.004 & -0.003 \\
 & MC SD & 0.074 & 0.061 & 0.040 &  & 0.057 & 0.045 & 0.031 \\
 & Avg. SE & 0.074 & 0.059 & 0.040 &  & 0.057 & 0.045 & 0.030 \\
 & Rej. freq. & 0.051 & 0.057 & 0.056 &  & 0.055 & 0.049 & 0.053 \\
\noalign{\smallskip}
400 & Bias & -0.005 & -0.004 & -0.003 &  & -0.003 & -0.001 & -0.001 \\
 & MC SD & 0.052 & 0.042 & 0.028 &  & 0.041 & 0.033 & 0.021 \\
 & Avg. SE & 0.052 & 0.042 & 0.028 &  & 0.040 & 0.032 & 0.021 \\
 & Rej. freq. & 0.049 & 0.048 & 0.054 &  & 0.055 & 0.057 & 0.049 \\
\noalign{\smallskip}
1,000 & Bias & -0.002 & -0.002 & -0.001 &  & -0.001 & -0.001 & -0.000 \\
 & MC SD & 0.033 & 0.027 & 0.018 &  & 0.025 & 0.020 & 0.013 \\
 & Avg. SE & 0.033 & 0.027 & 0.018 &  & 0.025 & 0.020 & 0.013 \\
 & Rej. freq. & 0.050 & 0.054 & 0.055 &  & 0.045 & 0.042 & 0.051 \\
\noalign{\hrule height 1pt}
\end{tabular}
\begin{tablenotes}[flushleft]
\footnotesize
\item[] \textit{Notes:} The first column reports the total number of peer groups, $G=\sum_r G_r$. For each value of $G$, the four rows report bias, Monte Carlo standard deviation, average estimated standard error, and the rejection frequency of the nominal 5 percent two-sided Wald test of the corresponding true parameter value. All other features follow the baseline design, except that $m_{gr} \in \{4,6\}$. Only results for $\lambda$ are reported. Each design cell has 5,000 replications.
\end{tablenotes}
\end{threeparttable}
\end{table}

\section{Additional Empirical Application}

We also apply the method to the data in \citet{guryan_peer_2009} to study peer effects among professional golf players. Golfers are randomly assigned to groups of three within tournament-category cells. We use first-round observations rather than pooling multiple rounds, avoiding repeated observations of the same player within a tournament. The data contain approximately 280 urns, defined by tournament-category combinations, and approximately 3,000 peer groups. However, repeated appearances of the same golfer across tournaments are not covered by Assumption \ref{assu:epsilon}. We therefore treat this application primarily as an illustration of the estimator rather than as substantive evidence on peer effects.

We first test the random-assignment implication using
\[
X_{igr}=\alpha_{r}+\lambda\bar{X}_{(-i)gr}+\epsilon_{igr},
\]
where $X$ is a predetermined characteristic of player $i$ in peer group $g$ of urn $r$. Following \citet{guryan_peer_2009}, we use ``handicap'' as the preferred measure of ability and consider driving distance, putts per round, and greens hit in regulation as alternatives. The remaining characteristics are first year on tour, which measures experience, and name length. Table \ref{tab:gkn-skill} reports the estimates. These results are comparable to those in Panel A of Table 3 in \citet{guryan_peer_2009}, although the samples differ and that paper estimates (\ref{eq:guryan}) by OLS. % Generated by application/peer\_cra\_GKN2009.do.
\begin{table}[!htbp]
\centering
\begin{threeparttable}
\setlength{\tabcolsep}{3pt}
\def\sym#1{\ifmmode^{#1}\else\(^{#1}\)\fi}
\caption{Peer effects in pre-determined skill and player characteristics}
\label{tab:gkn-skill}
\begin{tabular}{@{}l*{6}{c}@{}}
\hline\hline
 & (1) & (2) & (3) & (4) & (5) & (6) \\
 & Handicap & \shortstack{Driving\\ distance} & \shortstack{Greens in\\ regulation} & Putts & \shortstack{First\\ year} & \shortstack{Name\\ length} \\
\hline
$\hat{\lambda}$&     -0.0169         &     -0.0106         &     -0.0129         &     -0.0138         &      0.0016         &     -0.0006         \\
            &    (0.0109)         &    (0.0110)         &    (0.0116)         &    (0.0116)         &    (0.0107)         &    (0.0101)         \\
\hline
Observations&       8,254         &       8,113         &       8,113         &       8,113         &       8,113         &       8,113         \\
Tournament-category urns&         288         &         281         &         281         &         281         &         281         &         281         \\
Groups      &       3,010         &       2,964         &       2,964         &       2,964         &       2,964         &       2,964         \\
\noalign{\hrule height 1pt}
\end{tabular}
\begin{tablenotes}[flushleft]
\footnotesize
\item[] \textit{Notes:} Each column reports the GMM estimate of the endogenous peer effect using the indicated pre-determined skill or player characteristic as the outcome. Handicap is estimated on first-round observations; the remaining columns use the skill-analysis sample. Standard errors are in parentheses. \sym{*} \(p<0.10\), \sym{**} \(p<0.05\), \sym{***} \(p<0.01\).
\end{tablenotes}
\end{threeparttable}
\end{table}

For handicap, the preferred measure of ability, $\hat{\lambda}=-0.0169$ with a standard error of $0.0109$. Across all six characteristics, the estimates range from $-0.0169$ to $0.0016$, and none is statistically significant at the 10 percent level. We therefore do not reject the corresponding conditional random-assignment implication $H_{0}:\lambda_{0}=0$ for any of the six characteristics, consistent with \citet{guryan_peer_2009}.

To separate endogenous and contextual peer effects, we estimate
\[
Y_{igr}=\alpha_{r}+\lambda\bar{Y}_{(-i)gr}+X^{\prime}_{igr}\beta_{1}+\bar{X}^{\prime}_{(-i)gr}\beta_{2}+\epsilon_{igr},
\]
where $Y_{igr}$ is the score of player $i$ in peer group $g$ of urn $r$, $X_{igr}$ is a pre-assignment ability measure, and $\bar{X}_{(-i)gr}$ is the corresponding peer average. Column (1) omits $X_{igr}$ and $\bar{X}_{(-i)gr}$. Column (2) uses the preferred ability measure ``handicap'' as $X$. Columns (3)--(5) use driving distance, putts per round, and greens hit in regulation, respectively, as $X$. Column (6) includes all three alternative ability measures in $X$. Table \ref{tab:gkn-score} reports the results. Columns (2)--(6) can be compared with Columns (1)--(5) of Table 5 in \citet{guryan_peer_2009}, where $\lambda\bar{Y}_{(-i)grt}$ is omitted and the corresponding reduced forms are estimated by OLS.\footnote{Table 5 in \citet{guryan_peer_2009} uses data from rounds 1 and 2 and includes round fixed effects. To avoid dependence in scores across rounds, the reported specifications use first-round data only. Results using data from rounds 1 and 2 while controlling for round fixed effects are similar.}

Column (1) of Table 7 in \citet{guryan_peer_2009} reports an OLS estimate of $\lambda=0.055$, controlling for individual ability. This estimate is subject to the reflection problem discussed above. Under the maintained model, our estimates of $\lambda$ range from $0.0168$ to $0.0187$ across the six specifications. Five estimates are statistically significant at the 10 percent level, but none is significant at the 5 percent level. In the full specification in Column (6), the 95 percent confidence interval for $\lambda$ is $[-0.0024,0.0382]$, ruling out endogenous peer effects of $0.04$ or greater at the 5 percent level. The contextual peer-effect estimates are also small in magnitude and statistically insignificant throughout. Taken together, the results provide only limited evidence of peer effects on first-round golf scores, consistent with the conclusion in \citet{guryan_peer_2009} that peers' ability has no detectable effect on the performance of golfers.

% Generated by application/peer\_cra\_GKN2009.do.
\begin{table}[!htbp]
\centering
\begin{threeparttable}
\setlength{\tabcolsep}{3pt}
\def\sym#1{\ifmmode^{#1}\else\(^{#1}\)\fi}
\caption{Peer effects in golf scores}
\label{tab:gkn-score}
\begin{tabular}{@{}l*{6}{c}@{}}
\hline\hline
 & (1) & (2) & (3) & (4) & (5) & (6) \\
\hline
$\hat{\lambda}$&      0.0177\sym{*}  &      0.0187\sym{*}  &      0.0185\sym{*}  &      0.0178\sym{*}  &      0.0168         &      0.0179\sym{*}  \\
            &    (0.0104)         &    (0.0104)         &    (0.0104)         &    (0.0104)         &    (0.0104)         &    (0.0104)         \\
Handicap    &                     &      0.5091\sym{***}&                     &                     &                     &                     \\
            &                     &    (0.0452)         &                     &                     &                     &                     \\
Peer handicap&                     &     -0.0256         &                     &                     &                     &                     \\
            &                     &    (0.0508)         &                     &                     &                     &                     \\
Driving distance&                     &                     &     -0.0260\sym{***}&                     &                     &     -0.0180\sym{***}\\
            &                     &                     &    (0.0044)         &                     &                     &    (0.0044)         \\
Peer driving distance&                     &                     &      0.0051         &                     &                     &      0.0057         \\
            &                     &                     &    (0.0054)         &                     &                     &    (0.0055)         \\
Putts       &                     &                     &                     &      0.0946\sym{**} &                     &      0.1143\sym{***}\\
            &                     &                     &                     &    (0.0377)         &                     &    (0.0377)         \\
Peer putts  &                     &                     &                     &     -0.0377         &                     &     -0.0382         \\
            &                     &                     &                     &    (0.0477)         &                     &    (0.0479)         \\
Greens in regulation&                     &                     &                     &                     &     -0.5534\sym{***}&     -0.5260\sym{***}\\
            &                     &                     &                     &                     &    (0.0584)         &    (0.0592)         \\
Peer greens in regulation&                     &                     &                     &                     &      0.0015         &     -0.0049         \\
            &                     &                     &                     &                     &    (0.0728)         &    (0.0742)         \\
\hline
Observations&       8,254         &       8,254         &       8,113         &       8,113         &       8,113         &       8,113         \\
Urns        &         288         &         288         &         281         &         281         &         281         &         281         \\
Groups      &       3,010         &       3,010         &       2,964         &       2,964         &       2,964         &       2,964         \\
\noalign{\hrule height 1pt}
\end{tabular}
\begin{tablenotes}[flushleft]
\footnotesize
\item[] \textit{Notes:} All columns use first-round scores. Column (1) includes no covariates. Column (2) adds handicap. Columns (3)--(6) add the indicated own and peer controls. Standard errors are in parentheses. \sym{*} \(p<0.10\), \sym{**} \(p<0.05\), \sym{***} \(p<0.01\).
\end{tablenotes}
\end{threeparttable}
\end{table}

\end{document}